\documentclass[a4paper,10pt,reqno]{amsart}
          \usepackage{amssymb}
	  \usepackage{amsmath}
          \usepackage{amsfonts}
          \usepackage[english]{babel}
          \usepackage[utf8]{inputenc}

\usepackage[margin=3cm]{geometry}

\usepackage{color}
 \usepackage[unicode,colorlinks,plainpages=false,hyperindex=true,bookmarksnumbered=true,bookmarksopen=true,pdfpagelabels]{hyperref}
 \hypersetup{urlcolor=cyan,linkcolor=blue,citecolor=red,colorlinks=true} 
\makeatletter
\renewcommand*{\p@subsection}{\S\,}
\renewcommand*{\p@subsubsection}{\S\,}
\makeatother

\newtheorem{thm}{Theorem}[section]
\newtheorem{cor}[thm]{Corollary}
\newtheorem{lem}[thm]{Lemma}
\newtheorem{prop}[thm]{Proposition}

\newtheorem{rem}[thm]{Remark}

\theoremstyle{definition}

\numberwithin{equation}{section}

\newcommand{\CC}{\ensuremath{\mathbb{C}}}
\newcommand{\N}{\ensuremath{\mathbb{N}}}

\newcommand{\Z}{\ensuremath{\mathbb{Z}}}
\newcommand{\PP}{\operatorname{P}}
\newcommand{\tr}{\operatorname{tr}}

\newcommand{\diag}{\operatorname{diag}}

\newcommand{\Hom}{\operatorname{Hom}}
\newcommand{\Der}{\operatorname{Der}}
\newcommand{\Mat}{\operatorname{Mat}}
\newcommand{\Id}{\operatorname{Id}}
\newcommand{\id}{\operatorname{id}}
\newcommand{\Gl}{\operatorname{GL}}

\newcommand{\Rep}{\operatorname{Rep}}
\newcommand\dgal[1]{  \left\{\!\!\left\{#1\right\}\!\!\right\} }
\newcommand\dSN[1]{\left\{\!\!\left\{#1\right\}\!\!\right\}_{\operatorname{SN}}}
\newcommand\br[1]{\{ #1 \}} 
\newcommand\brSN[1]{\{ #1 \}_{\operatorname{SN}}} 
\newcommand\brap[1]{\{ #1 \}_{\PP}} 

\def\aalpha{{\widetilde{\alpha}}}
\def\qq{{\widetilde{q}}}

\def\h{\mathfrak{h}}
\def\Creg{\mathfrak{C}_{\mathrm{reg}}}

\def\VV{\mathcal V}
\def \Maq{\widetilde{\mathcal M}_{\alpha, q}(Q)}
\def\Cnm{\mathcal{C}_{n,\qq,d}(m)}
\def\Cnmo{\mathcal{C}_{n,\qq,d}^\circ(m)}
\def\Cnmp{\mathcal{C}_{n,\qq,d}'(m)}
\def\Cntd{\mathcal{C}_{n,t,d}} 

\def\CnmG{\mathcal{C}_{n,\qq,d}^{\mathcal H}(m)}
\def\Cnmdual{\mathcal{C}_{n,\hat{q},d}^{\mathcal H}(m)}

\newcommand{\aaa}{\ensuremath{\mathbf{a}}}

\newcommand{\ccc}{\ensuremath{\mathbf{c}}}

\newcommand{\As}{\ensuremath{\mathbf{A}}}
\newcommand{\Cs}{\ensuremath{\mathbf{C}}}
\newcommand{\Asm}{\ensuremath{\mathbf{A}^{(m)}}}
\newcommand{\Csm}{\ensuremath{\mathbf{C}^{(m)}}}
\newcommand{\QU}{\ensuremath{\mathcal{Q}_U}}

\begin{document}

\title[Spin versions of the RS model from cyclic quivers] 
{Spin versions of the complex trigonometric Ruijsenaars-Schneider model from cyclic quivers}


\author{Maxime Fairon}
 \address[Maxime Fairon]{School of Mathematics\\
         University of Leeds\\
         Leeds, LS2 9JT, UK\\}
 \email{mmmfai@leeds.ac.uk}


\begin{abstract}
 We study multiplicative quiver varieties associated to specific extensions of cyclic quivers with $m\geq 2$ vertices. Their global Poisson structure is characterised by quasi-Hamiltonian algebras related to these quivers, which were studied by Van den Bergh for an arbitrary quiver. We show that the spaces are generically isomorphic to the case $m=1$ corresponding to an extended Jordan quiver. This provides a set of local coordinates, which we use to interpret integrable systems as spin variants of the trigonometric Ruijsenaars-Schneider system. 
This generalises to new spin cases recent works on classical integrable systems in the Ruijsenaars-Schneider family. 
\end{abstract}

\maketitle

\section{Introduction}  \label{intro}

In this paper, we continue a recent attempt initiated in \cite{CF} to interpret the phase spaces of classical complex integrable systems in the Ruijsenaars-Schneider (or RS) family as moduli spaces constructed from particular quivers. Before focusing on this problem, let's recall the well understood interpretation in the non-relativistic  case of the Calogero-Moser (or CM) system, and  its spin variant. In the  pioneering work \cite{W}, Wilson unveils several structures related to the phase space for the complex CM system, one of which is the  hyperk\"{a}hler structure it possesses. The latter is naturally defined in the context of Nakajima quiver varieties \cite{Nak94}, which considers Hamiltonian reduction of representation spaces of quivers. Forgetting all but the topological structure, the phase space is nothing else than the reduced representation space of a deformed preprojective algebra associated to a Jordan quiver extended by one arrow. Therefore, it is natural to ask if one could obtain the symplectic structure also at the level of the algebra. This is indeed the case, if we consider non-commutative symplectic geometry \cite{G,CBEG}, or the analogue for non-commutative Poisson geometry \cite{VdB1}. We refer to the review \cite{T17} for some details. Going a step further, we can understand the spin generalisation of this model discovered by Gibbons and Hermsen \cite{GH}, by looking at a Jordan quiver consisting of a single loop-arrow, which we extend by several arrows coming from an additional vertex \cite{BP,T15}. We obtain in this way the model in type $A_n$, 
whose Weyl group $W=S_n$ determines the symmetry of the obtained system. A study of various extensions of cyclic quivers generalises the result to different complex reflection groups \cite{CS}, in particular the case $W=S_n \ltimes\Z^n_m$ to which we shall come back.

We postpone the quiver interpretation in the relativistic case, as we first need to focus on the geometric side of these systems. In a way similar to the CM case, it is possible to understand the trigonometric RS system geometrically, either as a symplectic leaf on a space defined by Poisson reduction \cite{FockRosly}, or  directly using quasi-Hamiltonian reduction \cite{Oblomkov,CF}. While the process of Hamiltonian reduction brings down a Poisson manifold to one of smaller dimension by considering  the action of a Lie group, in the quasi-Hamiltonian setting we begin with a space which has some failure to have a Poisson bracket, but we end up with a genuine Poisson manifold. These spaces that are called quasi-Poisson manifolds \cite{QuasiP}, which are first introduced in \cite{AMM} for the \textquoteleft quasi-symplectic\textquoteright~case, find their origin in the need to get Lie group valued moment maps. In some cases, it also provides a finite-dimensional framework for infinite-dimensional symplectic reductions introduced by Atiyah and Bott \cite{AB}. Therefore, the reduction described above provides an alternative formalism to understand  earlier works of Gorsky and Nekrasov \cite{GN,N}. 
However, if we leave the type $A_n$, we can observe that until recent works by Feh\'er and collaborators (see e.g. \cite{FG,FK09b,FK12,FKlu,FM} mostly in the real case), integrable systems in the trigonometric RS family are generally devised  using only a suitable Lax matrix, as they originally appear in \cite{RS86}, without  geometric perspectives.

The lack of a specific geometrical framework to derive these models is even more apparent for spin versions. To understand what is known at the moment, let us recall how the spin RS system is introduced in the first place. 
In a celebrated attempt to generalise the relation between the matrix KP equation and the spin CM system \cite{KBBT}, Krichever and Zabrodin investigate solutions of the non-abelian 2D Toda chain  and discover the Lax matrix for the real spin RS system, already in its elliptic form \cite{KrZ}. 
This system is parametrised by $n$ particles with positions $q_i$, each endowed with $d$ additional degrees of freedom $a_i^\alpha$, for which we have $d$ conjugate variables $c_i^\alpha$ that are function of the momentum $\dot{q}_i$. There are additional $n$ relations, so that we have $2nd$ independent coordinates, which appear in the Lax matrix $L$ (that we consider without spectral parameter) such that the Hamiltonian $H_1=\tr L$ defines the equations of motion for the spin RS system. 
A striking feature of this space is the existence of a natural action of a Lie group of dimension $d(d-1)$, such that on the corresponding orbit space  we can pick coordinates from the functions\footnote{Note that after this reduction, the spin variables $(f_{ij})$ are not attached to a specific particle any more, but they represent collective degrees of freedom.} $(q_i, f_{ij}=\sum_\alpha a_i^\alpha c_j^\alpha)$. Moreover,  the Hamiltonian $H_1$ descends to this reduced space where it becomes integrable in Liouville sense, and solutions to the equations of motion defined by $H_1$ can be found in terms of theta functions.   In the rational and trigonometric case, a simpler form for the equations of motions corresponding to the Hamiltonians $\tr L^k$  can be found \cite{AF,RaS}. However, except in the rational case \cite{AF} and for two particles in the elliptic case \cite{S09}, the Poisson structure of the space is only known in a universal form \cite{Kr} that is not easy to manipulate. It is the existence of a geometric formalism that allows to completely determine the Poisson brackets between the coordinates $(q_i,a_i^\alpha,c_i^\alpha)$ in the rational case for the type $A_n$ \cite{AF}. 
Similarly, it is the existence of a geometric interpretation that enables to prove the integrability of such system outside the type $A_n$ in the rational \cite{Re} or trigonometric \cite{Fe} cases. (Nevertheless, in those cases we work on the phase space where the individual spins $(a_i^\alpha, c_i^\alpha)$  are not naturally defined and we only know the collective spins $(f_{ij})$.) Thereupon, our understanding of the Hamiltonian formalism for spin RS systems  remains quite limited outside the rational case, but it is natural to expect to solve this issue if we find a  correct geometric framework.

\medskip
  
Now, let's come back to the algebraic interpretation with quivers. 
A key aspect of the work of Van den Bergh  \cite{VdB1} is that it also introduces non-commutative quasi-Poisson geometry. To an arbitrary quiver, considering its double, one can associate  a multiplicative preprojective algebra, and construct the corresponding multiplicative quiver varieties of Crawley-Boevey and Shaw \cite{CBShaw}. The latter spaces are Poisson varieties, obtained after quasi-Hamiltonian reduction from the representation spaces of the quiver path algebra, and all the geometric structure can be realised at the level of the path algebra  \cite{VdB1}. 
Therefore, it is a reasonable guess to investigate this theory applied to the quivers studied in \cite{CS}, and it is done by Chalykh and the author for a cyclic quiver with an extra arrow in \cite{CF}. In that case, it is shown that any multiplicative quiver variety contains the phase space for the (non-spin) trigonometric RS system.  For the simplest cyclic quiver consisting of one loop (i.e. a Jordan quiver) with $d\geq 2$ arrows coming from a new vertex, this is also done by Chalykh and the author, in the companion paper \cite{CF2}. The natural generalisation that these other multiplicative quiver varieties carry on an open subset the phase space for the spin trigonometric RS system of type $A_n$ is obtained. Henceforth it provides a crucial step to develop the geometric theory of spin RS models farther than the rational case.    
The next step is to turn to the application of this method on the cyclic quiver with $m \geq 2$ vertices and $d \geq 2$ new arrows pointing towards a chosen vertex in the cycle. This is the purpose of this work, and our most important result is that any such representation space can also be seen as the natural phase space of what we suggest to be the complex trigonometric spin RS system with $W=S_n \ltimes\Z^n_m$. This provides a natural generalisation of the case $m=1$ corresponding to the Jordan quiver  \cite{CF2}. Interestingly, there is a Poisson isomorphism between dense open subsets of the spaces corresponding to the cases $m\geq 2$ and $m=1$, hence we also obtain that the representation spaces that we construct carry the system with $W=S_n$. 
In particular, the flows defined by the symmetric functions of the corresponding Lax matrices can be explicitly integrated. 
We also study Liouville integrability in line with the original approach of Krichever and Zabrodin \cite{KrZ}, which  is based on the existence of a second reduction for $d\leq n$. The final step dealing with arbitrary extensions of cyclic quivers will be discussed in forthcoming works. 

\medskip

The paper is organised as follows. In Section \ref{prel}, we recall the foundations of the theory of double quasi-Poisson brackets, how we can define such brackets from quivers, and what is the counterpart to that theory on the  corresponding representation spaces, based on the work of Van den Bergh \cite{VdB1}. Our presentation of this work relies on \cite[Section 2]{CF}, but we reproduce these results and add useful remarks to provide a self-contained exposition of this non-standard subject. In Section \ref{ss:qHcyclic}, we apply the algebraic part of the theory to the so-called spin cyclic quivers. We obtain their structure of quasi-Hamiltonian algebras and gather several results based on computations with the double quasi-Poisson bracket. All the proofs for that section are collected in Appendix \ref{Ann:cyclic}. The formalism employed being quite new, we suggest to 
the reader interested in the integrable systems side of this work to skip the first part of this paper, and go directly to Section \ref{tadpole}, where we overview the multiplicative quiver varieties associated to the spin Jordan quiver (or spin one-loop quiver) recently introduced in \cite{CF2}. 
In Section \ref{cyclic}, we follow the method from \cite{CF2} applied to the cyclic quivers on $m\geq 2$ vertices, to get new multiplicative quiver varieties. They are generically isomorphic as complex Poisson manifolds to the space obtained in the case of a Jordan quiver reviewed in Section \ref{tadpole}, which we prove in Appendix \ref{Ann:iso}. We get three families of functions in involution on such a space for each $m \geq 2$, that we write in the set of local coordinates that exists in the Jordan quiver case. In particular, one of them contains the spin RS Hamiltonian $H_1$. We count the number of independent elements in each family and perform an additional reduction to obtain integrability in Liouville sense. We can also get an explicit description of some flows, based on computations in Appendix \ref{Ann:Dyn}.  We finish by explaining why we believe that this extra reduction should be avoided, and how the other two families corresponds to what should be seen as the spin RS system for $W=S_n \ltimes\Z^n_m$, or a modification of it. 

\medskip

{\bf Acknowledgement.} The author is grateful to O. Chalykh for suggesting the problem and for stimulating  conversations, as well as for his collaboration while working on \cite{CF2} which inspired the present paper. 
The author also thanks L. Feh\'er and V. Rubtsov for interesting discussions. 
Some of the results in this paper appear in the University of Leeds PhD thesis of
the author, supported by a University of Leeds 110 Anniversary Research Scholarship.

\subsection{Notations} \label{SNot}  The sets $\N,\Z,\CC$ denote the non-negative integers, integers and complex numbers. We write $\N^\times,\Z^\times,\CC^\times$ when we omit the zero element in those sets. 
Consider a finite set $J$, $|J|=k$, and totally ordered elements $(a_j)_{j\in J}$ such that $a_{j_1} <  \ldots < a_{j_k}$ for some $j_{(-)}:\{1,\ldots,k\}\to J$. Then  the corresponding right and left products are defined as $\prod \limits^{\longrightarrow} a_j=a_{j_1} \ldots a_{j_k}$, while $\prod \limits^{\longleftarrow} a_j=a_{j_k} \ldots a_{j_1}$. 

We write $\delta_{ij}$ or $\delta_{(i,j)}$ for Kronecker delta function. We extend this definition for a general proposition $P$ by setting $\delta_P=+1$ if $P$ is true and $\delta_P=0$ if $P$ is false. For example, $\delta_{(i\neq j)}=1-\delta_{ij}$.  

Fix  a positive integer $d\geq 2$. 
The \emph{ordering function} $o:\{1,\ldots,d\}^{\times 2}\to \{0,\pm1\}$ is defined by 
$o(\alpha,\beta)=+1$ if $\alpha<\beta$, $o(\alpha,\beta)=-1$ if $\alpha>\beta$, and 
$o(\alpha,\beta)=0$ if $\alpha=\beta$. In other words, it takes the value $+1$ if the first argument is strictly less than the second, the value $0$ if they are equal, and is $-1$ otherwise. In terms of Kronecker delta function, we have for example 
$o(\alpha,\beta)=\delta_{(\alpha< \beta)}-\delta_{(\alpha > \beta)}$.

\section{Preliminaries}  \label{prel}

We recall the necessary constructions needed in this paper as they are introduced in \cite{CF}, with some additional remarks. Details regarding double brackets and representation spaces can be found in  \cite{VdB1, VdB2}, while we refer to  \cite{CBShaw,Y} for generalities on multiplicative preprojective algebras and corresponding multiplicative quiver varieties.

\subsection{Double brackets} \label{ss:dAS} We review some results of \cite[Sections 2-4]{VdB1}. We take all tensor products over $\CC$, and  fix an associative unital $\CC$-algebra $A$. For an element $a\in {A}\otimes {A}$, we use Sweedler's notation $a'\otimes a''$ to denote $\sum_i a_i'\otimes a_i''$. We set $a^\circ=a''\otimes a'$. More generally,  
for any $s\in S_n$ we define $\tau_s:\,A^{\otimes n}\to A^{\otimes n}$ by 
$\tau_s(a_1\otimes \ldots \otimes a_n)=a_{s^{-1}(1)}\otimes \ldots \otimes a_{s^{-1}(n)}$, so we can write $a^\circ=\tau_{(12)}a$. 

We view $A^{\otimes n}$ as an $A$-bimodule via the \emph{outer} bimodule structure 
$b(a_1\otimes \ldots \otimes a_n)c=ba_1\otimes \ldots \otimes a_nc$. 
An \emph{$n$-bracket} is a linear map 
$\dgal{-,\ldots,-} : A^{\otimes n}\to A^{\otimes n}$ which is a derivation in its last 
argument for the outer bimodule structure on $A^{\otimes n}$, 
and which is cyclically anti-symmetric: 
\begin{equation*}
\tau_{(1\ldots n)}\circ \dgal{-,\ldots,-}\circ \tau^{-1}_{(1\ldots n)}
=(-1)^{n+1}\dgal{-,\ldots,-}\,. 
\end{equation*} 
In the cases of interest, there exists a $\CC$-algebra $B$ and a $\CC$-algebra map $B\to A$ turning $A$ into a $B$-algebra, and we identify $B$ with its image in $A$. Then  we assume that the bracket is $B$-linear, i.e. it vanishes if one argument is an element of $B$.

We focus on $2$- and  $3$-brackets, which we call \emph{double} and  \emph{triple brackets} respectively. In the particular case of a double bracket, the defining relations take the form $\dgal{a,b}=-\dgal{b,a}^\circ$ and 
$\dgal{a,bc}=b\dgal{a,c}+\dgal{a,b}c$ for any $a,b,c\in A$. This implies that $\dgal{bc,a}=\dgal{b,a}\ast c+b \ast \dgal{c,a}$, i.e. it is a derivation in the first argument for the  \emph{inner} $A$-bimodule structure on $A \otimes A$ given by $b \ast (a' \otimes a'') \ast c= a' c \otimes b a''$. 
 Also, any double bracket $\dgal{-,-}$  defines an induced triple bracket $\dgal{-,-,-}$ given by  
\begin{equation}
 \label{Eq:TripBr}
\begin{aligned}
   \dgal{a,b,c}=&\dgal{a,\dgal{b,c}'}\otimes \dgal{b,c}''+\tau_{(123)}\dgal{b,\dgal{c,a}'}\otimes \dgal{c,a}''
+\tau_{(132)}\dgal{c,\dgal{a,b}'}\otimes \dgal{a,b}'' \,.
\end{aligned}
\end{equation}

In a similar way to the commutative case, $n$-brackets can be defined from analogues of $n$-vector fields. 
Following Crawley-Boevey \cite{CB}, we assume from now on that $A$ is a $B$-algebra and  we call the elements of $D_{A/B}:=\Der_B(A,A\otimes A)$ \emph{double derivations}. We see $D_{A/B}$ as an $A$-bimodule by using the inner bimodule structure on $A\otimes A$. That is, if $\delta\in D_{A/B}$, then $(b\,\delta\, c) (a)=b \ast \delta(a)\ast c$ for any $a,b,c\in A$. Let $D_BA:=T_AD_{A/B}$ be the tensor algebra of this bimodule. 
\begin{prop}   \emph{(\cite[Proposition 4.1.1]{VdB1})}
\label{Prop:BrQ}
There is a well-defined linear map $\mu:(D_BA)_n\to \{B$-linear $n$-brackets on $A\}$, 
$Q\mapsto \dgal{-,\ldots,-}_Q$ which on $Q=\delta_1 \ldots \delta_n$ is given by 
\begin{equation*}
  \begin{aligned}
\label{Eq:BrQ}
\dgal{-,\ldots,-}_Q&=\sum_{i=0}^{n-1}(-1)^{(n-1)i}\tau^i_{(1\ldots n)}\circ\dgal{-,\ldots,-}_Q^{\widetilde{ }}
\circ\tau^{-i}_{(1\ldots n)}\,\,, \\
\dgal{a_1,\ldots,a_n}_Q^{\widetilde{ }}&=\delta_n(a_n)'\delta_1(a_1)''\otimes \delta_1(a_1)'\delta_2(a_2)''\otimes \ldots \otimes 
\delta_{n-1}(a_{n-1})'\delta_n(a_n)''\,.
 \end{aligned}
\end{equation*}
The map $\mu$ factors through $D_BA/[D_BA,D_BA]$ (for the graded commutator).
\end{prop}
In the particular case of $\delta_1\delta_2\in (D_BA)_2$ we have for any $b,c\in A$
\begin{equation} \label{Eq:BrQ2}
\dgal{b,c}_{\delta_1\delta_2}=
 \delta_2(c)'\delta_1(b)'' \otimes \delta_1(b)'\delta_2(c)''
-  \delta_1(c)'\delta_2(b)'' \otimes \delta_2(b)'\delta_1(c)''\,.
\end{equation} 
Note also that $D_BA$ admits a canonical \emph{double Schouten--Nijenhuis bracket}, which makes $D_BA$ into a double Gerstenhaber algebra \cite[\S 2.7,3.2]{VdB1}. This is a (graded) double bracket, that we denote by $\dSN{-,-}$. 

\medskip

For any $n \geq2$, the multiplication map $m:A^{\otimes n} \to A$ is defined by  concatenation of the factors, 
$m(a_1\otimes \ldots \otimes a_n)=a_1\ldots a_n$. 
For any $n$-bracket $\dgal{-,\ldots,-}$, this induces an associated bracket $\{-,\ldots,-\}:=m \circ \dgal{-,\ldots,-}$. In the case of a double bracket, 
\begin{equation}\label{sbra}
\{a,b\}=m\circ \dgal{a, b}=\dgal{a,b}'\dgal{a,b}''\,.  
\end{equation}
Assume that the double bracket $\dgal{-,-}$ is such that the bracket associated to the induced triple bracket \eqref{Eq:TripBr} satisfies $\br{-,-,-}=0$. Then the bracket $\br{-,-}$ associated to the double bracket is a \emph{left Loday bracket} (also called left Leibniz bracket), and it also satisfies Leibniz's rule in its second argument, i.e. $\br{-,-}:A\times A \to A$ is a bilinear map such that 
\begin{equation*}
  \br{a,\br{b,c}}=\br{\br{a,b},c}+\br{b,\br{a,c}}\,, \quad \br{a,bc}=\br{a,b}c + b \br{a,c}\,.
\end{equation*}
It descends to a map $A/[A,A]\times A \to A$, then to an antisymmetric map on the vector space $A/[A,A]$, so that $(A/[A,A],\br{-,-})$ is a Lie algebra \cite[\S 2.4]{VdB1}. Since each map $\br{a,-}$ is a derivation on $A$, thus $\br{-,-}$ endows $A$ with a non-commutative Poisson structure in the sense of \cite{CB2}, called an $H_0$-Poisson structure. 

 We can extend the procedure to graded setting, and denote the  bracket associated to $\dSN{-,-}$ as  $\brSN{-,-}:=m\circ\dSN{-,-}$.  Hence, the pair $(D_BA,\brSN{-,-})$ may be viewed as a non-commutative version of the algebra of polyvector fields on a manifold with the Schouten-Nijenhuis bracket.

\subsection{Double quasi-Poisson algebras} \label{ss:doubleqP}
Consider $B$ of the form $B={\CC} e_1 \oplus \ldots \oplus {\CC} e_K$ with $e_re_s=\delta_{rs}e_s$, so that $\sum_s e_s =1\in A$. We define for all $s$ a double derivation $E_s\in D_{A/B}$ such that 
$ E_s(a)=ae_s\otimes e_s - e_s\otimes e_s a$, and we call them the \emph{gauge elements}. 
As in \cite[Section 5]{VdB1}, a \emph{double quasi-Poisson bracket} on $A$ is a ($B$-linear) double bracket $\dgal{-,-}$, 
such that the induced triple bracket satisfies $\dgal{-,-,-}=\frac{1}{12}\sum_s \dgal{-,-,-}_{E_s^3}$, where the triple  brackets in the right-hand side are defined in Proposition \ref{Prop:BrQ}. 
In this case, we say that $A$ is a \emph{double quasi-Poisson algebra}. Note that the associated brackets  
$\br{-,-,-}_{E_s^3}$ are identically zero, so that the double quasi-Poisson bracket $\dgal{-,-}$ on $A$ defines a
left Loday bracket $\br{-,-}$  by \eqref{sbra}, which descends to a Lie bracket on $A/[A,A]$. 
Assume that there is an element $P\in (D_BA)_2$ such that 
$\{P,P\}_{\operatorname{SN}}=\frac{1}{6}\sum_sE_s^3$ mod $[D_BA,D_BA]$ (for the graded commutator). Then 
we say that $A$ is a \emph{differential double quasi-Poisson algebra} with the 
\emph{differential double quasi-Poisson bracket} $\dgal{-,-}_P$. This implies that $\dgal{-,-}_P$ is a double quasi-Poisson bracket using \cite[Theorem 4.2.3]{VdB1}.

A \emph{multiplicative moment map} for a double quasi-Poisson algebra $(A,\dgal{-,-})$ is an element 
$\Phi=\sum_{s=1}^K \Phi_s$ with $\Phi_s\in e_sAe_s$ such that we have   
$\dgal{\Phi_s,-}=\frac{1}{2}(\Phi_sE_s+E_s\Phi_s)\in D_{A/B}$ for any $s$. This condition may be written explicitly as requiring for all $a\in A$ 
\begin{equation} \label{Phim}
 \dgal{\Phi_s,a}=\frac{1}{2}(\Phi_sE_s+E_s\Phi_s)(a)
=\frac12 (ae_s\otimes \Phi_s-e_s \otimes \Phi_s a +  a \Phi_s \otimes e_s-\Phi_s \otimes e_s a)\,.
\end{equation}
When a double quasi-Poisson algebra 
is equipped with a multiplicative moment map, we say that it is a \emph{quasi-Hamiltonian algebra}. 

Combining \eqref{sbra} and \eqref{Phim}, we obtain $\br{\Phi,a}=a\Phi-\Phi a$ and $\br{a,\Phi}=0$ for any $a\in A$. Hence, if $q_0\in \CC$ and $\br{-,-}$ is the left Loday bracket obtained from the double bracket on $A$, we get that $\br{J_0,A}\subset J_0$, $\br{A,J_0}\subset J_0$ for $J_0$ the ideal generated by $\Phi-q_0$. Therefore $A/J_0$ is  a left Loday algebra. If we consider $q=\sum_s q_se_s \in B$ and write $J$ for the ideal generated by $\Phi-q$, we only have $\br{A,J}\subset J$ in general, so that $A_q:=A/J$ is not necessarily a left Loday algebra. Nevertheless, since $\br{J,A}\subset J$ modulo commutators, the vector space   $A_q/[A_q,A_q]$ is a Lie algebra for the Lie bracket obtained from $\br{-,-}$ through $A\to  A_q/[A_q,A_q]$. 
This endows $A_q$ with an $H_0$-Poisson structure \cite[Proposition 5.1.5]{VdB1}. 

Finally, assume that $A$ is (formally) smooth. If $A$ is a double quasi-Poisson algebra with double bracket defined by $P\in (D_BA)_2$, we say that the element $P$ is \emph{non-degenerate} if the map of $A$-bimodules $\Omega_B^1A \oplus (\oplus_s A E_s A) \to D_{A/B}$ given by $(b .da.c,\delta) \mapsto b \dgal{a,-}_P c+ \delta$ is surjective. Here, $\dgal{-,-}_P$ is the double bracket defined by Proposition \ref{Prop:BrQ}, and $\Omega_B^1A$ refer to the bimodule of non-commutative relative $1$-forms \cite[Section 2]{CQ95}. We refer to the brilliant work of Van den Bergh \cite{VdB2} for details and the relation to a \textquoteleft double version\textquoteright~of  \cite{AMM}.

\subsection{Multiplicative preprojective algebras} \label{ss:MultqPrep}

Let $Q=(Q, I)$ be a quiver with vertex set $I$ and arrow set $Q$, and consider the maps $t,h:Q\to I$ that associate to every arrow $a$ its tail and head, $t(a)$ and $h(a)$. We construct the double $\bar{Q}$ of $Q$ by adjoining to every $a\in Q$ an opposite arrow, denoted $a^{*}$. We naturally extend $t$ and $h$, so that $t(a)=h(a^*)$ and $h(a)=t(a^*)$.  
We define $\epsilon : \bar{Q}\to \{\pm 1\}$ the function that takes value $+1$ on arrows of $Q$, and $-1$ on each arrow of $\bar{Q} \setminus Q$. We write $\CC\bar{Q}$ for the path algebra of $\bar{Q}$, whose underlying vector space is spanned by all possible paths formed on $\bar{Q}$ (including each trivial path $e_s$ associated to $s\in I$). The multiplication is given by concatenation of paths, and in particular the $(e_s)_s$ form a complete set of orthogonal idempotents. We view $\CC\bar{Q}$ as a $B$-algebra, with $B=\oplus_{s\in I} \CC e_s$. Finally, we extend $*$ to an involution on $\CC\bar{Q}$ by setting $(a^*)^*=a$ for all $a\in Q$.

\begin{rem}\label{rem2.1}
As in \cite{CF,VdB1}, we write paths in  $\CC\bar{Q}$ from left to right. Hence, $ab$ means 
\textquoteleft$a$ followed by $b$\textquoteright, and the path $ab$ is trivially zero  if $h(a)\ne t(b)$. 
\end{rem}
Let $A$ be obtained from $\CC \bar{Q}$ by inverting all elements $(1+aa^*)_{a\in \bar{Q}}$.  For all $a\in \bar{Q}$, 
define the element $\frac{\partial}{\partial a}$ 
of $D_BA$ which on $b\in \bar{Q}$ acts as 
\begin{equation}
\label{Eq:Dba}
 \frac{\partial b}{\partial a}=\left\{\begin{array}{ll} e_{t(a)}\otimes e_{h(a)}&\text{if }a=b \\ 
0&\text{otherwise} \end{array}\right. 
\end{equation}
We consider a minor generalisation of the construction of multiplicative preprojective algebra, without the use of a total ordering on the arrows of the quiver, see the first remark at the end of \cite[\S 2.5]{CF}. 
For each $s\in I$, we fix a total ordering $<_s$ on the arrows meeting at $s$, that is on all $a\in \bar{Q}$ with $h(a)=s$ or $t(a)=s$. We also assume that if two arrows $a,b$ meet at $s$ and $r$, then either we have both $a <_s b$ and $a<_{r} b$ or we have both  $b <_s a$ and $b<_{r} a$. We denote such a relation by $<$ and refer to is as an \emph{ordering}, though it is not necessarily a partial order\footnote{Strictly speaking, what we only use to get Theorem \ref{Thm:QStruct} is the total ordering $<_s$ for each $s$ on the arrows $a\in \bar{Q}$ with $t(a)=s$. However, defining an ordering as we do is easier to write the assumption in Proposition \ref{Pr:dbr}.}. We define the element $\Phi=(\Phi_s)_s\in \oplus_se_sAe_s $ by 
\begin{equation}\label{pphi}
\Phi_s=\prod^{\longrightarrow}_{\substack{a\in \bar{Q}\\t(a)=s}} e_s(1+aa^*)^{\epsilon(a)}e_s\,,
\end{equation}
where the product is taken with respect to the ordering, see \ref{SNot}. Following \cite{CBShaw}, given $q=\sum_{s\in I} q_se_s$ with $q_s\in \CC^\times$, we define the (deformed) \emph{multiplicative preprojective algebra} as the quotient $\Lambda^q=A/(\Phi-q)$. Up to isomorphism, the algebra $\Lambda^q$ is independent of the ordering \cite[Theorem 1.4]{CBShaw}. 

The next result states that $\Phi$ is a moment map for a specific quasi-Hamiltonian algebra structure on $A$.  
In particular, we have an $H_0$-Poisson structure on  $\Lambda^q$ by \ref{ss:doubleqP}.  

\begin{thm} \emph{(\cite[Theorem 6.7.1]{VdB1})} 
\label{Thm:QStruct}
 The algebra $A$ is quasi-Hamiltonian for the differential double quasi-Poisson bracket defined by 
\begin{equation}\label{Eq:PP}
 \PP=\frac{1}{2}\sum_{a\in \bar{Q}} \epsilon(a) (1+a^*a)\frac{\partial}{\partial a} 
\frac{\partial}{\partial a^*} \, - \,\frac12 \sum_{\substack{a,b\in \bar{Q}\\ t(a)=t(b),\, a<b}} \left(\frac{\partial}{\partial a^*}a^*-a
\frac{\partial}{\partial a}\right)\left(\frac{\partial}{\partial b^*}b^*-b\frac{\partial}{\partial b}\right) 
\,,
\end{equation}
and the multiplicative moment map given by $\Phi=(\Phi_s)_s$, where $\Phi_s$ is defined in \eqref{pphi}. 
\end{thm}
In fact, $\PP$ is non-degenerate by \cite[Section 8]{VdB2}. 
The following result gives an explicit form to the double quasi-Poisson bracket which we denote as  $\dgal{-,-}$, and that is defined by $\PP$ using Proposition \ref{Prop:BrQ}. 
\begin{prop}\cite[Proposition 2.6]{CF} \label{Pr:dbr} Take an ordering in $\bar{Q}$  so that the arrows of $\bar{Q}$ are ordered in such a way that $a<a^*<b<b^*$ for any $a,b\in Q$ with $a < b$. 
Then one has 
 \begin{subequations}
       \begin{align}
\dgal{a,a}\,=\,&\frac{1}{2}\epsilon(a)\left( a^2\otimes e_{t(a)}- e_{h(a)}\otimes a^2 \right)\delta_{h(a),t(a)}\quad  \qquad (a\in\bar{Q})\,, \label{loop}\\
\dgal{a,a^*}\,=\,&e_{h(a)}\otimes e_{t(a)}
+\frac{1}{2} a^*a\otimes e_{t(a)} 
+\frac{1}{2} e_{h(a)}\otimes aa^* \nonumber\\
&+\frac{1}{2} (a^*\otimes a-a\otimes a^*)\delta_{h(a),t(a)}\qquad \qquad \qquad (a\in Q)\,, \label{aast}\\
\dgal{a,b}\,=\,&\frac{1}{2}(e_{h(a)}\otimes ab) \delta_{h(a),t(b)} 
+\frac{1}{2} (ba\otimes e_{t(a)}) \delta_{h(b),t(a)} \nonumber\\\label{a<b}
&-\frac{1}{2}(b\otimes a) \delta_{h(a),h(b)}-\frac{1}{2}(a\otimes b) \delta_{t(a),t(b)}\qquad (a,b\in\bar{Q}\,, \ a<b\,,\ b\ne a^*)\,.
\end{align}
  \end{subequations}
\end{prop} 
\noindent This defines all double brackets since when $a>b$, $\dgal{a,b}=-\dgal{b,a}^\circ$. 

\medskip

We finish by a remark on the structure of the moment map of a subquiver. assume that $\bar{Q}'$ is a quiver with vertex set  $I'\subset I$ and $\bar{Q}'=\{a\in \bar{Q}\,|\,t(a)\in I'\text{ and }h(a)\in I'\}$. This means that if we look at the subset of vertices $I'$ of $\bar{Q}$ and erase all the arrows of $\bar{Q}$ 
which are not both starting and ending at an element of $I'$, we get $\bar{Q}'$. Moreover, we 
require that $\bar{Q}$ and $\bar{Q}'$ are endowed respectively with  orderings $<,<'$ such that, whenever $a,b \in \bar{Q}'$, $a<'b$ if $a<b$ in the initial quiver $\bar{Q}$, 
and whenever $a\in \bar{Q}'$ but $c\in \bar{Q} \smallsetminus \bar{Q}'$ we have  $a<c$. 

We construct $A'$ as $A$ above, and we see $A'$ as a subalgebra of $A$ (after adding the removed idempotents $e_s$ for $s\in I \setminus I'$). Define elements $\Phi'$, $\PP'$ by replacing $\bar{Q}$ with $\bar{Q}'$ in \eqref{pphi} and \eqref{Eq:PP}. Remark that we can write $\PP=\PP'+\PP_{out}$ and $\Phi=(\Phi+\sum_{s \notin I'}e_s) \Phi_{out}$  for some $\PP_{out}\in (D_BA)_2$ and $\Phi_{out}=(\Phi_{out,s})_{s\in I}$. Last statement is, in fact, a consequence of the fusion process which is used to endow a quiver with a quasi-Hamiltonian structure \cite[\S 6.5-6.7]{VdB1}.

\begin{lem} \label{Lem1}
 For all $b,c\in A'\subset A$, we have $\dgal{b,c}_{\PP}=\dgal{b,c}_{\PP'}$. 
In particular, for all $s\in I'$, we have $\dgal{\Phi'_s,c}_{\PP}=\frac{1}{2}(\Phi'_sE_s+E_s\Phi'_s)(c)$. 
\end{lem}
\begin{proof}
 By linearity of the map in Proposition \ref{Prop:BrQ}, we can write 
$\dgal{-,-}_{\PP'}=\dgal{-,-}_{\PP}+\dgal{-,-}_{\PP_{out}}$. From \eqref{Eq:BrQ2}, we get that 
$\dgal{b,c}_{\PP_{out}}$ is a sum of terms of the form 
\begin{equation} \label{Eq:Lem1}
 \delta_2(c)'\delta_1(b)'' \otimes \delta_1(b)'\delta_2(c)''
-\delta_1(c)'\delta_2(b)'' \otimes \delta_2(b)'\delta_1(c)'' 
\end{equation} 
for any $b,c\in A$. 
By construction, $\PP_{out}$ is a sum of 
(double) biderivations, and each biderivation carries at least one factor $\partial/\partial d$ 
for $d\in \bar{Q}\smallsetminus \bar{Q}'$. Therefore,  if both $b,c \in A'$, 
all terms in \eqref{Eq:Lem1} must vanish, and $\dgal{b,c}_{\PP_{out}}=0$. 

Applying this to $\Phi_s'$ and $c\in A'$,  
$\dgal{\Phi_s',c}_{\PP}=\dgal{\Phi_s',c}_{\PP'}$. By construction, $\Phi'$ is a multiplicative moment map 
for $\dgal{-,-}_{\PP'}$, so it satisfies \eqref{Phim}.  
\end{proof}

\subsection{Geometric counterpart to the definitions} \label{ss:RepSp}
Fix a $\CC$-algebra $A$ and  $N\in\N$. The \emph{representation space} $\Rep(A, N)$ is the affine scheme whose coordinate ring $\mathcal{O}(\Rep(A, N))$ is generated by symbols $a_{ij}$ for $a\in A$ and $i,j=1,\ldots,N$, such that they are $\CC$-linear in $a$, they satisfy $(ab)_{ij}=\sum_k a_{ik}b_{kj}$ for any $a,b\in A$, and $1_{ij}=\delta_{ij}$. Alternatively, we can see $\Rep(A,N)$ as parametrising  algebra homomorphisms $\varrho: A \to \Mat_N(\CC)$, and we get  $a_{ij}(\varrho)=\varrho(a)_{ij}$ at any point $\varrho\in \Rep(A, N)$. 
Following \cite[Section 7]{VdB1}, to any $a\in A$ we associate a matrix-valued function $\mathcal{X}(a):=(a_{ij})_{ij}$ on $\Rep(A,N)$. Similarly, any double derivation $\delta\in\Der (A,A\otimes A)$ gives rise to a matrix-valued vector field $\mathcal{X}(\delta)=(\delta_{ij})_{ij}$ on $\Rep(A, N)$, where $\delta_{ij}$ is a derivation of $\mathcal{O}(\Rep(A, N))$ defined by the rule $\delta_{ij}(a_{uv})=\delta'(a)_{uj} \delta''(a)_{iv}$. In particular, if $\dgal{-,-}$ is a double bracket on $A$, we have for any $a\in A$ that the double derivation $\dgal{a,-}$ defines a matrix-valued vector field $X_a$ such that 
$(X_a)_{ij}(b_{uv})=\dgal{a,b}'_{uj}\dgal{a,b}''_{iv}$.  

\medskip

We can generalise the definition in a relative setting for a $B$-algebra $A$, where $B$ is of the form $B= {\CC} e_1 \oplus \ldots \oplus \CC e_K$ with $e_r e_s=\delta_{rs} e_s$. Representation spaces are now indexed by $K$-tuples $\alpha=(\alpha_1, \ldots, \alpha_K)\in \N^K$. Given $\alpha$ with $\alpha_1+\ldots +\alpha_K=N$, we embed $B$ diagonally into $\Mat_N(\CC)$ so that $\Id_N$ is split into a sum of $K$ diagonal blocks of respective sizes  $\alpha_1, \ldots, \alpha_K$, representing the idempotents $(e_s)_s$. By definition, $\Rep_B(A, \alpha)=\Hom_B(A,\Mat_{N}(\CC))$, and it can be viewed as an affine scheme in the same way as $\Rep(A, N)$. Note in particular that for any $\Phi\in\oplus_s \, e_s A e_s$, the matrix-valued function $\mathcal{X}(\Phi)$ on $\Rep_B(A,\alpha)$ is a block matrix $\mathcal{X}(\Phi) \in \prod_s\Mat_{\alpha_s}({\CC})$.

Assume that $A$ is equipped with a $B$-linear double bracket $\dgal{-,-}$. Then the representation spaces are endowed with an anti-symmetric biderivation as follows.
\begin{prop} \label{Prop:BrInduced} \emph{(\cite[Proposition 7.5.1, \S 7.8]{VdB1})} 
There is a unique anti-symmetric biderivation 
$\br{-,-}:\mathcal{O}(\Rep_B(A,\alpha))\times \mathcal{O}(\Rep_B(A,\alpha)) \to \mathcal{O}(\Rep_B(A,\alpha))$
such that for all $a,b\in A$, 
\begin{equation}\label{derr}
 \br{a_{ij},b_{uv}}=\dgal{a,b}'_{uj}\, \dgal{a,b}''_{iv}\,\, . 
\end{equation}
Moreover, if $\dgal{-,-}=\dgal{-,-}_P$ for some $P\in (D_B A)_2$, then $\br{-,-}$ is defined by the bivector field  $\tr(\mathcal{X}(P))$ and we denote it by $\br{-,-}_P$. 
\end{prop} 
The identity \ref{derr} extends to relate the double bracket $\dSN{-,-}$ on $D_BA$ with the Schouten-Nijenhuis bracket $[-,-]$ between polyvector fields on $\Rep_B(A,\alpha)$ \cite[Proposition 7.6.1]{VdB1}. 

On $\Rep(A, N)$ we have a natural action of $\Gl_N$, induced by conjugation action on $\Mat_N(\CC)$. Similarly, we have an action of  $\Gl_\alpha=\prod_s \Gl_{\alpha_s}$ on $\Rep_B(A, \alpha)$. Provided that $A$ is quasi-Hamiltonian,   $\Rep_B(A, \alpha)$ is a quasi-Hamiltonian manifold \cite{QuasiP}, as defined now in the smooth case (see \cite[\S 7.11-7.13]{VdB1} for the algebraic case). 

Let  $G$ be a Lie group with Lie algebra $\mathfrak{g}$. Moreover, assume that $\mathfrak{g}$ admits a non-degenerate $G$-invariant bilinear form $(-,-)$. If $(e_a),(e^a)$ are dual bases of $\mathfrak{g}$  with respect to $(-,-)$, we define the Cartan $3$-tensor  $\phi=\frac{1}{12}C^{abc}e_a\wedge e_b \wedge e_c$, for 
$C^{abc}=(e^a,[e^b,e^c])$ the tensor of structure constants. For all $\xi \in \mathfrak{g}$, write $\xi^L$ and $\xi^R$  to denote the left and right invariant vector fields on $G$ respectively. 

Given a $G$-manifold $M$, the $G$-action gives rise to a Lie algebra homomorphism 
$(-)_M:\mathfrak{g}\to \Der \mathcal O(M)$. This can be extended to polyvector fields and we can define the $3$-tensor $\phi_M$.  
We say that $M$ is a quasi-Poisson manifold if there exists an invariant bivector field $P$ on $M$ such that  $[P,P]=\phi_M$ under the Schouten-Nijenhuis bracket. We can use $P$ to define a bracket $\br{-,-}$ on $\mathcal O(M)$ in the obvious way.

A \emph{multiplicative moment map} is an $\operatorname{Ad}$-equivariant 
map $\Phi:M\to G$  satisfying 
\begin{equation} \label{qmom}
\{g\circ \Phi, -\}=\frac{1}{2}\left((e_a^L+e_a^R)(g)\circ \Phi\right)\, (e^a)_M\,, 
\end{equation}
for all functions $g\in\mathcal O(G)$, and we say that the triple $(M,P,\Phi)$ is a  \emph{Hamiltonian quasi-Poisson manifold}. In the case where the action of $G$ on $M$ is free and proper, for each conjugacy class $\mathcal{C}_g$ 
of $g\in G$ we can form the Poisson manifold $\Phi^{-1}(\mathcal{C}_g)/G$. This process is called 
\emph{quasi-Hamiltonian reduction}.

\begin{thm}
\label{Thm:Poiss}  \cite[\S 7.8, 7.13]{VdB1}
Assume that $(A,P)$ is a differential double quasi-Poisson algebra, which is quasi-Hamiltonian for the 
multiplicative moment map $\Phi\in \oplus_se_sAe_s$. We have that $\Rep_B(A, \alpha)$ is a $\Gl_\alpha$-space with a quasi-Poisson bracket $\br{-,-}_P$ determined from $\dgal{-,-}_P$ by \eqref{derr}.
Then the matrix-valued function $\mathcal{X}(\Phi):\Rep_B(A,\alpha)\to\prod_s\Mat_{\alpha_s}(\CC)$ is a  multiplicative moment map for $\Rep_B(A,\alpha)$. Therefore, if it smooth, $\Rep_B(A, \alpha)$  is  a Hamiltonian quasi-Poisson manifold.  
\end{thm}
We can note that any multiple of the identity $\Id_\alpha=\prod_s \Id_{\alpha_s}$ acts trivially on $\Rep_B(A,\alpha)$. This leads us to define 
\begin{align}\label{galpha}
G(\alpha)=\bigg(\prod_{s\in I} \Gl_{\alpha_s}\bigg)/\,{\CC}^\times\,,
\end{align}
where   $\CC^\times \subset \Gl_{\alpha}$ is the subgroup $\{\lambda \Id_\alpha \mid \lambda \in \CC^\times\}$. 
Combining Theorem \ref{Thm:Poiss} with \cite[Proposition 5.2]{VdB2} and \cite[Theorem 10.3]{QuasiP}, we obtain 

\begin{cor}
  \label{Cor:Sympl}  
Assume that $(A,P,\Phi)$ is a quasi-Hamiltonian algebra, with  $P$ non-degenerate. 
If $\mathcal{C}_g \subset \Gl_\alpha$ is a conjugacy class such that $Y:=\mathcal{X}(\Phi)^{-1}(\mathcal{C}_g)$ is smooth, and if the action of $G(\alpha)$  on $Y$ is free and the affine GIT quotient $Y/\! / G$ is a geometric quotient, then it is a Poisson manifold  with non-degenerate Poisson bracket defined by $\tr(\mathcal{X}(\PP))$, that we denote  $Y / G$. 
\end{cor}
Note that in the case where the conjugacy class is given by $\prod_s q_s \Id_{\alpha_s}$ with all $q_s\in\CC^\times$, we know that $Y/\! / G$ has a Poisson bracket because $A/(\Phi-\sum_s q_s e_s)$ has an $H_0$-Poisson structure by \cite[Section 4]{CB2}. However, we need the quasi-Poisson formalism to conclude that the Poisson bracket is non-degenerate in Corollary \ref{Cor:Sympl}.

Now, fix a conjugacy class $\mathcal{C}_g$ in $\operatorname{Lie}(\Gl_\alpha)$ and  assume that $F,G \in \mathcal O (\mathcal{X}(\Phi)^{-1} (\mathcal{C}_g))$ are invariant under the $\Gl_\alpha$ action. We can write $F=\tr (\mathcal{X}(a))$ and $G=\tr(\mathcal{X}(b))$ for some $a,b\in A$. Assuming that all spaces involved are smooth, we get from Proposition \ref{Prop:BrInduced}, Theorem \ref{Thm:Poiss} and \eqref{sbra} that 
\begin{equation} \label{relInv}
  \br{F,G}_P=\tr \mathcal X\left(\dgal{a,b}'\dgal{a,b}'' \right) =\tr \mathcal{X}(\br{a,b}),
\end{equation}
 where the bracket on the right-hand side is the associated bracket  $\br{-,-}=m \circ \dgal{-,-}$. 
In particular, we only need to compute $\br{a,b}$ modulo commutators in $A/[A,A]$ to get the Poisson bracket between $\tr(\mathcal{X}(a))$ and $\tr(\mathcal{X}(b))$. In slightly more general setting, given arbitrary $a,b\in A$ we find in the same way 
\begin{equation} \label{relBrDyn}
  \br{\tr(\mathcal{X}(a)),\mathcal{X}(b)}_P=\mathcal{X}(\br{a,b})\,,
\end{equation}
where this time we have the associated bracket  $\br{-,-}: A/[A,A] \times A \to A$.

\subsection{Multiplicative quiver varieties}\label{Sec:mqv}
From now on, fix $B=\oplus_{s \in I} \CC e_s$ as in \ref{ss:MultqPrep}. We always work in a relative setting and omit the subscript $B$ from the notation. 
The matrix $\mathcal X(a)$ representing an element $a\in A$ is an $|I| \times |I|$ block matrix. In the case of an arrow $a\in\bar{Q}$, we can use the idempotents to write $a=e_{t(a)}ae_{h(a)}$, so $a$ is represented by a matrix with at most one non-zero block of size $\alpha_{t(a)}\times \alpha_{h(a)}$ placed in the $t(a)$-th block row and $h(a)$-th block column. Therefore, this can be viewed as a \emph{quiver representation}, consisting of vector spaces $\VV_s=\CC^{\alpha_s}$, $s\in I$ and linear maps $X_a: \VV_{h(a)}\to \VV_{t(a)}$ for each $a\in \bar{Q}$.  With this interpretation, we have
\begin{equation}\label{xa}
X_a\in \Mat_{\alpha_{t(a)}, \alpha_{h(a)}} (\CC)\,, \qquad \Rep(\CC\bar{Q}, \alpha)\cong \prod_{a\in \bar{Q}} \Mat_{\alpha_{t(a)}, \alpha_{h(a)}} (\CC).
\end{equation}
Next, $\Rep(A, \alpha)$ is an affine open subset of $\Rep(\CC\bar{Q}, \alpha)$, so it is also smooth. It is naturally acted on by $\prod_{i\in I} \Gl_{\alpha_s}$ through conjugation. This induces an action of $G(\alpha)$ as defined in \eqref{galpha}. 
 By Theorems \ref{Thm:QStruct} and \ref{Thm:Poiss},   and Proposition \ref{Prop:BrInduced}, $\Rep(A, \alpha)$ is a quasi-Hamiltonian $G(\alpha)$-manifold, with  quasi-Poisson bracket defined by the bivector $\tr(\mathcal{X}(\PP))$ and with multiplicative moment map $\mathcal{X}(\Phi)$. The representation space $\Rep(\Lambda^q, \alpha)$ corresponds to the subset such that $\mathcal X (\Phi)=\prod_s q_s \Id_{\alpha_s}$, so it is a closed affine subvariety in $\Rep(A, \alpha)$. We set $q^{\alpha}=\prod_{s\in I} q^{\alpha_s}$, and note that $\Rep(\Lambda^q, \alpha)$ is empty when $q^\alpha\neq 1$ by  \cite[Lemma 1.5]{CBShaw}.

The points in the affine variety $\mathcal{S}_{\alpha, q}:=\Rep(\Lambda^q, \alpha)/\!/G(\alpha)$ are closed $G(\alpha)$ orbits of $\Rep(\Lambda^q, \alpha)$, so  correspond to semi-simple representations of $\Lambda^{q}$ of dimension $\alpha$. 
In the case where all representations in $\Rep(\Lambda^q, \alpha)$ are simple, we have the following description of the space.  
\begin{thm}  
\label{dim} \cite[Theorem 2.8]{CF}
Let $p(\alpha)=1+\sum_{a\in Q}\alpha_{t(a)}\alpha_{h(a)}-\alpha\cdot\alpha$, where $\alpha\cdot\alpha=\sum_{s\in I}\alpha_s^2$. Suppose that $\Rep(\Lambda^q, \alpha)$ is non-empty and all representations in $\Rep(\Lambda^q, \alpha)$ are simple. Then $\alpha$ is a positive root of $Q$ and $\Rep(\Lambda^q, \alpha)$ is a smooth affine variety of dimension $g+2p(\alpha)$, with $g=\dim G(\alpha)=\alpha\cdot\alpha-1$. The group $G(\alpha)$ acts freely on $\Rep(\Lambda^q, \alpha)$, so $\mathcal S_{\alpha, q}=\Rep(\Lambda^q, \alpha)/G(\alpha)$ is a Poisson manifold of dimension $2p(\alpha)$, obtained by quasi-Hamiltonian reduction.
\end{thm}
As we explained in \ref{ss:RepSp}, the Poisson bracket on $\mathcal O(\mathcal S_{\alpha, q})=\mathcal O(\Rep(\Lambda^q, \alpha))^{G(\alpha)}$ is obtained from the $G(\alpha)$-invariant bivector field $\tr(\mathcal{X}(\PP))$. Moreover, since $\PP$ is non-degenerate by \cite[Section 8]{VdB2}, Corollary \ref{Cor:Sympl} yields that   $\mathcal S_{\alpha, q}$ is a symplectic manifold when any representation in $\Rep(\Lambda^q, \alpha)$ is simple.

\medskip
  
Let $Q$ be an arbitrary quiver with vertex set $I$. A \emph{framing} of $Q$ is a quiver $\widetilde{Q}$ 
with set of vertices $\widetilde{I}=I \cup \{\infty\}$ and whose arrows are the ones of $Q$ together with  additional arrows $ \infty \to s$ to the vertices of $Q$. We allow multiple arrows to a single vertex. Given arbitrary $\alpha\in \N^I$ and $q\in({\CC^\times})^I$, we extend them from $I$ to $\widetilde{I}$ by putting $\alpha_\infty=1$ and $q_\infty=q^{-\alpha}$, i.e. $\aalpha=(1, \alpha)$ and $\qq=q^{-\alpha}e_{\infty}+\sum_{s\in I} q_se_s$. By construction $\qq^{\aalpha}=1$.
 We can consider the multiplicative preprojective algebra of $\widetilde{Q}$ with parameter $\qq$, and consider the representation space $\Rep(\Lambda^{{\qq}}, \aalpha)$.  We refer to the quotients 
 \begin{align}
 \label{qvar}
 \Maq:=\Rep(\Lambda^{\qq}, \aalpha)/\!/G(\aalpha)\,, \quad \text{where }
 G(\aalpha)\cong \prod_{s\in I} \Gl_{\alpha_s}=\Gl_{\alpha}\,,
 \end{align} 
as \emph{multiplicative quiver varieties}, that we abbreviate MQV. 

We say that $q=\sum_{s\in I} q_se_s$ is \emph{regular} if $q^\alpha\ne 1$ for any root $\alpha$ of the quiver $Q$. We have the following result, which is a multiplicative analogue of \cite[Theorem 2.8]{Nak94}, \cite[Proposition 3]{BCE}. 

\begin{prop}\label{mqvar}  \cite[Proposition 2.9]{CF}
Choose an arbitrary framing $\widetilde{Q}$ of $Q$ and let $\aalpha$ and $\qq$ be defined as above. If $q$ is regular, then every module of dimension $\aalpha$ over the multiplicative preprojective algebra $\Lambda^\qq$ is simple. Hence, the group $\Gl_\alpha$ acts freely on $\Rep(\Lambda^{{\qq}}, \aalpha)$ and the MQV $\Maq$ is smooth.
\end{prop}
When $q$ is regular and $\Maq\ne\emptyset$, this implies that $\aalpha=(1, \alpha)$ is a positive root of $\widetilde Q$ and $\Maq$ is a smooth affine variety of dimension $2p(\aalpha)$ by Theorem \ref{dim}.


\section{Quasi-Hamiltonian algebra structure}   \label{ss:qHcyclic} 

The developments of this section are parallel to \cite{CF,CF2}, and can be seen as application of Van den Bergh's work \cite{VdB1} that we recalled in \ref{ss:dAS}--\ref{ss:MultqPrep}. Fix $m,d\geq 2$ and let $I=\Z/m\Z$. Except when it is stated differently, we assume for the rest of this paper that we take the indices $r,s$ in $I$, and  that the Greek letters $\alpha,\beta,\gamma,\epsilon$ placed as indices always range through $1,\ldots,d$.

By a spin cyclic quiver, we mean the double quiver $\bar{Q}$ of a quiver $Q$, where $Q$ has $m+1$ vertices labelled by  $I \cup \{\infty\}$, $m$ arrows $x_s:s \to s+1$  
and $d$ framing arrows  $v_1,\ldots,v_d:\infty \to 0$. 
We write for the doubles $y_s=x_s^*:s+1 \to s$ and $w_\alpha=v_\alpha^*: 0 \to \infty$. We consider the following  ordering at each vertex 
\begin{equation*}
\begin{aligned}
     &\text{at }\infty:\quad & v_{1} < w_{1} < \ldots < v_{d} < w_{d} \,, \\
     &\text{at }s\in I \setminus \{0\}:\quad & x_s< y_s < x_{s-1} < y_{s-1} \,, \\
     &\text{at }0:\quad & x_0< y_0< x_{m-1}< y_{m-1} < v_{1} < w_{1} < \ldots < v_{d} < w_{d}\,.
\end{aligned}
\end{equation*}
We form the algebra $A$ obtained by inverting all the elements $(1+aa^*)_{a\in \bar{Q}}$ in $\CC\bar{Q}$. Using Proposition \ref{Pr:dbr}, we get a double  quasi-Poisson bracket on $A$, which satisfies the following identities between the arrows of  the cycle  
 \begin{subequations}
       \begin{align}
\dgal{x_r,x_s}\,=\,&\frac12  x_{r-1}x_r\otimes e_r\, \delta_{(s,r-1)}\,
-\,\frac12 e_{r+1}\otimes x_r x_{r+1}\,\delta_{(s,r+1)} \,, \label{cyida} \\
\dgal{y_r,y_s}\,=\,& \frac12 e_r \otimes y_ry_{r-1} \, \delta_{(s,r-1)}\, 
-\,\frac12 y_{r+1}y_r \otimes e_{r+1}\, \delta_{(s,r+1)}\,,\label{cyida'}\\
\dgal{x_r,y_s}\,=\,&\delta_{sr}\left( e_{r+1}\otimes e_{r}
+\frac{1}{2} y_rx_r\otimes e_{r} +\frac{1}{2} e_{r+1}\otimes x_ry_r \right) \nonumber \\
&-\frac12  x_r \otimes y_{r-1}\,\delta_{(s,r-1)} + \frac12  y_{r+1}\otimes x_r\, \delta_{(s,r+1)}
\,. \label{cyidb} 
\end{align}
  \end{subequations}
Note the difference of signs for last two terms in \eqref{cyidb} compared to \cite[(4.1c)]{CF}, which is a consequence of the different ordering taken at each vertex $s\in I$. 
The double brackets involving elements of the cycle and framing arrows are determined by
 \begin{subequations}
       \begin{align}
\dgal{x_r, w_{\alpha}}\,=\,& \frac12 \delta_{(r,m-1)}\, e_{r+1}\otimes x_r w_{\alpha}
-\frac12 \delta_{(r,0)}\, x_r\otimes w_{\alpha}\,,\label{cyidd}\\
\dgal{x_r, v_{\alpha}}\,=\,& \frac12 \delta_{(r,0)}\, v_{\alpha} x_r\otimes e_r
-\frac12 \delta_{(r,m-1)}\, v_{\alpha}\otimes x_r\,,\label{cyidd'}\\
\dgal{y_r, w_{\alpha}}\,=\,& \frac12 \delta_{(r,0)}\,e_r\otimes y_rw_{\alpha}
-\frac12 \delta_{(r,m-1)}\, y_r \otimes w_{\alpha}\,,\label{cyide} \\
\dgal{y_r, v_{\alpha}}\,=\,& \frac12\delta_{(r,m-1)} v_{\alpha} y_r\otimes e_{r+1}
-\frac12 \delta_{(r,0)} v_{\alpha}\otimes y_r\,.\label{cyide'} 
\end{align}
  \end{subequations}
The remaining double brackets are nothing else than 
\begin{subequations}
       \begin{align}
\dgal{v_\alpha,v_\beta}\,=\,&-\frac12 \,o(\alpha,\beta) 
\left(v_\beta\otimes v_\alpha + v_\alpha \otimes v_\beta \right)\,, \label{tadidv}\\
\dgal{w_\alpha,w_\beta}=\,&-\frac12 \,o(\alpha,\beta) 
\left(w_\beta\otimes w_\alpha + w_\alpha \otimes w_\beta \right)\,, \label{tadidw}\\
\dgal{v_\alpha,w_\beta}=\,& \delta_{\alpha \beta}\left(  e_0\otimes e_\infty
+\frac12 w_\alpha v_\alpha \otimes e_\infty + \frac12 e_0 \otimes v_\alpha w_\alpha \right) \nonumber \\
\,& + \frac12 \,o(\alpha,\beta) 
\left(e_0\otimes v_\alpha w_\beta + w_\beta v_\alpha \otimes e_\infty \right)\,.\label{tadidu}
	\end{align}
\end{subequations}
To derive \eqref{tadidv}, note that  Proposition \ref{Pr:dbr} gives for $\alpha < \beta$ that $\dgal{v_\alpha,v_\beta}=-\frac12  (v_\beta\otimes v_\alpha + v_\alpha \otimes v_\beta)$. This is because  $v_\alpha < v_\beta$ and their heads/tails coincide. We then find \eqref{tadidv} by cyclic antisymmetry of the double bracket. Identities \eqref{tadidw} and \eqref{tadidu} are obtained in the same way. 

Introduce the elements $x=\sum_s x_s$, $y=\sum_s y_s$ and set $F_a=\sum_{s\in I} e_{s+a}\otimes e_s \in A\otimes A$ for any $a\in \Z$. Of great help for our study are the elements 
\begin{equation}
 F_1=\sum_{s\in I} e_{s+1}\otimes e_s\,,\quad 
F_{-1}=\sum_{s\in I} e_{s-1}\otimes e_{s}=\sum_{s\in I} e_{s}\otimes e_{s+1}\,.
\end{equation}
With these notations, Equations \eqref{cyida}--\eqref{cyidb} become 
 \begin{subequations}
       \begin{align}
\dgal{x,x}\,=\,&\frac12 \left(x^2 F_1-F_1 x^2 \right) \,, \quad 
\dgal{y,y}\,=\,-\frac12 \left(y^2 F_{-1}-F_{-1}y^2 \right) \label{cyA}\\
\dgal{x,y}\,=\,&F_1+\frac12 \left(yx F_1 +F_1 xy - x F_1 y + y F_1 x \right)\,, \label{cyB}
\end{align}
  \end{subequations}
while \eqref{cyidd}--\eqref{cyide'} take the form 
 \begin{subequations}
       \begin{align}
\dgal{x, w_{\alpha}}\,=\,& \frac12\, e_{0}\otimes xw_{\alpha}
-\frac12 \, e_0 x\otimes w_{\alpha}\,,\quad 
\dgal{x, v_{\alpha}}\,=\, \frac12\, v_{\alpha} x\otimes e_0
-\frac12\, v_{\alpha}\otimes x e_0\,,\label{cyD}\\
\dgal{y, w_{\alpha}}\,=\,& \frac12\,e_0\otimes yw_{\alpha}
-\frac12\,e_0 y \otimes w_{\alpha}\,,\quad 
\dgal{y, v_{\alpha}}\,=\, \frac12\, v_{\alpha} y\otimes e_{0}
-\frac12 \, v_{\alpha}\otimes y e_0\,.\label{cyE} 
\end{align}
  \end{subequations}
(The expressions \eqref{cyD}--\eqref{cyE} could be written using $F_1$ and $F_{-1}$ instead of writing the 
idempotents $e_0$, but this form is not better suited for calculations.) Adding to $A$ local inverses $x_s^{-1}=e_{s+1} x_s^{-1} e_s$ such that 
$x_s x_s^{-1}=e_s$ and $x_s^{-1}x_s=e_{s+1}$,  we get locally invertible elements $z_s=y_s+x_s^{-1}$ and we form $z=x^{-1}+y=\sum_s z_s$. Equations \eqref{cyA}, \eqref{cyD}--\eqref{cyE} can be written with $z$ instead of $y$, while  \eqref{cyB} becomes 
\begin{equation}
  \dgal{x,z}\,=\,\frac12 \left(zx F_1 +F_1 xz - x F_1 z + z F_1 x \right)\,.  \label{cyBz}
\end{equation}

The algebra $A$ is quasi-Hamiltonian for the double bracket given above and the multiplicative moment map 
$\Phi=\sum_s e_s\Phi e_s+e_\infty \Phi e_\infty$ where  
\begin{subequations}
    \begin{align}
e_0 \Phi e_0=&(e_0+x_0y_0)(e_0+y_{m-1}x_{m-1})^{-1}\, \prod_{\alpha=1,\ldots,d}^{\longrightarrow}(e_0+w_\alpha v_\alpha)^{-1}\,, \label{cy0} \\ 
e_s \Phi e_s=&(e_s+x_sy_s)(e_s+y_{s-1}x_{s-1})^{-1}\,,\quad s\in I \setminus \{0\}\,, \label{cy1} \\ 
e_\infty \Phi e_\infty=&\, \prod_{\alpha=1,\ldots,d}^{\longrightarrow} (e_\infty+v_\alpha w_\alpha)\,. \label{cy2}
    \end{align}
\end{subequations}
Here, we use the invertibility of $1+x_0y_0$ and the idempotent decomposition of $1\in A$ to get that $e_0+x_0y_0$ is locally invertible with inverse $(e_0+x_0y_0)^{-1}:=e_0(1+x_0y_0)^{-1}e_0$, and the same holds for all the arrows in $\bar{Q}$. 
If we further localise at $x$, we can write $(e_s+x_sy_s)(e_s+y_{s-1}x_{s-1})^{-1}$ as $x_s z_s x_{s-1}^{-1}z_{s-1}^{-1}$. 

\medskip

Following the Jordan quiver case \cite[\S 3.1]{CF2}, we introduce the spin elements 
\begin{equation} \label{spinCyc}
 a'_{\alpha}=w_{\alpha}\,, \quad 
c'_{1}=v_{1}z\,, \quad 
c'_{\alpha}=v_{\alpha}(e_0+w_{\alpha-1}v_{\alpha-1})\ldots (e_0+w_{1}v_{1})z\,,
\end{equation}
and  we can define  $c'_{\alpha}$ inductively using 
\begin{equation} \label{cyInd}
 c'_{\alpha}=\sum_{\lambda=1}^{\alpha-1}v_{\alpha} w_{\lambda} c'_{\lambda} + v_{\alpha} z\,.
\end{equation}
It is important to remark that $ a'_{\alpha}=e_0  a'_{\alpha} e_\infty$ but 
$ c'_{\alpha}=e_\infty  c'_{\alpha} e_{m-1}$ is not a path to $0$. This is due to the fact that 
$v_{\alpha}z=v_{\alpha}z_{m-1}$. 
To get the double brackets between the elements $(x,z,a'_{\alpha},c'_{\alpha})$, it remains to find the ones involving $c'_\alpha$. The next result is obtained in Appendix \ref{Ann:cyclic} :  
\begin{lem} \label{Lem:Calg1}
  For any $\alpha,\beta=1,\ldots,d$, 
\begin{subequations}
 \begin{align}
\dgal{x, c'_{ \alpha}}\,=\,& \frac12 c'_{ \alpha} x\otimes e_{m-1}+\frac12 c'_{ \alpha}\otimes xe_{m-1}\,,
\quad \dgal{z, c'_{ \alpha}}=-\frac12 c'_{ \alpha} z\otimes e_{m-1}
+\frac12 c'_{ \alpha}\otimes z e_{m-1}\label{CySpin1}\\
\dgal{a'_{\alpha}, c'_{\beta}}\,=\,& 
\frac12  \left(o(\alpha,\beta)-\delta_{\alpha \beta} \right)\, 
e_\infty\otimes a'_{\alpha} c'_{\beta}
- \delta_{\alpha \beta} \left(e_\infty \otimes e_0 z + \sum_{\lambda=1}^{\beta-1} 
e_\infty \otimes a'_{\lambda} c'_{\lambda} \right) \,,\label{CySpin2} \\
\dgal{c'_{\alpha},c'_{\beta}}\,=\,& \frac12 o(\alpha,\beta) 
\left(c'_{\beta} \otimes c'_{\alpha} - c'_{\alpha} \otimes c'_{\beta} \right)  \,.\label{CySpin3} 
 \end{align}
\end{subequations}
where the last sum in \eqref{CySpin2} is omitted for $\lambda=1$. 
\end{lem}

Motivated by the geometric interpretation through \eqref{relInv}, we are interested in the bracket $\br{-,-}$ associated to $\dgal{-,-}$  between the elements  $x^k$ and $a'_{\alpha} c'_{\beta} x^l$, for any $k,l\in \N$. Hereafter, given any $b\in A$ we also denote by $b$ its equivalence class in $A/[A,A]$. 
We compute 
\begin{lem} \label{LodCyxac}
 For any  $k,l\geq 1$, we get in $A/[A,A]$, 
\begin{subequations}
 \begin{align}
  \br{x^k,x^l}=&0\,, \quad
\br{x^k,a'_{\alpha} c'_{\beta} x^l}=k\, a'_{\alpha} c'_{\beta} x^{k+l}\,, \label{xxacx}\\
\br{a'_{\gamma} c'_{\epsilon} x^k,a'_{\alpha} c'_{\beta} x^l}=&
\frac12 \left(\sum_{v=1}^k-\sum_{v=1}^l \right) 
\left(a'_{\alpha} c'_{\beta} x^v a'_{\gamma} c'_{\epsilon} x^{k+l-v}
+a'_{\alpha} c'_{\beta} x^{k+l-v} a'_{\gamma} c'_{\epsilon} x^{v} \right) \nonumber \\
&+\frac12  o(\alpha,\gamma) \left(a'_{\gamma} c'_{\epsilon} x^k a'_{\alpha} c'_{\beta} x^l 
+a'_{\alpha} c'_{\epsilon} x^k a'_{\gamma} c'_{\beta} x^l\right) \nonumber \\
&+\frac12  o(\epsilon,\beta) \left(a'_{\alpha} c'_{\beta} x^k a'_{\gamma} c'_{\epsilon} x^l 
-a'_{\alpha} c'_{\epsilon} x^k a'_{\gamma} c'_{\beta} x^l\right) \nonumber \\
&+\frac12 [o(\epsilon,\alpha)+\delta_{\alpha \epsilon}]\,a'_{\alpha} c'_{\epsilon} x^k a'_{\gamma} c'_{\beta} x^l 
-\frac12 [o(\beta,\gamma)+\delta_{\beta \gamma}]\,a'_{\alpha} c'_{\epsilon} x^k a'_{\gamma} c'_{\beta} x^l \nonumber \\
&+\delta_{\alpha \epsilon} \left(z+\sum_{\lambda=1}^{\epsilon-1} a'_{\lambda} c'_{\lambda}  \right)x^k a'_{\gamma} c'_{\beta} x^l 
-\delta_{\beta \gamma}\,\, a'_{\alpha} c'_{\epsilon} x^k
\left(z + \sum_{\mu=1}^{\beta-1} a'_{\mu} c'_{\mu}  \right)x^l\,. \label{xacxac}
 \end{align}
\end{subequations}
In particular, in order for the elements on which we take the bracket to be nonzero in $A/[A,A]$, we need $k= 0$ mod $m$ for $x^k$, or $l-1=0$ mod $m$ for $a'_{\alpha} c'_{\beta} x^l$.
\end{lem}
In fact, \eqref{xxacx} holds in $A$ for the left Loday bracket $\br{-,-}$, i.e. we do not necessarily need to regard it as the operation $A/[A,A] \times A/[A,A] \to A/[A,A]$. 
Note also that it is sufficient to consider in the two first sums over $v$ in \eqref{xacxac} the terms for which $v=1$ mod $m$. All the computations are provided in Appendix \ref{Ann:cyclic}. 

We work in slightly more general setting from now on, and consider $u\in \{x,y,z,\sum_s e_s+xy\}$. We already have   $\epsilon(x)=+1$, $\epsilon(y)=-1$, and we set  $\epsilon(z)=-1$, $\epsilon(\sum_s e_s +xy)=+1$.  We can write in the three first cases $\dgal{u,u}=\frac12 \epsilon(u)[u^2 F_{\epsilon(u)} - F_{\epsilon(u)} u^2]$, while   $\dgal{u,u}=\frac12 [u^2 F_{0} - F_{0}  u^2]$ if $u=\sum_s e_s+xy$. 
 The identities can be directly checked, see e.g. \cite[Lemmas A.1, A.2]{CF} for some of them. 
Using these brackets, we get the following result which is proved in Appendix \ref{Ann:cyclic}. 
\begin{lem} \label{Lem:Cytt} 
For any $k,l\geq 1$, $\alpha,\beta=1,\ldots,d$, we have $\dgal{u^k,w_\alpha v_\beta u^l}=0$ under the left Loday bracket in $A$. The equality descends to $A/[A,A]$ and we have in particular that the $\CC$-vector space generated by the elements $(u^k,w_1v_1 u^k)$ is a commutative Lie subalgebra in $(A/[A,A],\br{-,-})$.  
\end{lem}

Let $\phi$ be the moment map associated to the subquiver supported at $I$. That is $\phi=\sum_s \phi_s $ for $\phi_s=(e_s+x_sy_s)(e_s+y_{s-1}x_{s-1})^{-1}$ or $\phi_s= x_sz_sx_{s-1}^{-1}z_{s-1}^{-1}$ when we localise $A$ at $x$. We assume in the next result that $A$ is localised at $u$. 
 \begin{prop} \label{PropInvPhiCy}
 Let $U_{+,\eta}=u(1+\eta \phi)$, $U_{-,\eta}=u(1+\eta \phi^{-1})$, for arbitrary $\eta \in \CC$ playing the role of a spectral parameter. Let $K,L \in \N^\times$. Then, if $\epsilon(u)=-1$,  
 \begin{equation}
  \br{U_{+,\eta}^K,U_{+,\eta'}^L}=0 \quad \text{ mod }[A,A], \qquad \text{ for any }\eta,\eta'\in \CC\,.
 \end{equation}
 \end{prop}
 If $\epsilon(u)=+1$,  
 \begin{equation}
  \br{U_{-,\eta}^K,U_{-,\eta'}^L}=0 \quad \text{ mod }[A,A], 
 \qquad \text{ for any }\eta,\eta'\in \CC\,.
 \end{equation}
Note that when $u$ is not $\sum_s e_s+xy$ and $K$ is not divisible by $m$, then $U_{+,\eta}^K=0$ mod $[A,A]$ and the result is trivial. This is because $u^m\in \oplus_s e_s A e_s$ but $u\notin \oplus_s e_s A e_s$ in those cases. 
 The proof of Proposition \ref{PropInvPhiCy} is provided in Appendix \ref{Ann:cyclic}. 

\medskip  

For a $(m+1)$-uple $(q_\infty,q_s)\in \CC^\times\times (\CC^\times)^I$, we set $\qq=q_\infty e_\infty + \sum_s q_s e_s$ and the multiplicative preprojective algebra $\Lambda^\qq$ is the quotient of $A$ by the two-sided ideal $J$ generated by the relation 
$\Phi=\qq$, where $\Phi$ is given by \eqref{cy0}--\eqref{cy2}. The different equalities derived above in the Lie algebra $(A/[A,A],\br{-,-})$ descend to $\Lambda^\qq/[\Lambda^\qq,\Lambda^\qq]$ by \ref{ss:doubleqP}.


\section{A natural phase space for the spin RS model} \label{tadpole}  

We recall the important results from \cite{CF2} that are needed for our study, and we take the freedom to adapt the notations to suit our case. 

Fix $t\in \CC^\times$ not a root of unity, $n \in \N^\times$ and $d\geq 2$.  A dense open subset $\Cntd^\circ$ of the  MQV $\Cntd$ of dimension $(1,n)$ defined from a spin Jordan quiver with $d$ framing arrows is given by equivalence classes of 
$2d+2$ elements $(A,B,V_\alpha',W_\alpha')$,  where $A,B\in \Gl_n(\CC)$ and $V_\alpha'\in \Mat_{1\times n}$,  $W_\alpha' \in \Mat_{n\times 1}$ are matrices  satisfying 
\begin{equation} \label{momap}
  ABA^{-1}B^{-1}\prod_{\alpha=1,\ldots,d}^{\longrightarrow}(\Id_n+W_\alpha' V_\alpha')^{-1}\, =\, t \Id_n\,,
\end{equation}
under the equivalence defined from the action of the group $\Gl_n(\CC)$  by
\begin{equation}
g. (A,B,V_\alpha',W_\alpha')=(gAg^{-1},gBg^{-1}, V_\alpha' g^{-1}, gW_\alpha')\,,\quad g\in\Gl_n(\CC)\,.
\end{equation}
In fact, $\Cntd^\circ$ is a smooth symplectic complex manifold of dimension $2nd$, and we denote its Poisson bracket by $\brap{-,-}$.  For any representative of an equivalence class, we form the matrices $\As\in \Mat_{n\times d}(\CC)$ and $\Cs\in \Mat_{d\times n}(\CC)$ by 
  \begin{equation} \label{AsCs}
\As_{i\alpha}=\left[W_\alpha'\right]_i \, ,\quad
\Cs_{\alpha j}=\left[V_\alpha'(\Id_n+W_{\alpha-1}'V_{\alpha-1}')\ldots (\Id_n+W_1'V_1')B\right]_j\,,
  \end{equation}
so that the moment map equation can be rewritten as 
\begin{equation} \label{EqBspTad}
 ABA^{-1}=t B+ t \As \Cs\,.
\end{equation}
Then a point of $\Cntd^\circ$ is determined by the equivalence class of $(A,B,V_\alpha',W_\alpha')$ as above, or equivalently by an element $(A,B,\As,\Cs)$ modulo identification by the group action 
$g \cdot (A,B,\As,\Cs)=(gAg^{-1},gBg^{-1}, g \As,\Cs g^{-1})$ for any $g\in \Gl_n(\CC)$. 
 Choosing the functions 
\begin{equation} \label{fgTadp}
  f_k:= \tr(A^k)\,, \quad g_{k}^{\alpha \beta}=\tr(\As E_{\alpha \beta} \Cs A^k)\,, \quad k\in \N,\,\alpha,\beta=1,\ldots,d\,,
\end{equation}
 the Poisson structure is determined from the identities 
\begin{subequations}
 \begin{align}
\brap{f_k,f_l}
=\,&0\,, \quad \brap{f_k,g^{\alpha \beta}_l}
=\,k\,  g^{\alpha \beta}_{k+l}\,, \label{brtad1}\\
\brap{g^{\gamma \epsilon}_k,g^{\alpha \beta}_l}
=\,&
\frac12 \left(\sum_{r=1}^k-\sum_{r=1}^l \right) 
\left(\tr(\As E_{\alpha \beta} \Cs A^r \As E_{\gamma\epsilon} \Cs A^{k+l-r})
+\tr(\As E_{\alpha \beta} \Cs A^{k+l-r} \As E_{\gamma\epsilon} \Cs A^r)\right) \nonumber \\
&+\frac12 o(\alpha,\gamma) \left(\tr(\As E_{\gamma\epsilon} \Cs A^k \As E_{\alpha\beta} \Cs A^l)
+ \tr(\As E_{\alpha\epsilon} \Cs A^k \As E_{\gamma\beta} \Cs A^l) \right) \nonumber \\
&+\frac12 o(\epsilon,\beta) 
\left( \tr( \As E_{\alpha\beta} \Cs A^k \As E_{\gamma\epsilon} \Cs A^l) 
- \tr(\As E_{\alpha\epsilon} \Cs A^k \As E_{\gamma\beta} \Cs A^l) \right) \nonumber \\
&+\frac12 [o(\epsilon,\alpha)+\delta_{\alpha \epsilon}]\,
\tr(\As E_{\alpha\epsilon} \Cs A^k \As E_{\gamma\beta} \Cs A^l) 
-\frac12 [o(\beta,\gamma)+\delta_{\beta \gamma}]\,
\tr(\As E_{\alpha\epsilon} \Cs A^k \As E_{\gamma\beta} \Cs A^l) \nonumber \\
&+\delta_{\alpha \epsilon} 
\tr\left( \left[B  + \sum_{\lambda=1}^{\epsilon-1} \As E_{\lambda \lambda} \Cs  \right] A^k
\As E_{\gamma \beta} \Cs A^l\right)  
-\delta_{\beta \gamma} \tr\left( \left[B + \sum_{\mu=1}^{\beta-1}\As E_{\mu \mu} \Cs \right] A^l
\As E_{\alpha\epsilon} \Cs A^k \right)\,, \label{brtad2}
 \end{align}
\end{subequations}
see \cite[Lemma A.2]{CF2}. This space admits local coordinates on an open subset which is dense in (a connected component of) $\Cntd^\circ$ as follows. 
Define the open subspace $\Cntd' \subset \Cntd^\circ$ which is such that for any equivalence class of quadruple $(A,B,\As,\Cs)\in \Cntd^\circ$, the matrix $A$ is diagonalisable with nonzero eigenvalues $(x_i)_i$ satisfying $x_i \neq x_j, x_i \neq t x_j$ for each $i \neq j$, and when we choose a representative with $A$ in diagonal form, the matrix $\As$ is such that the  entries in each of its rows sum up to a nonzero value. Hence  we can pick a representative with $A$ in diagonal form as above, and such that  $\sum_\alpha \As_{i \alpha}=1$ in $\Cntd'$. Note that this is uniquely defined up to action by a permutation matrix. 
Introduce 
\begin{equation} \label{Eq:Creg}
 \Creg=\{x=(x_1,\ldots,x_n)\in\CC^n\, \mid \, x_i\ne 0\,,\ x_i\ne x_j\,,\ x_i\ne t x_j\ \text{for all}\ i\ne j\}\,.  
 \end{equation}
For $\alpha=1,\ldots,d$ take  $(\aaa^\alpha)^\top,\ccc^\alpha\in \CC^n$. We define 
$\h\subset \Creg \times (\CC^n)^{\times d} \times (\CC^n)^{\times d}$ to be a generic subspace such that on global coordinates $(x_i,\aaa^\alpha_i,\ccc^\alpha_i)_{i \alpha}$ we require $\sum_{\alpha}\aaa^\alpha_i=1$, see \cite[\S 4.1]{CF2}. 
We can define a map $\xi : \h \to \Cntd'$ which associates to $(x_i,\aaa^\alpha_i,\ccc^\alpha_i)_{i \alpha}$ the equivalence class of the element $(A,B,\As,\Cs)$, where 
\begin{equation} \label{Tadiffeo}
\begin{aligned}
 & A=\diag(x_1,\ldots,x_n)\,, \quad B=(B_{ij})\,, \quad \As=(\As_{i\alpha})\,,\,\, \Cs=(\Cs_{\alpha i}) \\
& \text{with } B_{ij}=t \frac{\sum_\alpha \aaa_i^\alpha \ccc_j^\alpha}{x_i x_j^{-1}-t}\,, \quad
\As_{i\alpha}=\aaa^\alpha_i\,,\quad \Cs_{\alpha i}=\ccc_i^\alpha\,.
\end{aligned}
\end{equation}
This map is such that the following result holds.  
\begin{prop} \cite[Propositions 4.1, 4.3]{CF2} \label{Proptadpole}
  The map $\xi : \h / S_n \to \Cntd'$ given by \eqref{Tadiffeo} defines a local diffeomorphism.  It is a Poisson morphism when $\h / S_n$ is equipped with the Poisson bracket $\br{-,-}$ defined on coordinates by 
\begin{subequations}
 \begin{align}
  \br{x_j,x_i}=&0\,,\quad \br{\aaa_j^\alpha,x_i}=0\,,\quad 
\br{\ccc_j^\alpha,x_i}=-\delta_{ij}\ccc_j^\alpha x_i\,, \label{Eqh1} \\
\br{\aaa_j^\gamma,\aaa_i^\alpha}=&\frac12 \delta_{(i\neq j)}\frac{x_j+x_i}{x_j-x_i}
(\aaa_j^\gamma \aaa_i^\alpha+\aaa_i^\gamma \aaa_j^\alpha-\aaa_j^\gamma \aaa_j^\alpha - 
\aaa_i^\gamma \aaa_i^\alpha) +\frac12 o(\alpha,\gamma)
(\aaa_j^\gamma \aaa_i^\alpha+\aaa_i^\gamma \aaa_j^\alpha) \nonumber \\
&+\frac12 \sum_{\sigma=1}^d o(\gamma,\sigma)\aaa_i^\alpha 
(\aaa_j^\gamma \aaa_i^\sigma+\aaa_i^\gamma \aaa_j^\sigma)
-\frac12 \sum_{\kappa=1}^d o(\alpha,\kappa)\aaa_j^\gamma 
(\aaa_j^\kappa \aaa_i^\alpha+\aaa_i^\kappa \aaa_j^\alpha)\,, \label{Eqh2} \\
\br{\ccc_j^\epsilon,\aaa_i^\alpha}=&\delta_{\epsilon \alpha}B_{ij}-\aaa_i^\alpha B_{ij}+
\frac12 \delta_{(i\neq j)}\frac{x_j+x_i}{x_j-x_i}\ccc_j^\epsilon (\aaa_j^\alpha-\aaa_i^\alpha)
-\delta_{(\alpha<\epsilon)}\aaa_i^\alpha \ccc_j^\epsilon \nonumber \\
&-\aaa_i^\alpha \sum_{\lambda=1}^{\epsilon-1}\aaa_i^\lambda (\ccc_j^\lambda-\ccc_j^\epsilon) 
+\delta_{\epsilon \alpha} \sum_{\lambda=1}^{\epsilon-1} \aaa_i^\lambda \ccc_j^\lambda 
+\frac12 \sum_{\kappa=1}^d o(\alpha,\kappa)\ccc_j^\epsilon 
(\aaa_j^\kappa \aaa_i^\alpha+\aaa_i^\kappa \aaa_j^\alpha) \,, \label{Eqh3} \\
\br{\ccc_j^\epsilon,\ccc_i^\beta}=&
\frac12 \delta_{(i\neq j)}\frac{x_j+x_i}{x_j-x_i} (\ccc_j^\epsilon \ccc_i^\beta + \ccc_i^\epsilon\ccc_j^\beta) 
+\ccc_i^\beta B_{ij} - \ccc_j^\epsilon B_{ji} +\frac12 o(\epsilon,\beta)
(\ccc_i^\epsilon\ccc_j^\beta-\ccc_j^\epsilon \ccc_i^\beta) \nonumber \\
&+\ccc_i^\beta \sum_{\lambda=1}^{\epsilon-1}\aaa_i^\lambda (\ccc_j^\lambda-\ccc_j^\epsilon)
-\ccc_j^\epsilon \sum_{\mu=1}^{\beta-1}\aaa_j^\mu (\ccc_i^\mu-\ccc_i^\beta)
 \,. \label{Eqh4}
 \end{align}
\end{subequations}
Furthermore, letting $f_{ij}=\sum_\alpha \aaa_i^\alpha \ccc_j^\alpha$, the subalgebra generated by  $(x_i,f_{ij})_{ij}$ inside $\mathcal O (\h / S_n)$ is a Poisson subalgebra under $\br{-,-}$. 
\end{prop}
In these local coordinates, the matrix $B$ is the Lax matrix of the spin trigonometric RS system. Furthermore, the equation of motions defined by $\frac{d}{dt}=\brap{\tr B,-}$ are normalised versions of the equations of motion for the spin RS Hamiltonian, as originally defined in \cite{KrZ}.  The form of the Poisson brackets between the elements $(x_i,f_{ij})_{ij}$ was conjectured in \cite{AF}.


\section{MQV for spin cyclic quivers} \label{cyclic}

Consider the multiplicative preprojective algebra for a spin cyclic quiver as in Section \ref{ss:qHcyclic}. In accordance with \ref{Sec:mqv}, choose a dimension vector $\aalpha=(1,\alpha)$ with $\alpha\in (\N^\times)^I$, and set $q_\infty=q^{-\alpha}=\prod_{s\in I}q_s^{-\alpha_s}$. 

A representation of $\Lambda^\qq$ of dimension $\aalpha$ is a collection of vector spaces $(\VV_\infty, \VV_s)=(\CC, \CC^{\alpha_s})$ together with linear maps representing arrows of $\bar{Q}$ and satisfying \eqref{cy0}--\eqref{cy2}. 
Denote in an obvious way the matrices representing the arrows as $(X_s, Y_s, V_\alpha, W_\alpha)$. Therefore, points of $\Rep(\Lambda^\qq, \aalpha)$ are represented by $2m+2d$ elements $(X_s,Y_s,V_\alpha,W_\alpha)$, 
\begin{equation*}
 X_s\in \Mat_{\alpha_s\times \alpha_{s+1}}(\CC),\quad Y_s\in \Mat_{\alpha_{s+1}\times \alpha_s}(\CC)\,, \quad V_\alpha\in \Mat_{\alpha_0\times 1}(\CC),\quad W_\alpha\in \Mat_{1\times \alpha_0}(\CC)\,,
\end{equation*}
for $s\in I$, $\alpha \in \{1,\ldots,d\}$ satisfying 
  \begin{subequations}
       \begin{align}
&(\Id_{\alpha_0}+X_0Y_0)(\Id_{\alpha_0}+Y_{m-1}X_{m-1})^{-1}=q_0 \prod_{\alpha=1,\ldots,d}^{\longleftarrow}(\Id_{\alpha_0}+W_\alpha V_\alpha)\,, \label{Eq:Condcy0} \\
&(\Id_{\alpha_s}+X_sY_s)(\Id_{\alpha_s}+Y_{s-1}X_{s-1})^{-1}=q_s \Id_{\alpha_s}\,, \qquad s\in I \setminus \br{0}\,, \label{Eq:Condcy1} \\
&\prod_{\alpha=1,\ldots,d}^{\longrightarrow} (1+V_\alpha W_\alpha)=q_{\infty}\,, \label{Eq:Condcy2}
       \end{align}
  \end{subequations} 
and such that all factors have nonzero determinant. 
The group $G(\aalpha)=\prod_{s\in I}\Gl_{\alpha_s}(\CC)$ acts on these elements by
\begin{equation}\label{gactCy}
g. (X_s,Y_s,V_\alpha,W_\alpha)=(g_sX_sg_{s+1}^{-1},g_{s+1}Y_sg_s^{-1}, V_\alpha g_0^{-1}, g_0 W_\alpha)\,,\quad g\in G(\aalpha)\,, 
\end{equation}
and the orbits in $\Rep\left(\Lambda^{\qq},\aalpha\right)/\!/G(\aalpha)$ correspond to isomorphism classes of semisimple representations. We are particularly interested in the cases where $X=\sum_s X_s$ is invertible at some points, hence we restrict our attention to the spaces such that $\alpha_s=n\in \N^\times$ for all $s\in I$. 
We now define 
\begin{equation*}
\Cnm=\Rep\left(\Lambda^{\qq},\aalpha\right)/\!/G(\aalpha)\,,
\end{equation*}
which is  the spin  analogue of the space $\mathcal C_{n,q}(m)$ introduced in \cite[Section 4]{CF}, see also \cite[Section 5]{BEF}. 
By construction, this is a MQV for a framed cyclic quiver, and we denote its Poisson bracket by $\brap{-,-}$.  Let us identify $I$ and $\{0,\ldots,m-1\}$ to introduce the elements 
\begin{equation}
  \label{tcyclic}
t_s:=\prod_{0\leq s' \leq s} q_{s'}\,, \quad s=0,\ldots,m-1\,,\quad t:=t_{m-1}\,,
\end{equation}
so that $q_\infty=t^{-n}$. We also set $t_{-1}=1$ to state the next result, which is an application of \ref{Sec:mqv} with the regularity condition from \cite[Section 4]{CF}. 
\begin{prop}
Suppose that $t_{s}^{-1}t_{s'}\neq t^k$ for any $k \in \Z$ and $-1\leq s\leq s'\leq m-1$, with $k\neq 0$ if $s=s'$. Then  $\Cnm$ is a smooth variety of dimension $2nd$, with a non-degenerate Poisson bracket denoted by $\brap{-,-}$.
\end{prop}
From now on, we assume the regularity condition of the proposition. In particular, $t$ is not a root of unity. 

\subsection{Local coordinates}   \label{ss:CyCoord}

Consider the open subset $\Cnmo\subset \Cnm$ on which the $X_s$ are invertible. Similarly to the Jordan quiver case reviewed in Section \ref{tadpole}, introduce $Z_s:=Y_s+X_s^{-1}$, and form the matrices $\Asm\in \Mat_{n\times d}(\CC)$, $\Csm\in \Mat_{d\times n}(\CC)$ by 
  \begin{equation} \label{AsCscyclic}
\Asm_{i\alpha}=\left[W_\alpha\right]_i \, ,\quad
\Csm_{\alpha j}=t^{-1}\,\left[V_\alpha(\Id_n+W_{\alpha-1}V_{\alpha-1})\ldots (\Id_n+W_1V_1)Z\right]_j\,.
  \end{equation}
so that the $\alpha$-th column of $\Asm$  represents the spin element $a'_\alpha$, while the $\alpha$-th row of $\Csm$ represent $t^{-1} c'_\alpha$. Note the factor $t^{-1}$ necessary to define $\Csm$. 
In particular, \eqref{Eq:Condcy0}--\eqref{Eq:Condcy1}   adopt the compact form 
\begin{equation} \label{EqZspCyc}
 X_0Z_0X_{m-1}^{-1}=q_0 Z_{m-1}+ q_0 t \,\Asm\Csm\,, \quad X_s Z_s =q_s Z_{s-1} X_{s-1}\,.
\end{equation}
 Then \eqref{Eq:Condcy2} is just a corollary of these relations. 
Up to changing basis, a point 
$(X_s,Y_s,V_\alpha,W_\alpha)$ can be represented by an element of the equivalence class such that 
$X_0,\ldots,X_{m-2}=\Id_n$. Therefore, setting $A:=X_{m-1}$ and $B:=q_0^{-1}Z_0$, we find that the condition \eqref{EqZspCyc} gives 
$Z_s=t_s B$ for $s\neq m-1$, $Z_{m-1}=t A^{-1} B$, and $A,B$ satisfy 
\begin{equation} \label{commutCy}
  q_0 B A^{-1} =q_0 t A^{-1}B + q_0 t \,\Asm\Csm\,.
\end{equation}
Hence a point in $\Cnmo$ is completely determined by the data $(A,B,\Asm,\Csm)$ up to $\Gl_n$ action by 
$g \cdot (g Ag^{-1},g B g^{-1},g\Asm,\Csm g^{-1})$ seen as $\prod_s g \in G(\aalpha)$, with $A,B$ invertible and the elements $\Id_n+W_\alpha V_\alpha$ (that can be reconstructed from \eqref{AsCscyclic}) invertible. 
Comparing with \eqref{EqBspTad}, we find 
\begin{prop}\label{isoCM} 
Let $\Cnmo \subset \Cnm$ be as above. 
Let $\Cntd^\circ$ be the spin Ruijsenaars-Schneider space considered in Section \ref{tadpole} with parameter $t=\prod_s q_s$,  so that $\Cntd^\circ$ is a smooth variety. Then the map $\psi: \Cntd^\circ \to \Cnmo$ sending $(A,B, \As,\Cs)$ to $X_s=\Id_n$, $Z_s=t_sB$ for $s=0,\dots, m-2$ and $X_{m-1}=A$, $Z_{m-1}=tA^{-1}B$, $\Asm=A^{-1} \As$, $\Csm=\Cs$ defines an isomorphism of varieties. 
\end{prop}
\begin{proof}
  The only non-trivial identity to show is that we can recover  $\det(\Id_n + W_\alpha V_\alpha)\neq 0$ for all $\alpha$. We define in $\Cntd^\circ$ the elements  $W_\beta'=(\As_{i \beta})_i$, $C_\beta'=(\Cs_{\beta i})_i$ and inductively 
\begin{equation*}
  V_\beta'=C_\beta' B^{-1} (\Id_n+W_1'V_1')^{-1}\ldots (\Id_n+W_{\beta-1}'V_{\beta-1}')^{-1}\,.
\end{equation*}
By definition of the space $\Cntd^\circ$ with \eqref{AsCs}, $1+ V_\beta' W_\beta'\neq 0$ for all $\beta$. Now we remark in $\Cnmo$ that $W_\beta= ( (A^{-1}\As)_{i \beta})_i =A^{-1} W_\beta'$. We can also rewrite \eqref{AsCscyclic} as 
\begin{equation*}
V_\beta=C_\beta' B^{-1} A (\Id_n+W_1V_1)^{-1}\ldots (\Id_n+W_{\beta-1}V_{\beta-1})^{-1}\,,
\end{equation*}
to get that $V_1=V_1' A$ and by induction $V_\beta=V_\beta' A$. Thus $1+ V_\beta W_\beta=1+ V_\beta' W_\beta'\neq 0$ for all $\beta$. 
\end{proof}
We can, in fact, compare the Poisson structures on both spaces. 

\begin{prop}\label{IsoCMPoiss}
The isomorphism $\psi: \Cntd^\circ \to \Cnmo$ from Proposition \ref{isoCM} is Poisson. 
\end{prop}
The proof can be found in Appendix \ref{Ann:iso}. In particular, we can transfer the invariant local coordinates on $\Cntd' \subset \Cntd^\circ$ obtained in Proposition \ref{Proptadpole}  to the open subset $\Cnmp \subset \Cnmo$ defined by $\Cnmp=\psi(\Cntd')$. In such a case, we can always consider for a point of $\Cnmp$ a gauge with representative of the form $(X,Z,\Asm,\Csm)$ given by Proposition \ref{isoCM}, with the extra condition that $X_{m-1}$ is in diagonal form with diagonal entries $(x_i)_i$ defining a point of $\Creg$ \eqref{Eq:Creg} and $\sum_\alpha (X  \Asm)_{i \alpha} =1$ for all $i$.

\subsection{New variants of the spin RS system}   \label{ss:CyIS} 

Set $X=\sum_s X_s$, $Y=\sum_s Y_s$ and denote by $1$ the sum of the $m$ copies of the identity matrix on each $\VV_s=\CC^n$, $s\in I$.  Let $\Theta=(1+XY)(1+YX)^{-1}$ be the moment map restricted to the cyclic quiver without framing, so that $\Theta=\mathcal X(\phi)$ for $\phi$ defined in Section \ref{ss:qHcyclic}.  
 Proposition \ref{PropInvPhiCy} and \eqref{relInv} imply the following result.
\begin{thm} \label{CycInvol}
The following families of functions are Poisson commuting for any parameter $\eta$:  
\begin{equation*}
\begin{aligned}
& \left\{\tr \left((1+\eta\Theta^{-1})X\right)^j\,\, \big|\,\, j\in \N\right\}\,, \quad 
\left\{\tr \left((1+\eta\Theta^{-1})(1+XY)\right)^j\,\, \big|\,\, j\in \N\right\}\,,
\\
& \left\{\tr \left((1+\eta\Theta)Y\right)^j\,\, \big|\,\, j\in \N\right\}\,,\quad
\left\{\tr \left((1+\eta\Theta)(Y+X^{-1})\right)^j\,\, \big|\,\, j\in \N\right\}\,.
\end{aligned}
\end{equation*}
\end{thm}
\noindent Apart from the first family, we need $j \in m\N$ to have nonzero elements. 


\medskip

We will write down the families from Theorem \ref{CycInvol} in $\Cnmp$, where we can use the coordinates $(x_i,f_{ij})_{ij}$ with $f_{ij}=\sum_\alpha \aaa_i^\alpha \ccc_j^\alpha$. Hence, our first task is to use the known matrices $(A,B,S=\As\Cs)$ instead of the matrices describing a point in $\Cnmp$. 
Using the isomorphism from Theorem \ref{isoCM}, they are given by 
\begin{equation*}
  A=\diag(x_1,\ldots,x_n)\,, \quad B=(B_{ij})_{ij}, \,\, B_{ij}=\frac{t f_{ij} x_j}{x_i-tx_j}\,, \quad 
S=(S_{ij})_{ij},\,\, S_{ij}=f_{ij}\,.
\end{equation*}
Decomposing the moment map (restricted to the cycle) $\Theta=XZX^{-1}Z^{-1}$ as $\Theta=\sum_s \Theta_s$, we get from \eqref{Eq:Condcy1} that $\Theta_s=q_s\Id_n$ for $s\neq 0$, and from \eqref{Eq:Condcy0} that $\Theta_0=q_0\Id_n+q_0t\,\As^{(m)}\Cs^{(m)} Z_{m-1}^{-1}$.  Hence 
\begin{equation}
  \begin{aligned}
    \Theta_0Z_{m-1}=&q_0Z_{m-1}+q_0t \As^{(m)}\Cs^{(m)}=q_0t A^{-1}B+q_0t A^{-1}\As \Cs=q_0tA^{-1} \,(B+S) \,,\\
\Theta_s Z_{s-1}=&q_s Z_{s-1}=q_st_{s-1} B\,, \qquad s\neq0\,.
  \end{aligned}
\end{equation}

The first family contains the symmetric functions of the positions $(x_i)_i$ so we omit it. 
For the fourth family  in Theorem \ref{CycInvol}, we see that $(1+\eta\Theta)(Y+X^{-1})$ is constituted of $m$ blocks  $Z_{s-1}+\Theta_s Z_{s-1}$. 
Thus, the block $Z_{s-1}+\eta \Theta_s Z_{s-1}$ is given by $(1+\eta q_0)t A^{-1}B+\eta q_0t A^{-1}S$ for $s=0$ and $t_{s-1}(1+\eta q_s)B$ otherwise. We can rewrite 
$[(1+\eta \Theta)Z]^m$ as a matrix with diagonal blocks 
\begin{equation}
 \Big(t^2\prod_{s\neq 0} t_{s-1}(1+\eta q_s)\Big)\,\, 
B^s A^{-1}\left([t^{-1}+\eta']B+\eta' S\right) B^{m-s-1}\,, \quad 
 \text{ for } \eta'=q_0t^{-1} \, \eta\,\,.
\end{equation}
In other words, we are interested in studying the family 
\begin{equation} \label{Eq:Gmj}
  G^m_j:= \tr\Big[\,A^{-1}\left([t^{-1}+\eta']B+\eta' S\right) B^{m-1} \,\big]^j\,.
\end{equation}
 In particular, if we write 
 $G^m_j$ as a polynomial in $\eta'$ under the form $G^m_j=\sum_{l=0}^j (\eta')^l G^m_{j,l}$, we get that all the $(G^m_{j,l})_{j,l}$ are Poisson commuting (for fixed $m$) by Theorem \ref{CycInvol}. 
Now, remark that  \eqref{EqBspTad}
gives  $[t^{-1}+\eta']B+\eta' S=t^{-1}B+\eta' t^{-1} ABA^{-1}$. 
We can write 
\begin{equation}
\begin{aligned}
  G^m_{j,l}=&
\,\,t^{-j}\sum_{i_1,\ldots,i_{jm}=1}^n
 \prod_{a=1}^{jm}\frac{f_{i_ai_{a+1}} x_{i_{a+1}}}{x_{i_a}-t x_{i_{a+1}}}\,\,
 \sum_{\substack{I\subset \{0,\ldots,j-1\} \\ |I|=l}}
\left(\prod_{s\in I}x_{i_{sm+2}}^{-1} \right)
\left(\prod_{s\notin I}x_{i_{sm+1}}^{-1} \right)\,.
\end{aligned}
\end{equation}
This was obtained by developing $G^m_j:= \tr\Big[(t^{-1}A^{-1}B+t^{-1}\eta' B A^{-1}) B^{m-1}\big]^j$, which explains the two products at the end of the expression, that represent whether $A^{-1}$ occurs before or after $B$ in the $(sm+1)$-th factor $A^{-1}B+\eta' B A^{-1}$. In particular, $G_{j,j}=G_{j,0}$ for all $j=1,\ldots,n$.

We now look at the third family  in Theorem \ref{CycInvol}. We can see that for any $j\geqslant 1$  in $\Cnmp$
\begin{equation} \label{EqCyY}
\begin{aligned}
    &\tr \left((1+\eta\Theta)Y\right)^{jm}= \tr \left((1+\eta\Theta)(Z-X^{-1})\right)^{jm} \\
=&  m \, \tr \left((\Id_n+\eta\Theta_{0})(Z_{m-1}-X^{-1}_{m-1})\ldots(\Id_n+\eta\Theta_{2})(Z_{1}-X^{-1}_{1})
(\Id_n+\eta\Theta_{1})(Z_{0}-X^{-1}_{0})\right)^{j}\,.
\end{aligned}
\end{equation}
We get, for any $s\neq 0$, 
\begin{equation*} 
\begin{aligned}
(\Id_n+\eta\Theta_{s})(Z_{s-1}-X^{-1}_{s-1})=&(1+\eta q_s)(t_{s-1}B-\Id_n)=t_{s-1}(1+\eta q_s)(B-t_{s-1}^{-1}\Id_n)\,, \\
(\Id_n+\eta\Theta_{0})(Z_{m-1}-X^{-1}_{m-1})=&(\Id_n+\eta q_0tA^{-1}(B+S) (tA^{-1}B)^{-1})(tA^{-1}B-A^{-1})\\
=&tA^{-1}((1+\eta q_0)\Id_n+\eta q_0\,S B^{-1})(B-t^{-1}\Id_n)\,.
\end{aligned}
\end{equation*}
Hence, introducing $H^m_j=m^{-1}C^{-j}\tr \left((1+\eta\Theta)Y\right)^{jm}$ for $C=t \prod_{s\neq 0}t_{s-1}(1+\eta q_s)$, we get 
\begin{equation} \label{Eq:Hmj}
\begin{aligned}
   H^m_j=&\tr \left(A^{-1}((1+\eta q_0)\Id_n+\eta q_0\,S B^{-1})(B-t_{m-1}^{-1}\Id_n)\ldots (B-t_{0}^{-1}\Id_n)\right)^{j}  \\
=& \tr \left([(1+\eta q_0)\Id_n+\eta q_0\,S B^{-1}]\, \mathbf{P}(B) A^{-1}\right)^{j}\,,
\end{aligned}
\end{equation}
by setting 
\begin{equation} \label{PolH}
  \mathbf{P}(B) = \prod^{\longleftarrow}_{0\leqslant s \leqslant m-1} (B-t^{-1}_s \Id_n)\,.
\end{equation}
Again, we are interested in the different functions obtained by developing with respect to $\eta$, that is 
we write $H^m_j=\sum_{l=0}^j \eta^l H^m_{j,l}$. We can explicitly write down the elements for $l=0$ and get 
\begin{equation}
\begin{aligned}
   H^m_{j,0}=&
\sum_{i_1,\ldots,i_{jm}}\,
 \prod_{a=0}^{j-1} x_{i_{a(m+1)+1}}\prod_{s=1}^m\,\,
\left(\frac{x_{i_{am+s+1}}f_{i_{am+s}i_{am+s+1}}}{x_{i_{am+s}} -t x_{i_{am+s+1}}} -t^{-1}_{s-1} \frac{\delta_{(i_{am+s},i_{am+s+1})}}{x_{i_{am+s+1}}} \right)\,.
\end{aligned}
\end{equation}
As noticed for the family $(G^m_{j,l})_{j,l}$, the functions $H^m_{j,j}$ and $H^m_{j,0}$ are not independent. 
Using that $\Theta=XZX^{-1}Z^{-1}$, we can write $(1+\eta\Theta)(Z-X^{-1})$ as $(Z-X^{-1})+\eta XZ(Z-X^{-1})Z^{-1}X^{-1}$, so that in \eqref{EqCyY} it gives after expanding in $\eta$ 
\begin{equation} 
  \tr \left((1+\eta\Theta)Y\right)^{jm}= \tr \left((Z-X^{-1})\right)^{jm}+\ldots+ \eta^{jm} \tr \left(XZ(Z-X^{-1})Z^{-1}X^{-1}\right)^{jm}\,,
\end{equation}
thus the factors in $\eta^0$ and $\eta^{jm}$ agree, implying that $H^m_{j,j}$ and $H^m_{j,0}$ are multiples of each other, after normalisation by the constant $m^{-1}C^{-j}$ from above.

Let's remark two results about those families. First, in the limit $q_0 \to \infty$ where we fix the other $q_s$, all $t_s\to \infty$ and we can see from \eqref{PolH} that $H^m_{j,0}\to G^m_{j,0}$. So, by rescaling the  $H^m_{j,l}$, we can recover all the $(G^m_{j,l})_{j,l}$ in that limit. Though it is an alternative proof of their involution, the phase space is not defined in that limit.  Second, For a given $m$, each function $H^m_{j,l}$ can be written as a linear combination of $(G^{m'}_{j',l'})_{_{j',l'}}$ with smaller indices. If we allow the definition of $G^m_{1,0}=t^{-1}\tr B^mA^{-1}$ and $G^m_{1,1}=t^{-1}\tr A^{-1}SB^{m-1}$ for $m=0,1$, we get for example, 
\begin{equation}
  H_{1,1}^2=tt_0 \left(G_{1,0}^2+G_{1,1}^2 \right) - t\left(\frac{t_0}{t_1}+1\right)\left(G_{1,0}^1+G_{1,1}^1 \right)
+ \frac{t}{t_1} \left(G_{1,0}^0+G_{1,1}^0 \right)\,.
\end{equation}

Finally, we look at the second family $\tr \left((1+\eta\Theta^{-1})(1+XY)\right)^{j}$ in Theorem \ref{CycInvol}, for any $j\in \N$. We first remark that in $\Cnmo$, $(1+\eta\Theta^{-1})XZ=XZ+\eta ZX$. Meanwhile, $X_sZ_s$ is nothing else than $t_s B$, while for $s\neq m-1$ $Z_s X_s=t_s B$ but $Z_{m-1}X_{m-1}=t A^{-1}BA$. This gives us, if we call the elements $F^m_j$, 
\begin{equation}
  F^m_j=\sum_{s=0}^{m-1}  \tr(X_s Z_s+\eta Z_sX_s)^j =
\left(\sum_{s=0}^{m-2} t_s (1+\eta) \right) \tr(B)^j + t \tr(B+\eta A^{-1}BA)^j\,.
\end{equation}
It is just a family equivalent to $(G_j^1)_j$ with $G_j^1=\tr \left(B+\eta A^{-1}BA\right)^{j}$, corresponding to the spin RS system, see \cite{CF2}. Developing $F^m_j=\sum_{l=0}^j \eta^l  F^m_{j,l}$, we also get that $F^m_{j,0}$ and $F^m_{j,j}$ are proportional.


\subsection{Explicit flows}  \label{ss:Flows}
 We now show that we can get explicit expressions for the flows associated to particular elements of the families in Theorem \ref{CycInvol} in $\Cnm$. 
Computations for the results that we use now and the general philosophy behind them are gathered in Appendix \ref{Ann:Dyn}. 

Recall that the family $(G_k^m)_k$ is defined in $\Cnmo$ from the elements $\tr (U_\eta^k)$, where $U_\eta=Z(1+\eta \Theta)$ represents the element $z(1+\eta \phi)\in A$. We get from Lemma \ref{Lem:Dyz} and \eqref{relBrDyn} 
\begin{equation*}
\frac1k \brap{\tr U_\eta^k,X}= -\eta \Theta U_\eta^{k-1} ZX  -X  U_\eta^{k-1} Z\,, \quad 
 \frac1k \brap{\tr U_\eta^k,Z}= -Z U_\eta^{k-1} Z  +  U_\eta^{k-1} Z^2\,,
\end{equation*}
while the Poisson brackets with $V_\beta$ or $W_\beta$ vanish. It does not look possible to explicitly integrate most flows. Indeed, even the matrix $U_\eta$ is not a constant of motion under $\tr U_\eta^{k}$, though its symmetric functions certainly are.  However,  looking at order $0$ in $\eta$, we get $\tr Z^k$ which is $G_{k,0}^m$ up to a constant, and if we look at the flow defined by $d/dt_{k}=\frac1k \brap{\tr Z^k,-}$, we get the defining ODEs 
\begin{equation*}
  \frac{d X}{dt_{k}}=-XZ^k\,, \quad \frac{d Z}{dt_{k}}=0\,,\quad   \frac{d V_\beta}{dt_{k}}=0\,, \quad \frac{d W_\beta}{dt_{k}}=0\,,
\end{equation*}
which can be easily integrated. 
\begin{prop} \label{Pr:FloZ}
  Given the initial condition $(X(0),Z(0),V_\beta(0),W_\beta(0))$, the flow  at time $t_{k}$ defined by the Hamiltonian $\tr Z^k$ for $k\in m\N$ is given by 
\begin{equation*}
X(t_{k})= X(0)\, \exp(-t_{k} Z(0)^k)\,, \quad 
Z(t_{k})= Z(0)\,, \quad  V_\beta(t_{k})= V_\beta(0)\,, \quad W_\beta(t_{k})= W_\beta(0)\,. 
\end{equation*}
\end{prop}
In particular,  the flows are complete in $\Cnmo$ so that we can reintroduce $\As^{(m)}$ and $\Cs^{(m)}$ for all times, although some $X(t_k)$ could be non diagonalisable. Remark also that this expression does not exactly coincide with \cite[Proposition 4.7]{CF} when $d=1$. This is due to our different choice of ordering at the vertices $s\in I$, so that $\Cnm$ for $d=1$ is isomorphic to the space in \cite[Section 4]{CF} but this map is not the identity map. This is also true for the next flows.

For the other family $(H_k^m)_k$, expressed from $\tr \bar{U}_\eta^k$ if we set $\bar{U}_\eta=Y(1+\eta \Theta)$, 
Lemma \ref{Lem:Dyy} and \eqref{relBrDyn} give 
\begin{equation*}
\frac1k \brap{\tr \bar{U}_\eta^k,X}=-\bar{U}_\eta^{k-1}   -X  \bar{U}_\eta^{k-1} Y -\eta \Theta \bar{U}_\eta^{k-1} (1+YX)\,, \quad 
 \frac1k \brap{\tr \bar{U}_\eta^k,Y}= -Y \bar{U}_\eta^{k-1} Y  +  U_\eta^{k-1} Y^2\,,
\end{equation*}
and the Poisson brackets with $V_\beta$ or $W_\beta$ vanish. 
Again, we can write the flows for the functions which are the coefficients at order $0$ in $\eta$. 
If we want to obtain the flow of $\tr Y^k$ which is a multiple  $H^1_{k,0}$, we get by writing $d/d\tau_{k}=\frac1k \brap{\tr Y^k,-}$ that 
\begin{equation*}
  \frac{d X}{d\tau_{k}}=-Y^{k-1}-XY^k\,, \quad \frac{d Y}{d\tau_{k}}=0\,,\quad   \frac{d V_\beta}{d\tau_{k}}=0\,, \quad \frac{d W_\beta}{d\tau_{k}}=0\,.
\end{equation*}

\begin{prop}  \label{Pr:FloY}
  Given the initial condition $(X(0),Y(0),V_\beta(0),W_\beta(0))$ the flow at time $\tau_{k}$ defined by the Hamiltonian $\tr Y^k$ for $k\in m\N$ is given by 
\begin{equation*}
  \begin{aligned}
       X(\tau_{k})=& X(0)\, \exp(-\tau_{k} Y(0)^k) + Y(0)^{-1}[\exp(-\tau_{k} Y(0)^k)-\Id_n]\,, \\
Y(\tau_{k})=& Y(0)\,, \quad  W_\beta(\tau_{k})= W_\beta(0)\,, \quad V_\beta(\tau_{k})= V_\beta(0)\,.
  \end{aligned}
\end{equation*}
\end{prop} 
The expression for $X(\tau_{k})$ is analytic in $Y(0)$ so does not require its invertibility. The dynamics take place in $\Cnm$.

Reproducing this scheme for  the family $(F_k^m)_k$ using Lemma \ref{Lem:Dy1xy} and \eqref{relBrDyn}, this yields  for $\tilde{U}_\eta=(1+XY)(1+\eta \Theta^{-1})$
\begin{equation*}
\begin{aligned}
\frac1k \brap{\tr \tilde{U}_\eta^k,X}=&
- \tilde U_\eta^{k-1} (1+XY)X -\eta X \Theta^{-1} \tilde U_\eta^{k-1} (1+XY) \\
 \frac1k \brap{\tr \tilde{U}_\eta^k,1+XY}=& 
(1+XY) \tilde U_\eta^{K-1} (1+XY) - \tilde U_\eta^{K-1} (1+XY)^2  \,,
\end{aligned}
\end{equation*}
and the Poisson brackets with $V_\beta$ or $W_\beta$ vanish. It is easier to work with $1+XY$ instead of $Y$ in this case, because when we look at order $0$ in $\eta$ we obtain for the flow of $\tr (1+XY)^k$ after setting  $d/d \tilde t_{k}=\frac1k \brap{\tr (1+XY)^k,-}$ that 
\begin{equation*}
  \frac{d X}{d\tilde t_{k}}=-(1+XY)^k X\,, \quad \frac{d (1+XY)}{d\tilde t_{k}}=0\,,\quad   \frac{d V_\beta}{d\tilde t_{k}}=0\,, \quad \frac{d W_\beta}{d\tilde t_{k}}=0\,.
\end{equation*}
\begin{prop} \label{Pr:FloT}
  Given the initial condition $(X(0),Y(0),V_\beta(0),W_\beta(0))$, setting $T=1+XY$, the flow at time $\tilde t_{k}$ defined by the Hamiltonian $\tr T^k$ for $k\in \N$ is given by  
\begin{equation*}
X(\tilde t_{k})=  \exp(-\tilde t_{k} T(0)^k)\, X(0)\,, \quad 
T(\tilde t_{k})= T(0)\,, \quad  V_\beta(\tilde t_{k})= V_\beta(0)\,, \quad W_\beta(\tilde t_{k})= W_\beta(0)\,. 
\end{equation*}
In particular, assuming that $X(0)$ is invertible, this completely determines the solution $Y(\tilde t_{k})=X(\tilde t_k)^{-1}[T(\tilde t_k)-1]$ for all time $\tilde t_{k}$. 
\end{prop}
The extra assumption that $X(0)$ is invertible is satisfied in our interpretation of this family as being the one containing the spin RS Hamiltonian. In that case, the flows take place in $\Cnmo$ where they are complete. Using the isomorphism of Proposition \ref{isoCM}, we can see that they can be related the corresponding flows derived in \cite{CF2,RaS}. 

\medskip

A natural question to ask is to obtain locally the flows that we could not compute explicitly. As we see in \ref{ss:Red}, we can define integrable systems from these families in order to compute them by quadrature. However, this requires the additional assumption $d\leqslant n$, which we comment in  \ref{d2}. 


\subsection{Linear independence}  \label{ss:LinInd}

We assume from now on that $d\leqslant n$ and look at the number of independent functions in each family. 
Since the family $(F_{j,l}^m)_{j,l}$ corresponds to functions for the spin RS system, we already know that this contains $nd-d(d-1)/2$ linearly independent elements \cite{KrZ}.

 For the families $(G_{j,l}^m)_{j,l}$ and $(H_{j,l}^m)_{j,l}$, remark that there can be at most $n(n+1)/2$ linearly independent functions. Indeed, we look at the symmetric functions of $n \times n$ matrices (see \eqref{Eq:Gmj}, \eqref{Eq:Hmj}), so we have by developing in $\eta$, $n+n(n+1)/2$ functions for $j=1,\ldots,n$, with $n$ constraints coming from the relation between the terms $(j,0)$ and $(j,j)$. It now suffices to remark that we can write these functions generically in the form $\{\tr (C+\eta T)^j \mid 1 \leqslant j \leqslant n\}$ for $C$ invertible with distinct eigenvalues, and $T$ of rank $d$, with distinct nonzero eigenvalues. This yields $(n-d)(n-d+1)/2$ additional constraints, as we now explain with the family $(G_{j,l}^m)_{j,l}$. The other case is similar.  

\medskip

We adapt the method  introduced for the spin Calogero-Moser family \cite{BBT,KBBT}, and we write $G^m_k=\tr (C+\eta T)^k$ for $C=A^{-1}B^m$, $T=SB^{m-1}A^{-1}$ and $\eta=\eta'/(t^{-1}+\eta')$, see \eqref{Eq:Gmj}. Remark that $C,T$ take the form stated above. 
 Consider the spectral curve 
\begin{equation}
  \Gamma(\eta,\mu)\,\equiv\, \det((C+\eta T)-\mu \Id_n)=0\,,
\end{equation}
We write $\Gamma(\eta,\mu)\equiv \sum_{i=0}^n r_i(\eta)\mu^i=0$, where $r_i(\eta)$ is a symmetric function of order $n-i$ of $C+\eta T$, so is a function of 
$\{G^m_{k,K}|0\leqslant K\leqslant k,\,\,1\leqslant k \leqslant n-i\}$, if we consider the expansion\footnote{The expansion differs from the one in \ref{ss:CyIS}, but each of these two families can be obtained from the other one.} $G_k^m=\sum_{K=0}^k G^m_{k,K} \eta^k$. 
We can expand each $r_i(\eta)$ in terms of $\eta$ as $r_i(\eta)=\sum_{s=0}^{n-i} I_{n-i,s}\eta^s$. Hence the set of linearly independent functions is contained in the $\{I_{n-i,s}\}$, which are functions of the $n+n(n+1)/2$ functions $\{G^m_{k,K}\}$. 
It remains to see  how many relations exist on the $\{I_{n-i,s}\}$. In a neighbourhood of $\eta=\infty$, we can write 
\begin{equation}
  \Gamma(\eta,\mu)=\prod_{i=1}^n (\mu-\mu_i(\eta))\,, \quad \text{ for }\mu_i(\eta)=\eta \nu_i\,,
\end{equation}
for $(\nu_i)_i$ the eigenvalues of $T$. At a generic point, 
we can order the $(\nu_i)_i$ so that $\nu_1<\nu_2< \ldots < \nu_d$ are nonzero, and $\nu_{d+1}=\ldots=\nu_n=0$. Thus near $\eta=\infty$ we write $ \Gamma(\eta,\mu)=\mu^{n-d}\,\prod_{i=1}^d (\mu-\eta \nu_i)$. From this behaviour at infinity, we require that if we write  
$\Gamma(\eta,\mu)\equiv \sum_{i=0}^n \Gamma_i(\mu)\eta^i$, then $\Gamma_i(\mu)=0$ for all $i=d+1,\ldots,n$. Each $\Gamma_i(\mu)$ has order $n-i$ as a polynomial in $\mu$ whose coefficients are functions of the $\{I_{n-i,s}\}$. Hence we get the expansion $\Gamma_i(\mu)=\sum_{s=0}^{n-i} J_{n-i,s}\,\mu^s$, for some functions $J_{n-i,s}(I_{k,t})$. 
Their vanishing for $i>d$ is equivalent to imposing 
\begin{equation*}
\sum_{i=d+1}^n (n-i+1) 
=  \frac{(n-d)(n-d+1)}{2}
\end{equation*}
relations, which proves our claim. Imposing the initial $n$ relations, which are independent from the ones just obtained, we have a total of $nd-d(d-1)/2$ independent functions.


\subsection{Additional reduction}  \label{ss:Red}

We would like to construct a space of dimension $2nd-d(d-1)$ into which the different families descend. Introduce the $d(d-1)$-dimensional algebraic group 
 \begin{equation} \label{LiegpH}
  \mathcal H=\Big\{h=(h_{\alpha\beta})\in\Gl_d(\CC) \, \Big| \, \sum_{\beta=1}^d h_{\alpha \beta}=1\,\text{ for all }\alpha\Big\}\,,
 \end{equation}
 whose elements are invertible $d\times d$ matrices such that the vector $(1,\ldots,1)^\top$ is an eigenvector with eigenvalue $+1$. This is precisely the algebraic group $\mathcal H$ needed to get Liouville integrability of the RS system in the original work \cite{KrZ}.

Define the action of $\mathcal H$ on $(X,Z,\Asm,\Csm)$ by $h\cdot(X,Z,\Asm,\Csm)=(X,Z,\Asm h,h^{-1}\Csm)$. 
By definition of $\Cnmp$ at the end of \ref{ss:CyCoord}, we can always take a representative $(X,Z,\Asm,\Csm)$ on this subspace such that $\sum_\alpha (X \Asm)_{i \alpha}=1$ for all $i$. This condition is preserved under the action of $\mathcal H$.  Hence, we define  the reduced  space $\CnmG$ as the affine GIT quotient $\CnmG=\Cnmp /\!/ \mathcal H$. It has dimension $2nd-d(d-1)$ and is generically smooth as we will shortly see. The coordinate ring $\mathcal O (\CnmG)$ is generated by elements of the form $\tr \gamma$, where $\gamma$ is a word in the letters $X,Z,S=\Asm\Csm$. If we write such functions in coordinates by lifting them to $\Cnmp$, they become invariant polynomial in the elements $(x_i,x_i x_j^{-1}-t,f_{ij})_{ij}$, which form a Poisson subalgebra of $\br{-,-}$ by Proposition \ref{Proptadpole}. Thus the Poisson bracket $\brap{-,-}$ descends to $\CnmG$. It is such that the projection  $\Cnmp\to \CnmG$ dual to the inclusion $\mathcal O (\Cnmp^{\mathcal H}) \to \mathcal O (\Cnmp)$ is a Poisson morphism.

\begin{thm} \label{ThmIS}
  The families $\{F^m_{j,l}\mid (j,l)\in J_d\}$,  $\{G^m_{j,l}\mid (j,l)\in J_d\}$ and  $\{H^m_{j,l}\mid (j,l)\in J_d\}$ define completely integrable systems on the smooth part of $\CnmG$, for $J_d=\{(j,l) \mid j=1,\ldots,n,\, l=0,\ldots, \min(j-1,d)\}$. 
\end{thm}
\begin{proof}
  We show the existence of a non-empty open subset of $\Cnmp$ where $\mathcal H$ acts properly and freely in Lemmas \ref{lfree} and \ref{lproper}, so that the corresponding space of $\mathcal H$-orbits defines a smooth complex manifold of dimension $2nd-d(d-1)$ inside $\CnmG$. In particular, a point $(X,Z,\As,\Cs)$ in the subspace is  characterised by the fact that all the $d$-dimensional minors of $\As$ are nonzero. This is the complement of the Zariski closed subsets defined by having a vanishing minor of dimension $d$. Thus this subspace is dense in $\Cnmp$, and so does its reduction  in $\CnmG$. The elements in each family are $\mathcal H$-invariant, and also  linearly independent by the argument developed in \ref{ss:LinInd}. 
Thus the proof follows for the first two families. For the last family, remark that we also need to restrict to the open subset where $Y$ is invertible before performing the reduction, but this is dense again.  
\end{proof}

Using the isomorphism of Proposition \ref{isoCM}, a point $(X,Z,\Asm,\Csm)$ of $\Cnmp$ can be equivalently characterised by a quadruple $(A,B,\As,\Cs)$ satisfying \eqref{Tadiffeo}. By abuse of notation, we denote this point by $(X,Z,\As,\Cs)$ and assume that it has the form just stated. We remark that the $\mathcal H$-action is given by $h \cdot (X,Z,\As,\Cs)=(X,Z,\As h,h^{-1}\Cs)$ so that $\sum_\alpha (\As h)_{i \alpha}=1$ for all $i$.

\begin{lem} \label{lfree}
 The action is free on the subset of $\Cnmp$ where, given a point $(X,Z,\As,\Cs)$, either $\As$ or $\Cs$ has rank $d$. 
\end{lem}
\begin{proof}
 Assume $\As$ has rank $d$, the proof being the same if we assume the latter for $\Cs$. 
By definition, there exists $K=(k_1,\ldots,k_d)\subset \{1,\ldots,n\}$ such that $\bar{\As}=(\As_{k_\alpha \beta})$ 
is a $d\times d$ matrix which has rank $d$, so is invertible. If we take some $h$ in the stabiliser of the point 
$(X,Z,\As,\Cs)$, then in particular $\As h=\As$ and thus $\bar{\As}h=\bar{\As}$. Indeed, 
\begin{equation}
 (\bar{\As} h)_{\alpha\beta}=\sum_\gamma \As_{k_\alpha \gamma}h_{\gamma \beta}=(\As h)_{k_\alpha \beta}=
\As_{k_\alpha \beta}=\bar{\As}_{\alpha\beta}\,.
\end{equation}
Since $\bar{\As}$ is invertible, $h=\Id_d$. 
\end{proof}

\begin{lem} \label{lproper}
 The action is proper on the subset $\overline{\mathcal{S}} \subset \Cnmp$ where, given a point $(X,Z,\As,\Cs)$, 
all the minors of dimension $d$ of $\As$ are invertible. 
\end{lem}
\begin{proof}
The claim follows if we can show that given sequences $(h_m)\subset \mathcal H$, $(X_m,Z_m,\As_m,\Cs_m)\subset \overline{\mathcal{S}}$ 
satisfying $(X_m,Z_m,\As_m,\Cs_m)\to (X,Z,\As,\Cs) \in \overline{\mathcal{S}}$ and $h_m\cdot(X_m,Z_m,\As_m,\Cs_m)\to(X',Z',\As',\Cs') \in \overline{\mathcal{S}}$, then 
$h_m$ converges in $\mathcal H$. Note that trivially $X'=X$ and $Z'=Z$. 

For any choice of $K=(k_1,\ldots,k_d)\subset \{1,\ldots,n\}$, we can form $\bar{\As}$ as in Lemma \ref{lfree}. We also use the notation  $\bar{D}$ for the $d \times d$ matrix obtained in that way from some $n \times d$ matrix $D$. We see that  $h_m=\bar{\As}_m^{-1} \overline{h_m \cdot \As_m}$, since the minors of $\As_m$ are invertible and  $\overline{h_m \cdot \As_m}=\bar{\As}_m h_m$. 

From this, form $h:=\bar{\As}^{-1}\overline{\As'}$. This element does not depend on the choice of $K$ : take any two $K,L\subset \{1,\ldots,n\}$ and construct $\bar{\As}^{(K)}_m$ and $\bar{\As}^{(L)}_m$ for all $m$ as before, where the superscript denotes the partition to which we refer. They are both invertible, 
so they are related by $\bar{\As}^{(K)}_m=T_m\bar{\As}^{(L)}_m$ for some $T_m\in \Gl_d(\CC)$. 
Forming $h^{(K)}$ and $h^{(L)}$ from them, we get 
\begin{equation*}
 h^{(K)}=\lim_{m\to \infty} (\bar{\As}^{(K)}_m)^{-1}(h_m\cdot \bar{\As}^{(K)}_m)=
\lim_{m\to \infty} (\bar{\As}^{(L)}_m)^{-1}T_m^{-1} T_m(h_m\cdot \bar{\As}^{(L)}_m)=h^{(L)}\,. 
\end{equation*}
Next, remark that $h\in \Gl_d(\CC)$: as  $h=\bar{\As}^{-1}\overline{\As'}$ and both elements on the right hand-side have nonzero determinant, so too has $h$. 
Finally, $h\in \mathcal H$ because  
\begin{equation*}
 \sum_{\beta} h_{\alpha \beta}=\sum_{\gamma,\beta}  (\bar{\As}^{-1})_{\alpha \gamma}\overline{\As'}_{\gamma\beta}
=\sum_{\gamma}  (\bar{\As}^{-1})_{\alpha \gamma} 1 \,=\, 1\,.
\end{equation*}
Here we use that $\sum_\alpha \As_{i\alpha}=1$ for all $i$, implies that we have $\sum_\alpha \bar{\As}_{\gamma \alpha}=1$  for all $\gamma$. That is $\bar{\As}\in \mathcal H$, which in turn yields  $\bar{\As}^{-1}\in \mathcal H$. 
\end{proof}

We summarise the projection from the representation space of the multiplicative preprojective algebra to the space we have just constructed as 
\begin{equation} \label{ProjMap}
  \Rep(\Lambda^\qq,\aalpha) \longrightarrow  \Cnm \supsetneq \Cnmp \longrightarrow   \CnmG\,.
\end{equation}
Let us formulate one last comment on the reduced space $\CnmG$. We can integrate some equations of motions for the families in Theorem \ref{CycInvol}, thus defining flows in $\Cnm$. If flows quit the subspace $\Cnmp$, then the last projection given in \eqref{ProjMap} can not be defined, so the flows are not complete in $\CnmG$. This suggests that $\CnmG$ is not the natural phase space for our systems in the complex case, and it  motivates a search for other first integrals, see \ref{d2}.


\subsection{Final remarks} We finish by some additional comments that could lead to new investigations about these models.

\subsubsection{Integrability before reduction} \label{d2} 
In \ref{ss:Red}, we constructed the space $\CnmG$ as the complex analogue of the phase space for the real trigonometric spin RS system considered by Krichever and Zabrodin \cite{KrZ}. However, we noticed  in the complex setting that some flows are not complete inside $\CnmG$, but they are in the larger space $\Cnmo$ by Proposition \ref{Pr:FloT}. This suggests that we should be able to build an integrable system containing the Hamiltonian for the trigonometric spin RS system directly in $\Cnmo$. This is proved in \cite{CF2}, and we can in fact adapt the method to our case for the different Hamiltonians $\tr Z^{km}, \tr Y^{km}, \tr X^{km}$ and $\tr (1+XY)^k$, $k \in \N$. In the case $d=2$, it easily follows from the next result, which is a direct application of Lemma \ref{Lem:Cytt}. 
\begin{thm} \label{Thm:Cy2}
The following families of functions on $\Cnm$ are linearly independent and in involution   
\begin{equation*}
\begin{aligned}
& \left\{\tr X^{jm}, \tr \left(W_1 V_1 X^{jm} \right)\,\, \big|\,\, j=1,\ldots,n\right\}\,,\quad 
\left\{\tr (1+XY)^j, \tr \left(W_1 V_1 (1+XY)^j \right)\,\, \big|\,\, j=1,\ldots,n\right\}\,,\\
& \left\{\tr Y^{jm}, \tr \left(W_1 V_1 Y^{jm} \right)\,\, \big|\,\, j=1,\ldots,n \right\}\,,\quad
\left\{\tr Z^{jm}, \tr \left(W_1 V_1 Z^{jm} \right)\,\, \big|\,\, j=1,\ldots,n \right\}\,,
\end{aligned}
\end{equation*}
where the last family is viewed on the subspace $\Cnmo\subset \Cnm$ where $X$ is invertible.  
\end{thm}
For the case $d\geq 3$, the construction is more involved as we need more Poisson commuting functions, and we leave the details of adapting \cite[Section 5.2]{CF2} to the reader. Rather, we will look at another feature of these systems which is their degenerate integrability (also called non-commutative integrability or superintegrability), and was first remarked for the spin RS system by Reshetikhin in the real rational case \cite{Re}. 

Our method follows  \cite{CF2}. 
We only consider  $U=Y$ or $U=Z$, as they define new non-trivial Hamiltonians. We introduce the commutative algebra $\QU$ generated by the elements $\tr W_\alpha V_\beta U^{lm}$ for all $1\leq \alpha,\beta \leq d$, with $l\in \N$. 
\begin{lem}
  The algebra $\QU$ is a Poisson algebra under the Poisson bracket $\brap{-,-}$.
\end{lem}
\begin{proof}
  We show that $\brap{\tr W_\alpha V_\beta Y^{lm},\tr W_\gamma V_\epsilon Y^{km}}\in \QU$ which proves the case $U=Y$. We leave the similar case $U=Z$ to the reader. 

By inspecting the double brackets between the elements $(y,v_\alpha,w_\alpha)$ in Section \ref{ss:qHcyclic}, we see that 
the double bracket $\dgal{w_\alpha v_\beta y^{lm},w_\gamma v_\epsilon y^{km}}$ is such that its two components are (sums of)  words in  $w_\alpha,w_\gamma, v_\beta,v_\epsilon, y$.  
At the same time, this double bracket is an element of $e_0 A e_0 \otimes e_0 A e_0$, so its two components are in fact words in $w_{\mu}v_{\nu}$ and $y^m$, with $\mu \in \{\alpha,\gamma\}$, $\nu \in \{\beta,\epsilon\}$. 
Applying the multiplication map also yields a word $\rho\in e_0 A e_0$ written with the same letters. Moreover, a careful analysis of the double bracket shows that $\rho$ has to contain at least a factor $w_{\mu}v_{\nu}$. 
Using \eqref{relInv}, these remarks yield that   
\begin{equation*}
  \brap{\tr W_\alpha V_\beta Y^{lm},\tr W_\gamma V_\epsilon Y^{km}}
=\tr R\,, \quad \text{for some } R \in \CC[Y^m,W_\alpha V_\beta,W_\alpha V_\epsilon,W_\gamma V_\epsilon,W_\gamma V_\beta] \setminus \CC[Y^m]\,.
\end{equation*}
Now, the terms of $\tr R$ are of the form 
\begin{equation*}
\tr\Big((Y^m)^{a_1} W_{\mu_1} V_{\nu_1} (Y^m)^{a_2} W_{\mu_2} V_{\nu_2} \ldots 
W_{\mu_j}V_{\nu_j} (Y^m)^{a_j} \Big)
= (V_{\nu_1} Y^{a_2 m} W_{\mu_2}) \ldots (V_{\nu_j} Y^{(a_j+a_1)m} W_{\mu_1})\,.
\end{equation*}
Since $V_\nu Y^a W_\mu=\tr W_\mu V_\nu Y^a \in \QU$, any term of $\tr R$ is an element of $\QU$.
\end{proof}

Next, we remark by multiplying the identities  \eqref{Eq:Condcy0}-\eqref{Eq:Condcy1} that we can write on the subspace $\{\det U\neq 0\}$ of $\Cnm$ that 
\begin{equation} \label{Spect}
  MU^m M^{-1} =t  \prod_{\alpha=1,\ldots,d}^{\longleftarrow}(\Id_n+W_\alpha V_\alpha)\, U^m\,,
\end{equation}
with $M=X_{0}$ if $U=Z$ or $M=(X_{0}+Y_{0}^{-1})$ if $U=Y$. In particular, by taking traces of higher powers of this identity, we get that any $\tr U^{lm} \in \QU$ for $U=Y,Z$. Meanwhile, we have that the elements $(\tr U^{km})$ Poisson commute with any function $\tr W_\alpha V_\beta U^{kl}$ by Lemma \ref{Lem:Cytt} and  \eqref{relInv}. Thus, they are in the centre of the Poisson algebra $\QU$. 
\begin{prop}
  The algebra $\QU$ has dimension $2nd-n$ and its centre is generated by the elements $\tr U^m,\ldots,\tr U^{mn}$, so it has dimension $n$. 
\end{prop}
\begin{proof}
  A first method is to adapt the  case when $m=1$ given in \cite[Proposition 5.2]{CF2}. We sketch another possible proof when $U=Y$, based on a suitable choice of local coordinates similar to \cite[Lemma 5.6]{CF2}. 

Consider nonzero elements $y_1,\ldots,y_n\in \CC$ satisfying $y_i \neq y_j$, $y_i \neq t y_j$, for $i \neq j$. Consider also arbitrary $w_{\alpha,i},v_{\alpha,i}\in \CC$ for $i=1,\ldots,n$ and $1 \leq \alpha <d$. We denote by $\h'$ the subspace of $\CC^{2nd-n}$ with the above elements as coordinates, under the additional $d-1$ conditions  that $\sum_i v_{\alpha,i}w_{\alpha,i}\neq 0$. We then define the matrices 
$Y_s \in \Gl_n(\CC)$ for $s\in I$, and $W_\alpha\in \Mat_{n \times 1}(\CC)$,  $V_\alpha\in \Mat_{1 \times n}(\CC)$ for $1\leq \alpha < d$ by 
\begin{equation*}
 Y_0=\diag(y_1,\ldots,y_n),\,\, Y_s=\Id_n \text{ for }s\neq0,\quad 
(V_{\alpha})_i=v_{\alpha,i}, \quad (W_\alpha)_i=w_{\alpha,i}\,.
\end{equation*}
For a generic point of $\h'$, we can then find $V_d\in \Mat_{n \times 1}(\CC)$ such that, if $W_d:=(1, \ldots,1)^\top$, the matrices $Y_0$ and $F_d=t(\Id_n+W_d V_d)\ldots (\Id_n+W_1 V_1)Y_0$ have the same spectrum. Indeed, this is just a rank one perturbation of the matrix  $t(\Id_n+W_{d-1} V_{d-1})\ldots (\Id_n+W_1 V_1)Y_0$. In particular, there exists a $n$-dimensional family of matrices $M$ that put $F_d$ in the diagonal form $Y_0$. 

By construction, all the above matrices satisfy \eqref{Spect}. We set $M_0=M$ and define inductively $M_s=q_s Y_{s-1}M_{s-1}Y_s^{-1}$ for $s=1, \ldots, m-1$. Then we put  $X_s=M_s-Y_s^{-1}$ for all $s\in I$. It is then easy to see that the relations in \eqref{Eq:Condcy0} and \eqref{Eq:Condcy1} hold, and that all the invertibility conditions  required to define a point in $\{\det Y^m\neq 0\}\subset \Cnm$ are satisfied. Hence, we can complete the $2nd-n$ functions $(y_i,v_{\alpha,i},w_{\alpha,i})$ to get a local coordinate system around a generic point of $\Cnm$. 

Finally, note that in terms of these coordinates we can write that 
\begin{equation*}
  f_k\!=\!\tr Y^{km}\!=\!\sum_i y_i^k, \quad 
g_{k,\alpha}\!=\!\tr W_{d} Y^{km} V_{\alpha}\!=\!\sum_i y_i^k v_{\alpha,i}, \quad 
h_{k,\alpha}\!=\!\tr W_\alpha Y^{km} V_1\!=\!\sum_i y_i^k v_{1,i} w_{\alpha,i},
\end{equation*}
for $\alpha \neq d$. 
These functions  belong to $\QU$, and we can easily check that the subset 
$\{f_k,g_{k,\alpha},h_{k,\alpha}\mid k=1,\ldots,n,\, 1\leq \alpha <d\}$ is formed of functionally independent elements. 
\end{proof}

As a consequence, the functions $\tr U^m,\ldots,\tr U^{mn}$ are degenerately integrable. Their flows are complete by Propositions \ref{Pr:FloZ} and \ref{Pr:FloY}. 

We will write down a complete proof for both Liouville and degenerate integrability of the four cases $U=X,Y,Z,1+XY$  for an arbitrary framing of a cyclic quiver in subsequent work.

\subsubsection{Self-duality} 
The work of Reshetikhin \cite{Re} considers the duality between the spin hyperbolic CM system and the spin rational RS system. This was discovered in the non-spin case by Ruijsenaars \cite{R88}, together with the self-duality of the hyperbolic RS system. 
In the complex setting where the hyperbolic and trigonometric cases are the same, the latter self-duality can be obtained by noticing that, with the notations of Section \ref{tadpole} in the non-spin case $m=d=1$, the transformation $\varpi:(A,B)\mapsto (B,A)$ is an (anti-)symplectic mapping \cite[Proposition 3.8]{CF}. We can make a step in that direction for the spin case, though this requires the additional reduction of \ref{ss:Red}. Hence, we assume $d\leq n$.  

To work in full generalities, let $A$ be the algebra localised at $x$ constructed in Section \ref{ss:qHcyclic}. Consider the quasi-Hamiltonian algebra $\hat{A}$ obtained from $A$ by removing the elements $v_\alpha,w_\alpha,e_\infty$, i.e. $\hat{A}=A/\langle e_\infty \rangle$. This can be seen as the analogue of $A$  obtained by construction from the non-framed cyclic quiver, that is the subquiver $\bar{Q}'\subset \bar{Q}$ supported at $I=\Z/m\Z$. We can easily see  that the algebra homomorphism $\iota : \hat{A}\to \hat{A}$ defined by 
\begin{equation} \label{iota}
  \iota(e_s)=e_{m-s}\,, \quad \iota(x_s)=z_{m-s-1}\,, \quad \iota(z_s)=x_{m-s-1}\,,
\end{equation}
satisfies $\iota^2=\id_{\hat{A}}$. It corresponds to flipping $\bar{Q}'$ such that the vertex $0$ is fixed. Moreover, from \eqref{cyA}, with $z$ instead of $y$, and \eqref{cyBz} we can show that 
\begin{equation*}
(\iota \otimes \iota)  \dgal{x,x}=- \dgal{\iota(x),\iota(x)}, \quad 
(\iota \otimes \iota)  \dgal{z,z}=- \dgal{\iota(z),\iota(z)}, \quad 
(\iota \otimes \iota)  \dgal{x,z}=- \dgal{\iota(x),\iota(z)},
\end{equation*}
so that $\iota$ is an anti-morphism of quasi-Poisson algebras. (We can check this equivalently on $\dgal{x_s,x_r}$, $\dgal{z_s,z_r}$ and $\dgal{x_r,z_s}$.) In particular, $\iota(\phi)=\iota(xzx^{-1}z^{-1})=\phi^{-1}$, so $\iota$ maps the moment map of $\hat{A}$ to its inverse.  

Remark that from \ref{ss:Red} and \eqref{EqZspCyc}, the coordinate ring $\mathcal O (\CnmG)$ is generated by elements of the form $\tr(\Gamma)$, where $\Gamma$ is a sum of matrices whose factors are either $X$ or $Z$, so that the Poisson structure on  $\CnmG$ is completely defined from the double brackets 
$\dgal{x,x}$, $\dgal{z,z}$ and $\dgal{x,z}$. Indeed, they define the quasi-Poisson brackets in 
$\Rep(\Lambda^\qq,\aalpha)$ for the elements $(X_{ij},Z_{ij})$, which determine the Poisson bracket on $\CnmG$ by construction. This yields the following result. 
\begin{prop} \label{Propdual}
  The map $\iota: \hat{A}\to \hat{A}$ induces a (generically defined) anti-Poisson morphism $\varpi:\CnmG \to \Cnmdual$ determined by $(X,Z)\mapsto (Z,X)$, where $\hat{q}=(\hat{q}_\infty,\hat{q}_s)$ is defined by $\hat{q}_s=q^{-1}_{m-s}$, $\hat{q}_\infty=q_\infty^{-1}$. 
\end{prop}
In particular, since $\iota(z)=x$, $\iota(\phi)=\phi^{-1}$, we have that 
$\varpi((1+\eta \Phi)Z)=(1+\eta \Phi^{-1})X$. Hence, we have that the first and fourth families in Theorem \ref{CycInvol} are swapped under $\varpi$. 

As indicated in Proposition \ref{Propdual}, $\varpi$ is only defined at a generic point, e.g. in $\CnmG$ there are points where the product $Z_0 \ldots Z_{m-1}$ is not semisimple.  Hence, we do not have self-duality of the system in the strict sense of \cite{R88,FK09a,FK12,FA} which requires a global Poisson isomorphism. Nevertheless, the underlying interpretation on the quiver $\bar{Q}'$ is easy to understand, and works also for $m=1$, where it  extends the geometric approach to the self-duality for the trigonometric RS to the spin case. Given the natural appearance of the self-duality for the non-spin case in gauge theory \cite{FGNR,GR}, it would be interesting to understand the interpretation of the spin case within this framework. 

Let us formulate two final remarks. Firstly, if we replace $z$ by $y$ in the definition of $\iota$ \eqref{iota}, this also gives an anti-morphism of quasi-Poisson algebras. Hence the integrable system containing $\tr X$ is dual to the integrable system containing $\tr Z$, but also to the one containing $\tr Y$  by adapting the above argument. Secondly, note that this isomorphism $\iota: \hat{A}\to \hat{A}$ does not directly extend to $A$ itself. We will return to this issue in further work, in order to lift the map $\varpi$ to a well-defined map on $\Cnmo$.

\subsubsection{Relation to generalised symmetric group} It is remarked in \cite{CF} that in the case $d=1$, the new  Hamiltonians obtained for a cyclic quiver on $m$ vertices correspond to $W=S_n \ltimes\Z^n_m$. In the study of (non-multiplicative) quiver varieties, the Hamiltonians of Calogero-Moser type obtained in the spin case have also that particular symmetry \cite[Section VI]{CS}. Thus, we would expect that the elements of the families $(G^m_{j,l})_{j,l}$ and  $(H^m_{j,l})_{j,l}$ are also related to the generalised symmetric groups $G(m,1,n)=S_n \ltimes\Z^n_m$. 
To establish this link, consider $p\in\Cnmp \subset \Cnmo$ determined by a point of $\Cntd'$ as in Proposition \ref{isoCM}. Using the local coordinates of $\Cntd'$ \eqref{Tadiffeo}, the point $p=(X_s,Z_s,\Asm,\Csm)$ is characterised by $X_s=\Id_n$ for $s\neq m-1$, $X_{m-1}=\diag(x_1,\ldots,x_n)$, and the matrices $(Z_s,\Asm,\Csm)$ are given by 
\begin{equation*} 
  (Z_s)_{ij}=t_s \frac{t f_{ij}}{x_i x_j^{-1}-t},\, s\neq m-1,\quad 
  (Z_{m-1})_{ij}=t \frac{t x_i^{-1} f_{ij}}{x_i x_j^{-1}-t}, \quad 
(\Asm)_{i \alpha}=\frac{\aaa_i^\alpha}{x_i},\,\, (\Csm)_{\alpha j}=\ccc_j^\alpha\,.
\end{equation*}
In particular, the matrices $(X_s)_s$ take two different forms : either the identity or a diagonal matrix whose entries are interpreted as particle positions. To set them all to the same diagonal matrix, recall that an element $g \in G(\aalpha)$ acts as in \eqref{gactCy}, so we can write this action as 
\begin{equation}\label{gactCy2}
g. (X_s,Z_s,\Asm,\Csm)=(g_sX_sg_{s+1}^{-1},g_{s+1}Z_sg_s^{-1}, \Csm g_0^{-1}, g_0 \Asm)\,,\quad g\in G(\aalpha)\,. 
\end{equation}
Choose elements $(\lambda_i)_i$ such that $\lambda_i^m=x_i$.  They are nonzero distinct and satisfy $\lambda_i^m \lambda_j^{-m}\neq t$ for all $i\neq j$. We can form the element $\prod_s g_s\in G(\aalpha)$ with  $g_s=\diag(\lambda_1^{m-s},\ldots,\lambda_n^{m-s})$ for $s=0,\ldots,m-1$,  and acting on $p$ in its above form yields 
\begin{equation*} 
X_s=\diag(\lambda_1,\ldots,\lambda_n),\,\,
  (Z_s)_{ij}=t_s \frac{t f_{ij}}{\lambda^m_i -t \lambda_j^m}\lambda_i^{m-s-1}\lambda_j^{s},\quad 
(\Asm)_{i \alpha}=\aaa_i^\alpha,\,\, (\Csm)_{\alpha j}=\frac{\ccc_j^\alpha}{\lambda^m_j}\,,
\end{equation*}
for any $s=0,\ldots,m-1$. Hence, the choice of a representative $(X_s,Z_s,\Asm,\Csm)$ in $\Cnmp$, such that all the $X_s$ are in the same diagonal form and $\sum_\alpha (\Asm)_{i \alpha}=1$,   is unique up to acting by  $S_n \ltimes\Z^n_m$. Here, the action of an element $(\sigma,\underline{M}) \in S_n \ltimes\Z^n_m$, is represented by the matrix $g=\prod_s g_\sigma g_{\underline{M}}^{-s}$, where $g_\sigma$ is the permutation matrix corresponding to $\sigma$ while $g_{\underline{M}}=\diag(\zeta^{M_1},\ldots,\zeta^{M_n})$ for $\underline{M}=(M_1,\ldots,M_n)$ and $\zeta$ is a primitive $m$-th root of unity. 

In the case $m=2$, write  $q_0=e^{-2 \gamma_0}$ and $q_1=e^{-2 \gamma_1}$ so that $t=e^{-2 \gamma}$ for $\gamma=\gamma_0+\gamma_1$. We get
\begin{equation*}
 \begin{aligned}
(Z_0)_{ij}=&\frac{e^{-2\gamma-2\gamma_0}}{2}\left(\frac{1}{\lambda_i- e^{-\gamma} \lambda_j}+\frac{1}{\lambda_i+ e^{-\gamma} \lambda_j} \right)\, f_{ij}\,, \\
(Z_1)_{ij}=&\frac{e^{-3\gamma}}{2}\left(\frac{1}{\lambda_i- e^{-\gamma} \lambda_j}-\frac{1}{\lambda_i+ e^{-\gamma} \lambda_j} \right)\, f_{ij}\,\,.
 \end{aligned}
\end{equation*}
We can write down  $\tr Z^2$ and $\tr Y^2$, which are multiples of $G^2_{1,0}$ and $H^2_{1,0}$ respectively, in $\Cnmp$ as 
 \begin{equation*}
 \begin{aligned}
     \tr Z^2 
 =&\frac{e^{-5\gamma-2\gamma_0}}{2}\sum_{i,j}\left(\frac{1}{\lambda_i- e^{-\gamma} \lambda_j}+\frac{1}{\lambda_i+ e^{-\gamma} \lambda_j} \right)  \left(\frac{1}{\lambda_j- e^{-\gamma} \lambda_i}-\frac{1}{\lambda_j+ e^{-\gamma} \lambda_i} \right)\, f_{ij}f_{ji}\,, \\
 \tr Y^2 =& \tr Z^2 -\sum_i \frac{e^{-2\gamma-2\gamma_0}+e^{-4\gamma}}{1-e^{-2\gamma}} \frac{f_{ii}}{\lambda_i^2}+\sum_i \frac{1}{\lambda_i^2}\,.
 \end{aligned}
 \end{equation*}
Comparing last two expressions with $\tr B^2$ and $\tr (B-A^{-1})^2$ obtained from Section \ref{tadpole} strengthens our claim that the $(G^m_{j,l})_{j,l}$ correspond to a spin RS system for $W=S_n \ltimes\Z^n_m$  (and $(H^m_{j,l})_{j,l}$ to a modification of it). 


\appendix

\section{Calculations for the spin cyclic quivers} \label{Ann:cyclic}

 We prove 
 Lemma \ref{Lem:Calg1} in \ref{Ann:B1},  
Lemma \ref{LodCyxac} in \ref{Ann:B2},  
Lemma  \ref{Lem:Cytt} in \ref{Ann:B2bis}, 
and Proposition \ref{PropInvPhiCy} in \ref{Ann:B3}. 
Most computations rely heavily on the properties of a double bracket given in \ref{ss:dAS}. 
\begin{rem}
  Note that the proofs of Lemmas \ref{Lem:Calg1}, \ref{LodCyxac} and Proposition \ref{PropInvPhiCy} also apply for their analogues in the case $m=1$ considered in \cite{CF2}. We make a comment on the changes that are needed in the latter case at the beginning of each of these proofs. The elements $a'_\alpha,c'_\alpha$ are denoted by $a_\alpha,b_\alpha$ in  \cite{CF2}. 
\end{rem}

\subsection{Double bracket with spin variables : Proof of Lemma \ref{Lem:Calg1}} \label{Ann:B1}

We will show that \eqref{CySpin1} and \eqref{CySpin3} holds, while we replace \eqref{CySpin2} by 
\begin{equation}
 \begin{aligned}
    \dgal{a'_{\alpha}, c'_{\beta}}\,=\,& -\frac12 c'_\beta a'_\alpha \otimes e_0 +
\frac12  \left(o(\alpha,\beta)-\delta_{\alpha \beta} \right)\, 
e_\infty\otimes a'_{\alpha} c'_{\beta} \\
&- \delta_{\alpha \beta} \left(e_\infty \otimes e_0 z + \sum_{\lambda=1}^{\beta-1} 
e_\infty \otimes a'_{\lambda} c'_{\lambda} \right) \,. \label{CySpin2bis}
 \end{aligned}
\end{equation}
This is nothing else than \eqref{CySpin2} because the first term vanishes. Indeed, note that $c'_\beta=e_\infty c'_\beta e_{m-1}$, so that $c'_\beta \gamma=0$ for any $\gamma$ which is a path beginning by $x,z$ or some $a'_\alpha=w_\alpha$ since then $\gamma=e_0\gamma$. However, we will carry on such terms of the form $c'_\beta \gamma$, because our proof also applies in the case of a Jordan quiver (see \cite[Lemma 3.1]{CF2}) where it does not vanish. 
Indeed, if we allow the case $m=1$ and set $F_b=e_0\otimes e_0$ for any $b\in \Z$, The double brackets between the elements $x,z,v_\alpha,w_\alpha$ given in Section \ref{ss:qHcyclic} exactly match the double brackets in \cite[\S 3.1.2]{CF2}.

\medskip

We prove the results by induction using \eqref{cyInd}. Knowing the double brackets in Section \ref{ss:qHcyclic},  if we want to compute 
the bracket $\dgal{\Gamma, c'_\beta}$ for some $\Gamma\in A$, we first find $\dgal{\Gamma,c'_1}=\dgal{\Gamma,v_1 z}$ and 
then show our statement by induction using 
\begin{equation} \label{Induction}
 \dgal{\Gamma,c'_\alpha}=
\sum_{\lambda=1}^{\alpha-1} \left( v_\alpha w_\lambda \dgal{\Gamma,  c'_\lambda} 
+ \dgal{\Gamma, v_\alpha w_\lambda } c'_\lambda \right) + \dgal{\Gamma,v_\alpha z} \,.
\end{equation}

\medskip

To get \eqref{CySpin1}, we first compute $\dgal{x, c'_{\alpha}}$ and show how to deal with the idempotents. Recalling that $F_1=\sum_s e_s \otimes e_{s-1}$, we get from the double brackets in Section \ref{ss:qHcyclic} 
\begin{equation*}
 \begin{aligned}
  \dgal{x, v_{\alpha}z}=& \dgal{x, v_{\alpha}}z +  v_{\alpha} \dgal{x, z} \\
=& \frac12\,\left( v_\alpha  x\otimes e_0 z - v_\alpha \otimes x e_0 z
+v_\alpha  zx F_1 +v_\alpha  F_1 xz - v_\alpha  x F_1 z +v_\alpha  z F_1 x\right) \\
=& \frac12\,\big( v_\alpha  x e_{1}\otimes e_0 z - v_\alpha e_0\otimes e_{m-1}x z +v_\alpha  zx e_0 \otimes e_{m-1}\\ 
& \quad +v_\alpha  e_0 \otimes e_{m-1} xz - v_\alpha  x e_{1}\otimes e_0 z +v_\alpha  z e_{m-1} \otimes e_{m-2} x\big) \\
=& 
 \frac12\,\big( v_\alpha  zx \otimes e_{m-1}  +v_\alpha  z  \otimes  x e_{m-1}\big)\,.
 \end{aligned}
\end{equation*}
In order to simplify the $F_1$ to go from the second to the third equality, we used that $v_\alpha=v_\alpha e_0$, $x\in \oplus_{s\in I} e_s A e_{s+1}$ and $z \in \oplus_{s\in I} e_{s+1} A e_s$. For example, $v_\alpha zx F_1=v_{\alpha}zx e_0 F_1=v_{\alpha}zx e_0\otimes e_{m-1}$. 
In particular, if we use $c'_1=v_1z$ we get the expression for $\dgal{x, c'_1}$ given in \eqref{CySpin1} as our basis for the induction. 

Meanwhile, we compute 
\begin{equation*}
 \begin{aligned}
  \dgal{x,v_\alpha w_\lambda}=v_\alpha \dgal{x,w_\lambda}+\dgal{x,v_\alpha} w_\lambda
=\frac{1}{2} (v_\alpha\otimes xw_\lambda- v_\alpha x\otimes w_\lambda)
+\frac12 (v_\alpha x\otimes w_\lambda- v_\alpha\otimes xw_\lambda ) =0\,,
 \end{aligned}
\end{equation*}
so that if we assume that the first equality in \eqref{CySpin1} is true for any $\lambda< \alpha$, we get from \eqref{Induction} 
\begin{equation*}
 \begin{aligned}
  \dgal{x,c'_\alpha}=&
\sum_{\lambda=1}^{\alpha-1} \left( v_\alpha w_\lambda \dgal{x,  c'_\lambda} 
+ \dgal{x, v_\alpha w_\lambda } c'_\lambda \right) + \dgal{x,v_\alpha z} \\
=&\sum_{\lambda=1}^{\alpha-1} v_\alpha w_\lambda  \left( 
\frac12 c'_\lambda x\otimes e_{m-1}+\frac12 c'_\lambda \otimes x e_{m-1}
  \right) + \frac{1}{2} (v_\alpha zx\otimes e_{m-1}  + v_\alpha z\otimes x e_{m-1})\\
=&\frac12 \left(\sum_{\lambda=1}^{\alpha-1} v_\alpha w_\lambda c'_\lambda+v_\alpha z \right) x \otimes e_{m-1}
+ \frac12 \left(\sum_{\lambda=1}^{\alpha-1}v_\alpha w_\lambda c'_\lambda+v_\alpha z \right) \otimes x e_{m-1}\,,
 \end{aligned}
\end{equation*}
which is exactly the first equality in \eqref{CySpin1} by  using  \eqref{cyInd}.  For the bracket $\dgal{z, c'_{\alpha}}$, we  compute 
\begin{equation*}
 \begin{aligned}
  \dgal{z, v_{\alpha}z}=& \dgal{z, v_{\alpha}}z +  v_{\alpha} \dgal{z, z} \\
=& \frac12\,\left( v_\alpha  z\otimes e_0 z - v_\alpha \otimes z e_0 z
-v_\alpha  z^2 F_{-1} +v_\alpha  F_{-1} z^2 \right) \\
=& \frac12\,\big( v_\alpha  z e_{1}\otimes e_0 z - v_\alpha e_0\otimes e_{1}z^2 -v_\alpha  z^2 e_{m-2} \otimes e_{m-1} +v_\alpha  e_0 \otimes e_{1} z^2 \big) \\
=& 
 \frac12\,\big(- v_\alpha  z^2 \otimes e_{m-1}  +v_\alpha  z  \otimes  z e_{m-1}\big)\,,
 \end{aligned}
\end{equation*}
and $\dgal{z,c_1'}$ given by the second equality in \eqref{CySpin1} holds. We can find $\dgal{z,v_\alpha w_\lambda}=0$ so that the general case follows by induction, in a way similar to $\dgal{x,c_\alpha'}$.

\medskip

To get \eqref{CySpin2bis}, recall that $a'_\alpha=w_\alpha$ by definition. We first compute 
\begin{equation*}
 \begin{aligned}
  \dgal{v_\beta z,a'_\alpha}=&v_\beta \ast \dgal{z,a'_\alpha}+\dgal{v_\beta , a'_\alpha} \ast z\\
=&\frac12 ( e_{0}\otimes v_\beta za'_\alpha- e_0 z \otimes v_\beta a'_\alpha)+
\left[\delta_{\alpha \beta}  \, e_0 z\otimes e_\infty
+ \frac12 \,[o(\beta,\alpha)+\delta_{\alpha \beta}] 
\left(e_0 z\otimes v_\beta a'_\alpha + a'_\alpha v_\beta z \otimes e_\infty \right) \right] \,.
 \end{aligned}
\end{equation*}

Using $\dgal{a'_\alpha,v_\beta z}=-\dgal{v_\beta z,a'_\alpha}^\circ$, we can write 
\begin{equation*}
 \begin{aligned}
  \dgal{a'_\alpha,v_\beta z}
=&\frac12 ( v_\beta a'_\alpha \otimes e_0 z - v_\beta za'_\alpha \otimes e_0)
-\delta_{\alpha \beta}  \,  e_\infty\otimes e_0 z
+ \frac12 \,[o(\alpha,\beta)-\delta_{\alpha \beta}] 
\left(v_\beta a'_\alpha \otimes e_0 z + e_\infty \otimes a'_\alpha v_\beta z  \right)\\
=&-\delta_{(\alpha \geqslant \beta)} \left(\frac12 e_\infty \otimes a'_\alpha v_\beta z+\frac12 v_\beta za'_\alpha \otimes e_0+\delta_{\alpha \beta}  \,  e_\infty\otimes e_0 z\right)  \\
&+ \delta_{(\alpha < \beta)}
\left( v_\beta a'_\alpha \otimes e_0 z - \frac12 v_\beta za'_\alpha \otimes e_0 +\frac12 e_\infty \otimes a'_\alpha v_\beta z  \right)\,,
 \end{aligned}
\end{equation*} 
recalling that $o(\alpha,\beta)=\delta_{(\alpha < \beta)}-\delta_{(\alpha > \beta)}$. 
In particular, this yields 
\begin{equation*}
   \dgal{a'_\alpha,c'_1}
=-\frac12 c'_1 a'_\alpha \otimes e_0-\frac12 e_\infty \otimes a'_\alpha c'_1-\delta_{\alpha \beta}  \,  e_\infty\otimes e_0 z
\end{equation*}
which is exactly the case $\beta=1$ in \eqref{CySpin2bis} (and $c'_1 a'_\alpha=0$ as we mentioned at the beginning of the proof to get in fact \eqref{CySpin2}). 
Next, we can compute 
\begin{equation*}
 \begin{aligned}
  \dgal{a'_\alpha,v_\beta w_\lambda}
=&v_\beta \dgal{a'_\alpha, w_\lambda}+  \dgal{a'_\alpha,v_\beta } w_\lambda \\
=& -\frac12 \,o(\alpha,\lambda) \left(v_\beta w_\lambda\otimes w_\alpha + v_\beta w_\alpha \otimes w_\lambda \right) \\
&\,-\delta_{\alpha \beta}  \,  e_\infty\otimes w_\lambda
+ \frac12 \,[o(\alpha,\beta)-\delta_{\alpha \beta}] 
\left(v_\beta a'_\alpha \otimes w_\lambda + e_\infty \otimes a'_\alpha v_\beta w_\lambda  \right)\,,
 \end{aligned}
\end{equation*}
and this implies that 
\begin{equation*}
 \begin{aligned}
 \sum_{\lambda=1}^{\beta-1} \dgal{a'_\alpha,v_\beta w_\lambda} c'_\lambda
=& -\frac12 \sum_{\lambda=1}^{\beta-1} \,o(\alpha,\lambda) 
\left(v_\beta w_\lambda\otimes w_\alpha c'_\lambda + v_\beta w_\alpha \otimes w_\lambda c'_\lambda \right) 
-\delta_{\alpha \beta}  \,\sum_{\lambda=1}^{\beta-1}  e_\infty\otimes w_\lambda c'_\lambda \\
&+ \frac12 \,[o(\alpha,\beta)-\delta_{\alpha \beta}]  \sum_{\lambda=1}^{\beta-1}
\left(v_\beta a'_\alpha \otimes w_\lambda c'_\lambda + e_\infty \otimes a'_\alpha v_\beta w_\lambda  c'_\lambda \right)\,.
 \end{aligned}
\end{equation*}
In the case $\alpha\geqslant \beta$ this gives since $w_\alpha=a'_\alpha$
\begin{equation*}
 \begin{aligned}
 \sum_{\lambda=1}^{\beta-1} \dgal{a'_\alpha,v_\beta w_\lambda} c'_\lambda
\stackrel{\alpha\geqslant \beta}{\,=\,}& \,
-\delta_{\alpha \beta}  \,\sum_{\lambda=1}^{\beta-1}  e_\infty\otimes w_\lambda c'_\lambda
+\frac12 \sum_{\lambda=1}^{\beta-1}  \left(v_\beta w_\lambda\otimes w_\alpha c'_\lambda 
- e_\infty \otimes a'_\alpha v_\beta w_\lambda  c'_\lambda \right)\,.
 \end{aligned}
\end{equation*}
Otherwise, we just write 
\begin{equation*}
 \begin{aligned}
 \sum_{\lambda=1}^{\beta-1} \dgal{a'_\alpha,v_\beta w_\lambda} c'_\lambda
\stackrel{\alpha <  \beta}{\,=\,}& \,
 -\frac12 \left(\sum_{\lambda=\alpha+1}^{\beta-1} - \sum_{\lambda=1}^{\alpha-1} \right) \, 
\left(v_\beta w_\lambda\otimes w_\alpha c'_\lambda + v_\beta w_\alpha \otimes w_\lambda c'_\lambda \right) \\
&+ \frac12 \, \sum_{\lambda=1}^{\beta-1}\left(v_\beta a'_\alpha \otimes w_\lambda c'_\lambda 
+ e_\infty \otimes a'_\alpha v_\beta w_\lambda  c'_\lambda \right)
 \end{aligned}
\end{equation*}
Now, assume by induction that \eqref{CySpin2bis} holds for any $\lambda<\beta$, which we can write
\begin{equation*}
 \begin{aligned}
  \dgal{a'_\alpha, c'_\lambda}\,\stackrel{\alpha\geqslant \lambda}{\,=\,}\,& 
-\frac12 c'_\lambda a'_\alpha \otimes e_0
-\frac12  e_\infty \otimes a'_\alpha c'_\lambda
-\delta_{\alpha \lambda} \left(e_\infty \otimes e_0 z + \sum_{\gamma=1}^{\lambda-1} 
e_\infty \otimes a'_\gamma c'_\gamma \right)\,,  \\
\dgal{a'_\alpha, c'_\lambda}\,\stackrel{\alpha <  \lambda}{\,=\,}\,& 
-\frac12 c'_\lambda a'_\alpha \otimes e_0+\frac12  e_\infty \otimes a'_\alpha c'_\lambda \,,
 \end{aligned}
\end{equation*}
In the first case, $\alpha \geq \beta$, we find from \eqref{Induction} and \eqref{cyInd} 
\begin{equation*}
 \begin{aligned}
 \dgal{a'_\alpha, c'_\beta}\,\stackrel{\alpha\geqslant \beta}{\,=\,}\,& 
-\frac12 \sum_{\lambda=1}^{\beta-1}  v_\beta w_\lambda  c'_\lambda a'_\alpha \otimes e_0
-\frac12 \sum_{\lambda=1}^{\beta-1}  v_\beta w_\lambda   e_\infty \otimes a'_\alpha c'_\lambda \\
&-\delta_{\alpha \beta}  \,\sum_{\lambda=1}^{\beta-1}  e_\infty\otimes w_\lambda c'_\lambda
+\frac12 \sum_{\lambda=1}^{\beta-1}  \left(v_\beta w_\lambda\otimes w_\alpha c'_\lambda 
- e_\infty \otimes a'_\alpha v_\beta w_\lambda  c'_\lambda \right) \\
&-\left(\frac12 e_\infty \otimes a'_\alpha v_\beta z+\frac12 v_\beta za'_\alpha \otimes e_0+\delta_{\alpha \beta}  \,  e_\infty\otimes e_0 z\right) \\
=&-\frac12 c'_\beta a'_\alpha \otimes e_0 -\frac12 e_\infty \otimes a'_\alpha c'_\beta 
-\delta_{\alpha \beta} e_\infty \otimes \left(\sum_{\lambda=1}^{\beta-1} a'_\lambda c'_\lambda+ e_0 z \right) 
 \end{aligned}
\end{equation*}
which coincide with \eqref{CySpin2bis}. In the second case, we get 
\begin{equation*}
 \begin{aligned}
 \dgal{a'_\alpha, c'_\beta}\,\stackrel{\alpha< \beta}{\,=\,}\,& 
\sum_{\lambda=1}^{\alpha} 
\Big[-\frac12 v_\beta w_\lambda c'_\lambda a'_\alpha \otimes e_0
-\frac12  v_\beta w_\lambda  \otimes a'_\alpha c'_\lambda  \Big]
- v_\beta w_\alpha  \otimes e_0 z - \sum_{\gamma=1}^{\alpha-1} 
v_\beta w_\alpha \otimes a'_\gamma c'_\gamma  \\
&\,+ \sum_{\lambda=\alpha+1}^{\beta-1} \left( -\frac12 v_\beta w_\lambda c'_\lambda a'_\alpha \otimes e_0
+\frac12 v_\beta w_\lambda \otimes a'_\alpha c'_\lambda\right)
 \\
& -\frac12 \left(\sum_{\lambda=\alpha+1}^{\beta-1} - \sum_{\lambda=1}^{\alpha-1} \right) \, 
\left(v_\beta w_\lambda\otimes w_\alpha c'_\lambda + v_\beta w_\alpha \otimes w_\lambda c'_\lambda \right)\\
&+ \frac12 \, \sum_{\lambda=1}^{\beta-1}\left(v_\beta a'_\alpha \otimes w_\lambda c'_\lambda 
+ e_\infty \otimes a'_\alpha v_\beta w_\lambda  c'_\lambda \right)
 \\
&+\left( v_\beta a'_\alpha \otimes e_0 z - \frac12 v_\beta za'_\alpha \otimes e_0 +\frac12 e_\infty \otimes a'_\alpha v_\beta z  \right)
 \end{aligned}
\end{equation*}
which, after some easy manipulations on the sums, yields 
\begin{equation*}
 \begin{aligned}
 \dgal{a'_\alpha, c'_\beta}\,\stackrel{\alpha< \beta}{\,=\,}\,& 
-\frac12 \sum_{\lambda=1}^{\beta-1}  v_\beta w_\lambda c'_\lambda a'_\alpha \otimes e_0
+\frac12 \, \sum_{\lambda=1}^{\beta-1}  e_\infty \otimes a'_\alpha v_\beta w_\lambda  c'_\lambda
- \frac12 v_\beta za'_\alpha \otimes e_0 +\frac12 e_\infty \otimes a'_\alpha v_\beta z\\
=&- \frac12 c'_\beta a'_\alpha \otimes e_0+ \frac12 e_\infty \otimes a'_\alpha c'_\beta 
 \end{aligned}
\end{equation*}
as expected from \eqref{CySpin2bis} since the first term is zero. 

\medskip

As an intermediate result for \eqref{CySpin3}, we need 
\begin{lem} \label{Calg3}
 For any  $\alpha,\beta=1,\ldots,d$, 
 \begin{equation*}
  \dgal{v_{\alpha}, c'_{\beta}}\,=\,
\frac12 c'_\beta e_0\otimes v_\alpha 
- \frac12  \left(o(\alpha,\beta)+\delta_{\alpha \beta} \right) v_\alpha\otimes c'_\beta
 \end{equation*}
\end{lem}
\begin{proof}
Note that the first term vanishes, but we keep it for the case $m=1$ as explained at the beginning of the proof. 
 We compute 
\begin{equation*}
\begin{aligned}
  \dgal{v_\alpha, v_\beta z}=&\dgal{v_\alpha, v_\beta } z + v_\beta  \dgal{v_\alpha, z} \\
=&-\frac12 \,o(\alpha,\beta) \left(v_\beta\otimes v_\alpha z + v_\alpha \otimes v_\beta z \right)
-\frac12  v_\beta \otimes v_\alpha z+\frac12 v_\beta z e_0 \otimes v_\alpha\,. 
\end{aligned}
\end{equation*} 
We keep the last vanishing term for computations. In particular, we get 
\begin{equation*}
 \dgal{v_\alpha,c'_1}\,\stackrel{\alpha>1}{\,=\,}\,\frac12 v_\alpha \otimes c_1' + \frac12 c_1' e_0 \otimes v_\alpha\,, \quad \dgal{v_1,c_1'}\,=\,-\frac12 v_1 \otimes c'_1 +\frac12 c_1' e_0\otimes v_1\,,  
\end{equation*}
which agrees with our statement for $\beta=1$. Now, we compute 
 \begin{equation*}
\begin{aligned}
\dgal{v_\alpha, v_\beta w_\lambda}=&\dgal{v_\alpha, v_\beta }w_\lambda+ v_\beta \dgal{v_\alpha, w_\lambda} \\
=&-\frac12 \,o(\alpha,\beta) \left(v_\beta\otimes v_\alpha w_\lambda 
+ v_\alpha \otimes v_\beta w_\lambda \right)
+ \delta_{\alpha \lambda}  v_\beta \otimes e_\infty \\
&  + \frac12 [o(\alpha,\lambda)+\delta_{\alpha \lambda}] 
 \left(v_\beta \otimes v_\alpha w_\lambda + v_\beta w_\lambda v_\alpha \otimes e_\infty \right)\,.
\end{aligned}
\end{equation*}
Assume by induction that for all $\lambda<\beta$, 
 \begin{equation*}
  \dgal{v_\alpha, c'_\lambda}\,=\,\frac12 c'_\lambda e_0 \otimes v_\alpha 
- \frac12  \left(o(\alpha,\lambda)+\delta_{\alpha \lambda} \right) v_\alpha\otimes c'_\lambda \,,
 \end{equation*}
then we get by \eqref{Induction}
 \begin{equation*}
\begin{aligned}
\dgal{v_\alpha, c'_\beta}=&
\frac12 \sum_{\lambda=1}^{\beta-1} v_\beta w_\lambda  c'_\lambda e_0 \otimes v_\alpha 
-\frac12 \,o(\alpha,\beta) \sum_{\lambda=1}^{\beta-1}
 \left(v_\beta\otimes v_\alpha w_\lambda  c'_\lambda+ v_\alpha \otimes v_\beta w_\lambda c'_\lambda \right)\\ 
&+ \delta_{(\alpha<\beta)}  v_\beta \otimes c'_\alpha
+\frac12 \sum_{\lambda=1}^{\beta-1}  [o(\alpha,\lambda)+\delta_{\alpha \lambda}] 
 v_\beta \otimes v_\alpha w_\lambda c'_\lambda \\
&-\frac12 \,o(\alpha,\beta) \left(v_\beta\otimes v_\alpha z + v_\alpha \otimes v_\beta z \right)
-\frac12  v_\beta \otimes v_\alpha z+\frac12 v_\beta z e_0 \otimes v_\alpha \,.
\end{aligned}
\end{equation*}
In the case $\alpha>\beta$ we find 
 \begin{equation*}
\begin{aligned}
\dgal{v_\alpha, c'_\beta}\,\stackrel{\alpha> \beta}{\,=\,}\,&
\frac12 \sum_{\lambda=1}^{\beta-1} v_\beta w_\lambda  c'_\lambda e_0 \otimes v_\alpha 
+\frac12 v_\beta z e_0 \otimes v_\alpha 
+\frac12  \sum_{\lambda=1}^{\beta-1}   v_\alpha \otimes v_\beta w_\lambda c'_\lambda 
+\frac12   v_\alpha \otimes v_\beta z \\
=&\, \frac12 c'_\beta e_0 \otimes v_\alpha + \frac12   v_\alpha\otimes c'_\beta\,.
\end{aligned}
\end{equation*}
In the case $\alpha=\beta$  we have 
 \begin{equation*}
\begin{aligned}
\dgal{v_\alpha, c'_\beta}\,\stackrel{\alpha= \beta}{\,=\,}\,&
\frac12 \sum_{\lambda=1}^{\beta-1} v_\beta w_\lambda  c'_\lambda e_0 \otimes v_\alpha 
+\frac12 v_\beta z e_0 \otimes v_\alpha
-\frac12 \sum_{\lambda=1}^{\beta-1}   
 v_\beta \otimes v_\alpha w_\lambda c'_\lambda 
-\frac12  v_\beta \otimes v_\alpha z \\
=&\, \frac12 c'_\beta e_0 \otimes v_\alpha - \frac12   v_\alpha\otimes c'_\beta\,.
\end{aligned}
\end{equation*}
Finally, for $\alpha<\beta$ we get 
 \begin{equation*}
\begin{aligned}
\dgal{v_\alpha, c'_\beta}\,\stackrel{\alpha< \beta}{\,=\,}\,&
\frac12 \sum_{\lambda=1}^{\beta-1} v_\beta w_\lambda  c'_\lambda e_0 \otimes v_\alpha 
-\frac12  \sum_{\lambda=1}^{\beta-1}
 \left(v_\beta\otimes v_\alpha w_\lambda  c'_\lambda+ v_\alpha \otimes v_\beta w_\lambda c'_\lambda \right)   + v_\beta \otimes c'_\alpha \\
&+\frac12 \left[ \sum_{\lambda=\alpha}^{\beta-1} - \sum_{\lambda=1}^{\alpha-1} \right] 
 v_\beta \otimes v_\alpha w_\lambda c'_\lambda 
-\frac12  \left(v_\beta\otimes v_\alpha z + v_\alpha \otimes v_\beta z \right)
-\frac12  v_\beta \otimes v_\alpha z+\frac12 v_\beta z e_0 \otimes v_\alpha \\
=& \frac12 \sum_{\lambda=1}^{\beta-1} v_\beta w_\lambda  c'_\lambda \otimes v_\alpha 
+\frac12 v_\beta z \otimes v_\alpha 
-\frac12  \sum_{\lambda=1}^{\beta-1}  v_\alpha \otimes v_\beta w_\lambda c'_\lambda - 
\frac12 v_\alpha \otimes v_\beta z \\
&+ v_\beta \otimes c'_\alpha -  \sum_{\lambda=1}^{\alpha-1}
v_\beta\otimes v_\alpha w_\lambda  c'_\lambda - v_\beta \otimes v_\alpha z\,,
\end{aligned}
\end{equation*}
which is exactly $\frac12 c'_\beta e_0\otimes v_\alpha - \frac12   v_\alpha\otimes c'_\beta$ 
since the two last terms can be written as $-v_\beta \otimes c'_\alpha$. 
\end{proof}

We can finish the proof of Lemma \ref{Lem:Calg1} by showing  \eqref{CySpin3}. It is easier to  use the induction in the first variable, that is 
 \begin{equation} \label{Induction2}
 \dgal{c'_\alpha,\Gamma}=
\sum_{\lambda=1}^{\alpha-1} \left( v_\alpha w_\lambda \ast  \dgal{ c'_\lambda, \Gamma} 
+ \dgal{ v_\alpha w_\lambda ,\Gamma} \ast  c'_\lambda \right) + \dgal{v_\alpha z,\Gamma} 
\end{equation}
with $\Gamma=c'_\beta$ in our case. By doing so, we can repeatedly use \eqref{CySpin1}, \eqref{CySpin2bis} and 
Lemma \ref{Calg3}. 
We first get 
\begin{equation*}
 \begin{aligned}
  \dgal{v_\alpha z,c'_\beta}=& \dgal{v_\alpha,c'_\beta}\ast z+ v_\alpha \ast \dgal{ z,c'_\beta} \\
=&\frac12 c'_\beta  e_0 z \otimes v_\alpha 
- \frac12  \left(o(\alpha,\beta)+\delta_{\alpha \beta} \right) v_\alpha z\otimes c'_\beta 
-\frac12 c'_\beta z\otimes v_\alpha e_{m-1} +\frac12 c'_\beta\otimes v_\alpha z \\
=&- \frac12  \left(o(\alpha,\beta)+\delta_{\alpha \beta} \right) v_\alpha z\otimes c'_\beta 
+\frac12 c'_\beta\otimes v_\alpha z \,,
 \end{aligned}
\end{equation*}
using that the first and third terms vanishes (they would cancel out in the Jordan quiver case). This  gives in particular $\dgal{c'_1,c'_\beta}=-\frac12 c'_1\otimes c'_\beta + \frac12 c'_\beta \otimes c'_1$. Now using  \eqref{CySpin2bis} as $w_\lambda=a'_\lambda$, 
we can compute 
\begin{equation*}
 \begin{aligned}
  \dgal{v_\alpha w_\lambda,c'_\beta}=& \dgal{v_\alpha,c'_\beta}\ast w_\lambda+ v_\alpha \ast \dgal{ w_\lambda,c'_\beta} \\
=&
\frac12 c'_\beta e_0 w_\lambda\otimes v_\alpha  
- \frac12  \left(o(\alpha,\beta)+\delta_{\alpha \beta} \right) v_\alpha w_\lambda\otimes c'_\beta
-\frac12 c'_\beta w_\lambda \otimes v_\alpha \\
& +\frac12 \left(o(\lambda,\beta)-\delta_{\lambda \beta} \right) e_\infty \otimes v_\alpha w_\lambda c'_\beta
 -\delta_{\lambda \beta} \left(e_\infty \otimes v_\alpha z + \sum_{\gamma=1}^{\beta-1} 
 e_\infty \otimes v_\alpha w_\gamma c'_\gamma \right)\,.
 \end{aligned}
\end{equation*}
The first and third terms cancel out, so we can write 
\begin{equation*}
 \begin{aligned}
 \sum_{\lambda=1}^{\alpha-1} \dgal{v_\alpha w_\lambda,c'_\beta} \ast c'_\lambda 
=& 
- \frac12  \left(o(\alpha,\beta)+\delta_{\alpha \beta} \right) \sum_{\lambda=1}^{\alpha-1} v_\alpha w_\lambda c'_\lambda\otimes c'_\beta
+\frac12 \sum_{\lambda=1}^{\alpha-1} \left(o(\lambda,\beta)-\delta_{\lambda \beta} \right) c'_\lambda \otimes v_\alpha w_\lambda c'_\beta \\
& -\delta_{(\beta<\alpha)} c'_\beta \otimes v_\alpha z -\delta_{(\beta<\alpha)} \sum_{\gamma=1}^{\beta-1} 
 c'_\beta \otimes v_\alpha w_\gamma c'_\gamma \,.
 \end{aligned}
\end{equation*}

Now, assume by induction that for all $\lambda<\alpha$, 
\begin{equation*}
 \dgal{c'_\lambda,c'_\beta}\,=\, \frac12 [o(\lambda,\beta)+\delta_{\lambda\beta}] 
\left(c'_\beta \otimes c'_\lambda - c'_\lambda \otimes c'_\beta \right)\,,
\end{equation*}
and let us show that this holds for $\lambda=\alpha$. Note that it is exactly \eqref{CySpin3} since in the case $\lambda=\beta$ the two terms cancel out. We find by \eqref{Induction2} and our previous computations 
\begin{equation*}
 \begin{aligned}
  \dgal{c'_\alpha,c'_\beta} 
=& \frac12 \sum_{\lambda=1}^{\alpha-1} [o(\lambda,\beta)+\delta_{\lambda\beta}] 
\left(c'_\beta \otimes v_\alpha w_\lambda c'_\lambda - c'_\lambda \otimes v_\alpha w_\lambda c'_\beta \right) \\
&- \frac12  \left(o(\alpha,\beta)+\delta_{\alpha \beta} \right) \sum_{\lambda=1}^{\alpha-1} v_\alpha w_\lambda c'_\lambda\otimes c'_\beta
+\frac12 \sum_{\lambda=1}^{\alpha-1} \left(o(\lambda,\beta)-\delta_{\lambda \beta} \right) c'_\lambda \otimes v_\alpha w_\lambda c'_\beta \\
& -\delta_{(\beta<\alpha)} c'_\beta \otimes v_\alpha z -\delta_{(\beta<\alpha)} \sum_{\gamma=1}^{\beta-1} 
 c'_\beta \otimes v_\alpha w_\gamma c'_\gamma 
- \frac12  \left(o(\alpha,\beta)+\delta_{\alpha \beta} \right) v_\alpha z\otimes c'_\beta 
+\frac12 c'_\beta\otimes v_\alpha z \,.
 \end{aligned}
\end{equation*}
If $\alpha > \beta$ we find 
\begin{equation*}
 \begin{aligned}
  \dgal{c'_\alpha,c'_\beta} 
\,\stackrel{\alpha >  \beta}{\,=\,}\,& 
\frac12 \left( \sum_{\lambda=1}^{\beta}- \sum_{\lambda=\beta+1}^{\alpha-1}\right)  
\left(c'_\beta \otimes v_\alpha w_\lambda c'_\lambda - c'_\lambda \otimes v_\alpha w_\lambda c'_\beta \right) + \frac12   \sum_{\lambda=1}^{\alpha-1} v_\alpha w_\lambda c'_\lambda\otimes c'_\beta \\
&
+\frac12 \left(\sum_{\lambda=1}^{\beta-1} - \sum_{\lambda=\beta}^{\alpha-1} \right) c'_\lambda \otimes v_\alpha w_\lambda c'_\beta 
 - \sum_{\lambda=1}^{\beta-1} 
 c'_\beta \otimes v_\alpha w_\lambda c'_\lambda 
+ \frac12   v_\alpha z\otimes c'_\beta 
-\frac12 c'_\beta\otimes v_\alpha z \\
=&-\frac12 \sum_{\lambda=1}^{\alpha-1}  c'_\beta \otimes v_\alpha w_\lambda c'_\lambda 
-\frac12 c'_\beta\otimes v_\alpha z 
+ \frac12  \sum_{\lambda=1}^{\alpha-1} v_\alpha w_\lambda c'_\lambda\otimes c'_\beta 
+ \frac12   v_\alpha z\otimes c'_\beta \,,
 \end{aligned}
\end{equation*}
which gives us $-\frac12 (c'_\beta \otimes c'_\alpha - c'_\alpha \otimes c'_\beta )$. In the other cases, 
\begin{equation*}
 \begin{aligned}
  \dgal{c'_\alpha,c'_\beta} 
\,\stackrel{\alpha \leqslant   \beta}{\,=\,}\,&  \frac12 \sum_{\lambda=1}^{\alpha-1} 
\left(c'_\beta \otimes v_\alpha w_\lambda c'_\lambda - c'_\lambda \otimes v_\alpha w_\lambda c'_\beta \right) 
- \frac12   \sum_{\lambda=1}^{\alpha-1} v_\alpha w_\lambda c'_\lambda\otimes c'_\beta \\
&
+\frac12 \sum_{\lambda=1}^{\alpha-1}  c'_\lambda \otimes v_\alpha w_\lambda c'_\beta 
- \frac12   v_\alpha z\otimes c'_\beta 
+\frac12 c'_\beta\otimes v_\alpha z \,,
 \end{aligned}
\end{equation*}
and this is trivially $+\frac12 (c'_\beta \otimes c'_\alpha - c'_\alpha \otimes c'_\beta )$.  \qed

\subsection{Proof of Lemma \ref{LodCyxac}} \label{Ann:B2}
This proof can be applied without change in the case $m=1$ treated in \cite{CF2}.  Indeed, we use \eqref{CySpin2bis} instead of \eqref{CySpin2} when computing $\{\!\{a'_\gamma c'_\epsilon ,a'_\alpha c'_\beta\}\!\}$ below. In that way,  all the double brackets that we use during this proof are the ones in \cite{CF2} if we set $m=1$. 

\medskip

  We assume that the integers $k,l\geq1$ satisfy the conditions for the elements to be nonzero, otherwise the proof is trivial. The  first equality is an easy computations, or can be obtained as a consequence of \cite[Lemma A.3]{CF}. Next, we compute  from \eqref{cyD} and \eqref{CySpin1} 
\begin{equation} \label{LodCy1}
 \dgal{x,a'_\alpha c'_\beta}=\frac12 e_{0}\otimes xa'_\alpha c'_\beta-\frac12 e_0 x\otimes a'_\alpha c'_\beta 
+\frac12 a'_\alpha c'_\beta x\otimes e_{m-1}+\frac12 a'_\alpha c'_\beta\otimes x e_{m-1} \,.
\end{equation}
Combining this result with \eqref{cyA}, $S_1:=\{\!\{x^k,a'_\alpha c'_\beta x^l\}\!\}$  becomes 
\begin{equation}
 \begin{aligned} \label{LodCy2}
S_1=&
\sum_{\sigma=1}^k x^{\sigma-1} \ast \dgal{x,a'_\alpha c'_\beta} x^l \ast x^{k-\sigma} 
+\sum_{\sigma=1}^k \sum_{\tau=1}^l x^{\sigma-1} \ast a'_\alpha c'_\beta x^{\tau-1} \dgal{x,x} x^{l-\tau} \ast x^{k-\sigma} \\
=& \frac12 \sum_{\sigma=1}^k 
\Big( e_0 x^{k-\sigma} \otimes x^{\sigma} a'_\alpha c'_\beta x^l  - e_0 x^{k-\sigma+1}  \otimes x^{\sigma-1} a'_\alpha c'_\beta  x^l  \\
&\qquad \qquad  + a'_\alpha c'_\beta x^{k-\sigma+1} \otimes x^{\sigma-1}e_{m-1}x^l  + a'_\alpha c'_\beta x^{k-\sigma} \otimes x^{\sigma}e_{m-1}x^l \Big) \\
&+\frac{1}{2} \sum_{\sigma=1}^k \sum_{\tau=1}^l \sum_{s\in I}
\big(a'_\alpha c'_\beta x^{\tau+1} e_s x^{k-\sigma} \otimes x^{\sigma-1} e_{s-1} x^{l-\tau} \\
& \qquad \qquad \qquad \qquad  - a'_\alpha c'_\beta x^{\tau-1}e_s x^{k-\sigma} \otimes x^{\sigma-1}e_{s-1}x^{l-\tau+2} \big) \,.
 \end{aligned}
\end{equation}
If we apply the multiplication $m$, we have in the last two terms that the nonvanishing terms are for $s\in I$ such that $e_{m-1} x^{\tau+1}e_s=x^{\tau+1}$ and $e_{m-1} x^{\tau-1}e_s=x^{\tau-1}$ respectively, and we get that 
only the third and fourth terms do not cancel out. We find that  $\br{x^k,a'_\alpha c'_\beta x^l}= k \,a'_\alpha c'_\beta x^{k}e_{m-1}x^{l}=k \,a'_\alpha c'_\beta x^{k+l}$ since $x^k e_{m-1}=e_{m-1} x^k$ by assumption on $k$. 
Next, 
\begin{equation*}
 \begin{aligned} 
  \dgal{a'_\gamma c'_\epsilon x^k,a'_\alpha c'_\beta x^l}=&
a'_\gamma c'_\epsilon  \ast \dgal{x^k,a'_\alpha c'_\beta x^l}
+ a'_\alpha c'_\beta \dgal{a'_\gamma c'_\epsilon , x^l} \ast x^k
+\dgal{a'_\gamma c'_\epsilon ,a'_\alpha c'_\beta } x^l \ast x^k\,.
 \end{aligned}
\end{equation*}
From \eqref{LodCy1} and \eqref{LodCy2} we can get for the first two terms 
\begin{equation*}
 \begin{aligned}
S_2:=&a'_\gamma c'_\epsilon  \ast \dgal{x^k,a'_\alpha c'_\beta x^l}
- \sum_{\tau=1}^l a'_\alpha c'_\beta x^{\tau-1} \dgal{x,a'_\gamma c'_\epsilon}^\circ x^{l-\tau} \ast x^k \\
=&\quad  \frac12 \sum_{\sigma=1}^k 
\Big( e_0 x^{k-\sigma} \otimes a'_\gamma c'_\epsilon x^{\sigma} a'_\alpha c'_\beta x^l  - e_0 x^{k-\sigma+1}  \otimes a'_\gamma c'_\epsilon x^{\sigma-1} a'_\alpha c'_\beta  x^l  \\
&\qquad \qquad  + a'_\alpha c'_\beta x^{k-\sigma+1} \otimes a'_\gamma c'_\epsilon x^{\sigma-1}e_{m-1}x^l  + a'_\alpha c'_\beta x^{k-\sigma} \otimes a'_\gamma c'_\epsilon x^{\sigma}e_{m-1}x^l \Big) \\
&+\frac{1}{2} \sum_{\sigma=1}^k \sum_{\tau=1}^l \sum_{s\in I}
\big(a'_\alpha c'_\beta x^{\tau+1} e_s x^{k-\sigma} \otimes a'_\gamma c'_\epsilon x^{\sigma-1} e_{s-1} x^{l-\tau} \\
& \qquad \qquad \qquad \qquad  - a'_\alpha c'_\beta x^{\tau-1}e_s x^{k-\sigma} \otimes a'_\gamma c'_\epsilon x^{\sigma-1}e_{s-1}x^{l-\tau+2} \big)
 \\
&+\frac12\sum_{\tau=1}^l \left(-a'_\alpha c'_\beta x^{\tau}a'_\gamma c'_\epsilon x^k \otimes e_0 x^{l-\tau}
+a'_\alpha c'_\beta  x^{\tau-1} a'_\gamma c'_\epsilon x^k \otimes e_0 x^{l-\tau+1}
\right) \\
&+\frac12 \sum_{\tau=1}^l \left(
- a'_\alpha c'_\beta  x^{\tau-1} e_{m-1} x^k \otimes a'_\gamma c'_\epsilon x^{l-\tau+1}
- a'_\alpha c'_\beta  x^{\tau} e_{m-1} x^k \otimes a'_\gamma c'_\epsilon x^{l-t}
\right)\,.
 \end{aligned}
\end{equation*}
Applying the multiplication map and relabelling indices yields 
\begin{equation*}
 \begin{aligned}
S_2=&\quad  \frac12 
\left[ \sum_{\sigma=1}^k - \sum_{\sigma=0}^{k-1} \right]   e_0 x^{k-\sigma}  a'_\gamma c'_\epsilon x^{\sigma} a'_\alpha c'_\beta x^l     +\frac12 \left[ \sum_{\sigma=0}^{k-1} + \sum_{\sigma=1}^{k} \right]  a'_\alpha c'_\beta x^{k-\sigma}  a'_\gamma c'_\epsilon x^{\sigma}e_{m-1}x^l  \\
&+\frac{1}{2} \sum_{\sigma=1}^k \sum_{\tau=1}^l \sum_{s\in I}
\big(a'_\alpha c'_\beta x^{\tau+1} e_s x^{k-\sigma}  a'_\gamma c'_\epsilon x^{\sigma-1} e_{s-1} x^{l-\tau} \\
& \qquad \qquad \qquad \qquad  - a'_\alpha c'_\beta x^{\tau-1}e_s x^{k-\sigma}  a'_\gamma c'_\epsilon x^{\sigma-1}e_{s-1}x^{l-\tau+2} \big)
 \\
&+\frac12\left[-\sum_{\tau=1}^l+\sum_{\tau=0}^{l-1} \right] a'_\alpha c'_\beta x^{\tau}a'_\gamma c'_\epsilon x^k  e_0 x^{l-\tau}
-\frac12 \left[ \sum_{\tau=0}^{l-1} + \sum_{\tau=1}^l \right]  a'_\alpha c'_\beta  x^{\tau} e_{m-1} x^k \otimes a'_\gamma c'_\epsilon x^{l-t}\,.
 \end{aligned}
\end{equation*}
By assumption, $l,k=1$ modulo $m$, so that $e_s x^k=x^k e_{s+1}$ and $e_s x^k=x^k e_{s+1}$ for any $s\in I$. 
Hence $a'_\alpha c'_\beta x^le_0=a'_\alpha c'_\beta x^l$ and $e_{m-1}x^l a'_\alpha c'_\beta=x^l a'_\alpha c'_\beta$, so we can drop the idempotents in the first line modulo commutators. Similarly, this can be done in the last line. 
For the first term in the middle line, $a'_\alpha c'_\beta x^{\tau+1} e_s$ gives $s=\tau$ mod $m$ while $e_{s-1} x^{l-\tau} a'_\alpha c'_\beta$ gives $s-1+(1-\tau)=0$ mod $m$, which is the same condition. Thus we can drop the idempotent corresponding to $s=\tau$ modulo commutators. For the second term, we get $s=\tau-2$ in the same way and we can drop the idempotent.  So we can write 
\begin{equation*}
 \begin{aligned}
S_2=&\quad  \frac12 
\left[ \sum_{\sigma=1}^k - \sum_{\sigma=0}^{k-1} \right]    x^{k-\sigma}  a'_\gamma c'_\epsilon x^{\sigma} a'_\alpha c'_\beta x^l     +\frac12 \left[ \sum_{\sigma=0}^{k-1} + \sum_{\sigma=1}^{k} \right]  a'_\alpha c'_\beta x^{k-\sigma}  a'_\gamma c'_\epsilon x^{l+\sigma}  \\
&+\frac{1}{2} \sum_{\sigma=1}^k \sum_{\tau=1}^l 
\big(a'_\alpha c'_\beta  x^{k-\sigma+\tau+1}  a'_\gamma c'_\epsilon   x^{l-\tau+\sigma-1}  - a'_\alpha c'_\beta x^{k-\sigma+\tau-1}  a'_\gamma c'_\epsilon x^{l-\tau+\sigma+1} \big)
 \\
&+\frac12\left[-\sum_{\tau=1}^l+\sum_{\tau=0}^{l-1} \right] a'_\alpha c'_\beta x^{\tau}a'_\gamma c'_\epsilon    x^{k+l-\tau}
-\frac12 \left[ \sum_{\tau=0}^{l-1} + \sum_{\tau=1}^l \right]  a'_\alpha c'_\beta  x^{k+\tau} a'_\gamma c'_\epsilon x^{l-t}\,.
 \end{aligned}
\end{equation*}
Note that the middle line can be decomposed as 
\begin{equation*}
\frac12 \left[ \sum_{\sigma=0}^{k-1} \sum_{\tau=l}+ \sum_{\sigma=0} \sum_{\tau=1}^{l-1}- \sum_{\sigma=1}^{k-1} \sum_{\tau=0} - \sum_{\sigma=k} \sum_{\tau=0}^{l-1}  \right]
a'_\alpha c'_\beta x^{k-\sigma+\tau}   a'_\gamma c'_\epsilon x^{l-\tau+\sigma}  \,.
\end{equation*}
We can then rewrite $S_2$ as 
\begin{equation*}
 \begin{aligned}
S_2=&
\frac12 a'_\gamma c'_\epsilon x^{k} a'_\alpha c'_\beta x^l - \frac12 x^{k}   a'_\gamma c'_\epsilon a'_\alpha c'_\beta x^l 
+\frac12 a'_\alpha c'_\beta x^{k}   a'_\gamma c'_\epsilon x^{l} + \frac12 a'_\alpha c'_\beta a'_\gamma c'_\epsilon x^{k+l} 
+\sum_{\sigma=1}^{k-1}   a'_\alpha c'_\beta x^{k-\sigma}   a'_\gamma c'_\epsilon x^{l+\sigma} \\
&+\frac{1}{2}  \sum_{\sigma=0}^{k-1}  a'_\alpha c'_\beta x^{k+l-\sigma}   a'_\gamma c'_\epsilon x^{\sigma} 
+\frac{1}{2} \sum_{\tau=1}^{l-1} a'_\alpha c'_\beta x^{k+\tau}   a'_\gamma c'_\epsilon x^{l-\tau} \\
&-\frac{1}{2} \sum_{\sigma=1}^{k-1}  a'_\alpha c'_\beta x^{k-\sigma}   a'_\gamma c'_\epsilon x^{l+\sigma} 
-\frac{1}{2} \sum_{\tau=0}^{l-1} a'_\alpha c'_\beta x^{\tau}   a'_\gamma c'_\epsilon x^{k+l-\tau} \\
&+\frac12 a'_\alpha c'_\beta a'_\gamma c'_\epsilon x^{k+l} - \frac12
a'_\alpha c'_\beta x^{l}a'_\gamma c'_\epsilon x^k   
-\frac12 a'_\alpha c'_\beta  x^{k}   a'_\gamma c'_\epsilon x^{l}-\frac12 a'_\alpha c'_\beta  x^{k+l}   a'_\gamma c'_\epsilon 
-  \sum_{\tau=1}^{l-1}  a'_\alpha c'_\beta  x^{k+\tau}   a'_\gamma c'_\epsilon x^{l-\tau}
 \end{aligned}
\end{equation*}
We can cancel terms together (some of them modulo commutators) to obtain 
\begin{equation*}
 \begin{aligned}
S_2=&
 +\frac12 a'_\alpha c'_\beta a'_\gamma c'_\epsilon x^{k+l} 
-\frac12 a'_\alpha c'_\beta  x^{k+l}   a'_\gamma c'_\epsilon + \frac12 \sum_{\sigma=1}^{k-1}   \left(a'_\alpha c'_\beta x^{k-\sigma}   a'_\gamma c'_\epsilon x^{l+\sigma}+a'_\alpha c'_\beta x^{k+l-\sigma}   a'_\gamma c'_\epsilon x^{\sigma}\right) \\
&
-\frac{1}{2} \sum_{\tau=1}^{l-1} \left( a'_\alpha c'_\beta x^{\tau}   a'_\gamma c'_\epsilon x^{k+l-\tau} +a'_\alpha c'_\beta  x^{k+\tau}   a'_\gamma c'_\epsilon x^{l-\tau}\right) \\
=&\frac12 a'_\alpha c'_\beta a'_\gamma c'_\epsilon x^{k+l} 
-\frac12 a'_\gamma c'_\epsilon  a'_\alpha c'_\beta  x^{k+l}   
+ \frac12 \left[\sum_{v=1}^{k}-\sum_{v=1}^{l} \right]  \left(a'_\alpha c'_\beta x^{v}   a'_\gamma c'_\epsilon x^{k+l-v}+a'_\alpha c'_\beta x^{k+l-v}   a'_\gamma c'_\epsilon x^{v}\right) \,,
 \end{aligned}
\end{equation*}
where we added the terms $v=k,l$ in the sums because they cancel out together. Now, we compute 
\begin{equation*}
 \begin{aligned}
\dgal{a'_\gamma c'_\epsilon ,a'_\alpha c'_\beta}=&
\dgal{a'_\gamma ,a'_\alpha } c'_\beta \ast c'_\epsilon 
+a'_\alpha \dgal{a'_\gamma  ,c'_\beta} \ast c'_\epsilon
+a'_\gamma  \ast \dgal{c'_\epsilon ,a'_\alpha } c'_\beta
+a'_\gamma \ast a'_\alpha \dgal{ c'_\epsilon , c'_\beta} \\
=&
-\frac12 \,o(\gamma,\alpha) \left(a'_\gamma c'_\epsilon\otimes a'_\alpha c'_\beta + a'_\alpha c'_\epsilon \otimes a'_\gamma c'_\beta\right) +\frac12 o(\epsilon,\beta) \left(a'_\alpha c'_\beta \otimes a'_\gamma c'_\epsilon -a'_\alpha c'_\epsilon \otimes a'_\gamma c'_\beta \right) \\
& -\frac12 a'_\alpha c'_\beta a'_\gamma c'_\epsilon \otimes e_0
 +\frac12 \left(o(\gamma,\beta)-\delta_{\gamma \beta} \right) a'_\alpha c'_\epsilon \otimes a'_\gamma c'_\beta
 -\delta_{\gamma \beta} \left(a'_\alpha c'_\epsilon \otimes e_0 z + \sum_{\mu=1}^{\beta-1} 
 a'_\alpha c'_\epsilon \otimes a'_\mu c'_\mu \right) \\
&
+\frac12 e_0 \otimes a'_\gamma c'_\epsilon a'_\alpha c'_\beta
 - \frac12 \left(o(\alpha,\epsilon)-\delta_{\alpha \epsilon} \right)  a'_\alpha c'_\epsilon \otimes a'_\gamma c'_\beta
 +\delta_{\alpha \epsilon} \left(e_0 z\otimes a'_\gamma c'_\beta + \sum_{\lambda=1}^{\epsilon-1} 
  a'_\lambda c'_\lambda   \otimes a'_\gamma c'_\beta \right) \,,
 \end{aligned}
\end{equation*}
which we have to multiply on the right by $x^l$ (for the outer bimodule structure) and $x^k$ (for the inner bimodule structure). After doing so, we apply the multiplication map and denote by $S_3$ the expression obtained in that way, i.e. $S_3=m \circ (\{\!\{a'_\gamma c'_\epsilon ,a'_\alpha c'_\beta\}\!\}x^l \ast x^k)$. 
We finally get 
\begin{equation*}
 \begin{aligned}
&\br{a'_\gamma c'_\epsilon x^k,a'_\alpha c'_\beta x^l}=S_2+S_3 \\
=&+ \frac12 \left[\sum_{v=1}^{k}-\sum_{v=1}^{l} \right] 
 \left(a'_\alpha c'_\beta x^{v}   a'_\gamma c'_\epsilon x^{k+l-v}+a'_\alpha c'_\beta x^{k+l-v}   a'_\gamma c'_\epsilon x^{v}\right) \\
&+\frac12 \,o(\alpha,\gamma) \left(a'_\gamma c'_\epsilon x^k  a'_\alpha c'_\beta x^l + a'_\alpha c'_\epsilon  x^k  a'_\gamma c'_\beta x^l\right) +\frac12 o(\epsilon,\beta) \left(a'_\alpha c'_\beta  x^k  a'_\gamma c'_\epsilon x^l - a'_\alpha c'_\epsilon  x^k  a'_\gamma c'_\beta x^l \right) \\
&  -\frac12 \left(o(\beta,\gamma)+\delta_{\gamma \beta} \right) a'_\alpha c'_\epsilon  x^k  a'_\gamma c'_\beta x^l
 -\delta_{\gamma \beta} \left(a'_\alpha c'_\epsilon  x^k  z x^l + \sum_{\mu=1}^{\beta-1} 
 a'_\alpha c'_\epsilon  x^k  a'_\mu c'_\mu x^l \right) \\
& + \frac12 \left(o(\epsilon,\alpha)+\delta_{\alpha \epsilon} \right)  a'_\alpha c'_\epsilon  x^k  a'_\gamma c'_\beta x^l
 +\delta_{\alpha \epsilon} \left(z x^k  a'_\gamma c'_\beta x^l + \sum_{\lambda=1}^{\epsilon-1} 
  a'_\lambda c'_\lambda    x^k  a'_\gamma c'_\beta x^l \right) \,,
 \end{aligned}
\end{equation*}
modulo commutators. This is our claim. \qed

\subsection{Proof of Lemma \ref{Lem:Cytt}} \label{Ann:B2bis}

Let's restate the setting. For $u\in \{x,y,z,\sum_s e_s+xy\}$, set  $\epsilon(x)=+1$, $\epsilon(y)=-1$, $\epsilon(z)=-1$ or $\epsilon(\sum_s e_s +xy)=+1$.  We also set $\theta(u)=\epsilon(u)$ if $u=x,y,z$, while $\theta(\sum_s e_s+xy)=0$. With these notations, we have $u \in \oplus_s e_s A e_{s+\theta(u)}$, and we can write  that $\dgal{u,u}=\frac12 \epsilon(u)[u^2 F_{\theta(u)} - F_{\theta(u)} u^2]$. Moreover, we can also obtain in all cases that 
\begin{equation} \label{EquvwCy}
  \dgal{u, w_\alpha}= \frac12 e_{0}\otimes uw_\alpha-\frac12 e_0 u\otimes w_\alpha\,,\quad 
\dgal{u, v_\alpha}= \frac12 v_\alpha u\otimes e_0-\frac12 v_\alpha\otimes u e_0\,.
\end{equation}

For the first statement, we compute for any $\alpha,\beta=1,\ldots,d$ 
\begin{equation*} 
  \dgal{u, w_\alpha v_\beta}=
\frac12 (w_\alpha v_\beta u\otimes e_0 -w_\alpha v_\beta\otimes u e_0  +   e_{0}\otimes uw_\alpha v_\beta-   e_0 u\otimes w_\alpha v_\beta )\,,
\end{equation*}
We find  in a similar way to \eqref{LodCy2} in the proof of Lemma \ref{LodCyxac}, 
\begin{equation*}
 \begin{aligned} 
\dgal{u^k,w_\alpha v_\beta u^l}=&
 \frac12 \sum_{\sigma=1}^k 
\Big( e_0 u^{k-\sigma} \otimes u^{\sigma} w_\alpha v_\beta u^l  - e_0 u^{k-\sigma+1}  \otimes u^{\sigma-1} w_\alpha v_\beta u^l  \\
&\qquad \qquad  + w_\alpha v_\beta u^{k-\sigma+1} \otimes u^{\sigma-1}e_{0}u^l  - w_\alpha v_\beta u^{k-\sigma} \otimes u^{\sigma}e_{0}u^l \Big) \\
&+\frac{1}{2} \sum_{\sigma=1}^k \sum_{\tau=1}^l \sum_{s\in I}
\big(w_\alpha v_\beta u^{\tau+1} e_s u^{k-\sigma} \otimes u^{\sigma-1} e_{s-\theta(u)} u^{l-\tau} \\
& \qquad \qquad \qquad \qquad  - w_\alpha v_\beta u^{\tau-1}e_s u^{k-\sigma} \otimes u^{\sigma-1}e_{s-\theta(u)}u^{l-\tau+2} \big) \,.
 \end{aligned}
\end{equation*}
We now apply the multiplication map, and we clearly see that the terms in the first two lines cancel out. In the last two lines, we obtain factors  
$e_s u^{k-1} e_{s-\theta(u)}=u^{k-1}e_{s+(k-1)\theta(u)}e_{s-\theta(u)}$ because $u \in \oplus_s e_s A e_{s+\theta(u)}$. Thus, these two lines clearly disappear if $k$ is not divisible by $m$. Assuming now that $k=0$ mod $m$, we also remark that we have the factor $w_\alpha v_\beta u^{\tau+1} e_s$ in the third line, and since $v_\beta=v_\beta e_0$ this implies that the only $s$ that gives a nonzero term is such that $s-(\tau+1)\theta(u)=0$ mod $m$. The same argument in the last line allows to remove the idempotents and the sum over $s\in I$. Thus, all the terms in those sums are just $w_\alpha v_\beta u^{\tau+k+l}$, and they cancel out together. 

\medskip

For the second claim, we  show more generally that for any fixed $\alpha=1,\ldots,d$, the elements $(u^k,w_\alpha v_\alpha u^l)$ form a commutative Lie subalgebra in $A/[A,A]$. 
It is just an application of  \cite[Lemma A.3]{CF} to show $\br{u^k,u^l}=0$, and we have from the previous part that 
$\br{u^k,w_\alpha v_\alpha u^l}=0$. Thus, it remains to prove that $\br{w_\alpha v_\alpha u^k,w_\alpha v_\alpha u^l}=0$ in $A/[A,A]$. On one hand, we get again by adapting the argument in the proof of Lemma \ref{LodCyxac} 
\begin{equation*}
  \begin{aligned}
 &m \circ \left( w_\alpha v_\alpha \ast \dgal{u^k,w_\alpha v_\alpha u^l}+w_\alpha v_\alpha\dgal{w_\alpha v_\alpha ,u^l}\ast u^k \right)\\
=&\quad  \frac12 
\left[ \sum_{\sigma=1}^k - \sum_{\sigma=0}^{k-1} \right]    u^{k-\sigma}  w_\alpha v_\alpha  u^{\sigma}w_\alpha v_\alpha  u^l     +\frac12 \left[ \sum_{\sigma=0}^{k-1} - \sum_{\sigma=1}^{k} \right]  w_\alpha v_\alpha  u^{k-\sigma}  w_\alpha v_\alpha  u^{l+\sigma}  \\
&+ \frac12 \left[ \sum_{\sigma=0}^{k-1} \sum_{\tau=l}+ \sum_{\sigma=0} \sum_{\tau=1}^{l-1}- \sum_{\sigma=1}^{k-1} \sum_{\tau=0} - \sum_{\sigma=k} \sum_{\tau=0}^{l-1}  \right]
w_\alpha v_\alpha x^{k-\sigma + \tau}   w_\alpha v_\alpha x^{l-\tau + \sigma}
 \\
&+\frac12\left[-\sum_{\tau=1}^l+\sum_{\tau=0}^{l-1} \right] w_\alpha v_\alpha  u^{\tau}w_\alpha v_\alpha     u^{k+l-\tau}
+\frac12 \left[- \sum_{\tau=0}^{l-1} + \sum_{\tau=1}^l \right] w_\alpha v_\alpha   u^{k+\tau} w_\alpha v_\alpha  u^{l-t}\,,
  \end{aligned}
\end{equation*}
because we can get rid of the idempotents modulo commutators, after careful analysis. After simplification, all terms vanish modulo commutators. 

On the other hand,  we compute 
\begin{equation*}
  \dgal{w_\alpha v_\alpha,w_\alpha v_\alpha}= e_0 \otimes w_\alpha v_\alpha - w_\alpha v_\alpha \otimes e_0 
+\frac12 e_0 \otimes (w_\alpha v_\alpha)^2 - \frac12 (w_\alpha v_\alpha)^2 \otimes e_0\,.
 \end{equation*}
Hence $m \circ (\dgal{w_\alpha v_\alpha,w_\alpha v_\alpha}x^l \ast x^k)=0$ modulo commutators and we can conclude. \qed

\subsection{Proof of Proposition \ref{PropInvPhiCy}} \label{Ann:B3}
This proof can be applied without change in the case $m=1$ treated in \cite{CF2}, after setting $e_s=e_0$ for each $s\in I$, and $F_b=e_0 \otimes e_0$ for all $b \in \Z$. 

\medskip

We begin with the first identity,  and let $U_\alpha=u(1+\eta \phi)$ instead of $U_{+,\alpha}$ to ease notations. 
Writing $\dgal{U_{\alpha},U_{\eta}}=a'\otimes a''$, we get  
 \begin{equation} \label{EqUU}
   \frac{1}{KL}\br{U_{\alpha}^K,U_{\eta}^L}=U_{\eta}^{L-1} a' U_{\alpha}^{K-1} a'' \quad \text{ mod }[A,A]\,, 
 \end{equation}
So we have to compute 
\begin{equation} \label{EqUUCy}
 \dgal{u+\alpha u\phi,u+\eta u\phi}= \dgal{u,u}+\alpha \dgal{u\phi,u}+\eta \dgal{u,u\phi}
+\alpha \eta \dgal{u\phi,u\phi}\,. 
\end{equation}
With the notations of the proof of Lemma \ref{Lem:Cytt}, we find 
\begin{equation*}
  \dgal{u,u}=\frac12 \epsilon(u)[u^2 F_{\theta(u)} - F_{\theta(u)} u^2]\,, \quad 
\dgal{\phi,\gamma}=\frac{1}{2}\phi \ast (\gamma F_0-F_0 \gamma)+ \frac12 (\gamma F_0 - F_0 \gamma)\ast \phi\,,\,,
\end{equation*}
where  $\gamma$ is any word in the letters $\{e_s,x_s,y_s\}$ (with possible inverses). The second equation is obtained by combining \eqref{Phim} and Lemma \ref{Lem1} applied to the subquiver based at $I$, the set of all vertices in the cycle. We see that we can write   
\begin{equation*}
\begin{aligned}
    \dgal{\phi,u}=\frac{1}{2} (u F_0 \phi-F_0 \phi u)+ \frac12 (u\phi F_0 -\phi F_0 u)\,,  \quad 
\dgal{\phi,\phi}=\frac{1}{2}(\phi^2 F_0-F_0 \phi^2)\,,
\end{aligned}
\end{equation*}
because $\phi\in \oplus_s e_s A e_s$, so $\phi$ commutes with any $e_s$. 
As $u \in \oplus_s e_s A e_{s+\theta(u)}$ we have $u e_s = e_{s-\theta(u)}u$,  so that 
\begin{equation*}
  \dgal{u,\phi}=\frac{1}{2}\phi (u F_{\theta(u)}-F_{\theta(u)} u)
+ \frac12 (uF_{\theta(u)} - F_{\theta(u)} u)\phi\,.
\end{equation*}

From these basic results, we directly get the first term in \eqref{EqUUCy}. For the second term, we compute 
\begin{equation*}
\begin{aligned}
  \dgal{u\phi, u}=&u\ast \dgal{\phi,u}+\dgal{u,u}\ast \phi \\
=&\frac12  u \ast (u F_0 \phi-F_0 \phi u)+ \frac12 u\ast (u\phi F_0  - \phi F_0 u) 
+\frac12 \epsilon(u) (u^2 \phi F_{\theta(u)}- \phi F_{\theta(u)} u^2)  \\
=&\frac12   (u F_{\theta(u)} u \phi-F_{\theta(u)} u \phi u + u\phi F_{\theta(u)}u  - \phi F_{\theta(u)} u^2) 
+\frac12 \epsilon(u) (u^2 \phi F_{\theta(u)}- \phi F_{\theta(u)} u^2)
\end{aligned}
\end{equation*}
using that $u \ast F_0 =F_{\theta(u)}u$ since $u e_s= e_{s-\theta(u)}u$. The term 
$\dgal{u,u\phi}=-\dgal{u\phi, u}^\circ$ follows from the following result. 
\begin{lem}
Fix some $r \in \N$ and let  $a \in \oplus_s e_s A e_s$, $b_0,b_1 \in \oplus_s e_s A e_{s+r}$ and $c\in \oplus_s e_s A e_{s+2r}$. Then $(b_0 F_r b_1)^\circ = b_1 F_r b_0$ and $(c F_r a)^\circ = a F_r c$. 
\end{lem}
\begin{proof}
  We compute $(b_0 F_r b_1)^\circ = \sum_s e_{s}b_1\otimes b_0 e_{s+r}
=\sum_s b_1 e_{s+r}\otimes e_{s}b_0 =b_1 F_{r} b_0$. The second equality follows similarly. 
\end{proof}

Taking $r=\theta(u)$, $b_0=u$ and $b_1=u \theta$ gives $(u F_{\theta(u)} u \phi)^\circ=u \phi F_{\theta(u)} u$. Using the lemma on the other terms of $\dgal{u\phi, u}$ yields 
\begin{equation*}
  \dgal{u,u\phi}=-\frac12   (u \phi F_{\theta(u)} u - u \phi u F_{\theta(u)}  + u F_{\theta(u)}u \phi  - u^2 F_{\theta(u)} \phi) 
+\frac12 \epsilon(u) (u^2  F_{\theta(u)} \phi-  F_{\theta(u)} u^2 \phi)\,.
\end{equation*}
Computing that 
\begin{equation*}
\begin{aligned}
   \dgal{u\phi,\phi}=&\dgal{u,\phi}\ast \phi + u \ast \dgal{\phi,\phi} \\
=& \frac{1}{2}\phi (u F_{\theta(u)}-F_{\theta(u)} u) \ast \phi
+ \frac12 (uF_{\theta(u)} - F_{\theta(u)} u)\phi \ast \phi 
+ \frac12 u \ast  (\phi^2 F_0-F_0 \phi^2) \\
=& \frac{1}{2}(\phi u \phi F_{\theta(u)}+ u\phi F_{\theta(u)} \phi 
- \phi F_{\theta(u)} u \phi-F_{\theta(u)}u \phi^2)\,,
\end{aligned}
\end{equation*}
we find for the fourth term $\dgal{u\phi,u\phi}=u\dgal{u\phi,\phi}+\dgal{u\phi,u}\phi$ that 
\begin{equation*}
  \dgal{u\phi,u\phi}
=\frac{1}{2}(u\phi u \phi F_{\theta(u)}+ u^2\phi F_{\theta(u)} \phi 
 - F_{\theta(u)}u \phi u \phi- \phi F_{\theta(u)}u^2\phi ) 
+\frac12 \epsilon(u) (u^2 \phi F_{\theta(u)} \phi- \phi F_{\theta(u)} u^2\phi) \,.
\end{equation*}
Remarking that  $U_\alpha-u=\alpha u\phi$ and $U_\eta-u=\eta u\phi$, we find 
\begin{equation*}
 \begin{aligned}
\alpha\eta  \dgal{u\phi,u\phi}=&
\frac{1}{2}\eta u\phi (U_\alpha-u) F_{\theta(u)}+ \frac12\eta u(U_\alpha-u) F_{\theta(u)} \phi 
 -\frac12 \alpha F_{\theta(u)}u \phi (U_\eta-u) \\
&-\frac12 \alpha \phi F_{\theta(u)}u (U_\eta-u)  
+\frac12 \epsilon(u) \eta u(U_\alpha-u) F_{\theta(u)} \phi
- \frac12 \epsilon(u)\alpha  \phi F_{\theta(u)} u (U_\eta-u)\,.
 \end{aligned}
\end{equation*}
Now, we sum all the terms appearing in \eqref{EqUUCy}, which yields 
\begin{equation*}
 \begin{aligned}
  \dgal{U_\alpha,U_\eta}=&\,
\frac12 \epsilon(u) (u^2 F_{\theta(u)}- F_{\theta(u)} u^2) 
+\frac12 \alpha  \epsilon(u) u^2 \phi F_{\theta(u)} 
-\frac12 \eta \epsilon(u)   F_{\theta(u)} u^2 \phi 
\\
&+\frac12 \alpha  (u F_{\theta(u)}u\phi 
+ u\phi F_{\theta(u)} u) 
-\frac12 \eta (u\phi F_{\theta(u)} u 
+ u F_{\theta(u)} u\phi) 
\\
&+\frac{1}{2}\eta u\phi U_\alpha F_{\theta(u)}+ \frac12\eta uU_\alpha F_{\theta(u)} \phi 
 -\frac12 \alpha F_{\theta(u)}u \phi U_\eta \\
&-\frac12 \alpha \phi F_{\theta(u)}u U_\eta 
+\frac12 \epsilon(u) \eta u U_\alpha F_{\theta(u)} \phi
- \frac12 \epsilon(u)\alpha  \phi F_{\theta(u)} u U_\eta\,.
 \end{aligned}
\end{equation*}
As $u$ is assumed to be invertible, we can repeat the substitution under the form $u^{-1}(U_\alpha-u)=\alpha \phi$.  We find in this way  
\begin{equation*}
 \begin{aligned}
  \dgal{U_\alpha,U_\eta}=&\,
\frac12 \epsilon(u) (u^2 F_{\theta(u)}- F_{\theta(u)} u^2) 
+\frac12   \epsilon(u) u (U_\alpha-u) F_{\theta(u)} 
-\frac12  \epsilon(u)   F_{\theta(u)} u (U_\eta-u)
\\
&+\frac12 (u F_{\theta(u)} (U_\alpha-u)
+ (U_\alpha-u) F_{\theta(u)} u) 
-\frac12  ((U_\eta-u) F_{\theta(u)} u 
+ u F_{\theta(u)} (U_\eta-u)) 
\\
&+\frac{1}{2} (U_\eta-u) U_\alpha F_{\theta(u)}+ \frac12 uU_\alpha F_{\theta(u)} u^{-1}(U_\eta-u)
 -\frac12 F_{\theta(u)} (U_\alpha-u) U_\eta \\
&-\frac12 u^{-1}(U_\alpha-u) F_{\theta(u)}u U_\eta 
+\frac12 \epsilon(u)  u U_\alpha F_{\theta(u)} u^{-1}(U_\eta-u)
- \frac12 \epsilon(u)  u^{-1}(U_\alpha-u) F_{\theta(u)} u U_\eta
 \end{aligned}
\end{equation*}
This may be reduced to the form 
\begin{equation*}
 \begin{aligned}
  \dgal{U_\alpha,U_\eta}=&\,
+\frac12 (1+\epsilon(u)) \left[   u U_\alpha F_{\theta(u)} u^{-1}U_\eta
-  u^{-1}U_\alpha F_{\theta(u)} u U_\eta \right]
\\
&+\frac12 (u F_{\theta(u)} U_\alpha
+ U_\alpha F_{\theta(u)} u) 
-\frac12  (U_\eta F_{\theta(u)} u 
+ u F_{\theta(u)} U_\eta) 
\\
&-u U_\alpha F_{\theta(u)}+ F_{\theta(u)} u U_\eta
+\frac{1}{2} (U_\eta U_\alpha F_{\theta(u)}-F_{\theta(u)} U_\alpha U_\eta )
 \end{aligned}
\end{equation*}
The latter expression can be put back in \eqref{EqUU}, and we get 
\begin{equation} \label{EqUKLend}
 \begin{aligned}
  \br{U_\alpha^K,U_\eta^L}=&\,
\frac12 (1+\epsilon(u)) \left[   U_{\eta}^L u U_\alpha^K  u^{-1}
- U_\eta^L  u^{-1}U_\alpha^K u \right]
\\
&+\frac12 (-U_\eta^{L-1} u U_\alpha^{K} + U_\eta^{L-1} U_\alpha^{K}u
+ U_\eta^{L} U_\alpha^{K-1} u - U_\eta^{L} u U_\alpha^{K-1} )
 \end{aligned}
\end{equation}
all mod $[A,A]$. 
Indeed, since we have the decomposition $U_\alpha\in \oplus_s e_s A e_{s+\theta(u)}$ we can write when we insert the first term of $\dgal{U_\alpha,U_\eta}$ in  \eqref{EqUU} that 
\begin{equation*} 
\begin{aligned}  
\sum_s U_\eta^{L-1} u U_\alpha e_{s+\theta(u)} U_{\alpha}^{K-1} e_s u^{-1}U_\eta
=& U_\eta^{L-1} u U_\alpha \left(\sum_s e_{s+\theta(u)} e_{s-(K-1)\theta(u)}\right) U_{\alpha}^{K-1}  u^{-1}U_\eta \\
=& U_\eta^{L-1} u U_\alpha  U_{\alpha}^{K-1}  u^{-1}U_\eta
\end{aligned}
\end{equation*}
for $(K-1)\theta(u) {\equiv}-\theta(u)$ mod $m$, that is $K$ is divisible by $m$ if $\theta(u)=\pm1$ (for   $u=x,y,z$), while $K \geq 1$ for $\theta(u)=0$ (for $u=\sum_s e_s+xy$). For every other element, we can proceed in the same way and establish \eqref{EqUKLend}. By assumption, the first line in  \eqref{EqUKLend} vanishes for $\epsilon(u)=-1$. The second line is trivially zero if $\alpha=\eta$. Otherwise we use  $u=\frac{1}{\alpha-\eta} (\alpha U_\eta - \eta U_\alpha)$ and we also get that the second line vanishes. 

\medskip

The same proof works when $\epsilon(z)=+1$ to show that  $\br{U_{-,\alpha}^K,U_{-,\eta}^L}=0$ modulo commutators for $U_{-,\alpha}=u(1+\alpha \phi^{-1})$. We only need to notice that  $\dgal{\phi^{-1},a}=-\phi^{-1} \ast \dgal{\phi,a} \ast \phi^{-1}$, so we just need to replace in the expression $\dgal{\phi,a}$ the factors $\phi$ by $\phi^{-1}$ 
and multiply  by an overall factor $-1$. Thus, reproducing the proof in the first case with some sign changes, we get 
\begin{equation*}
 \frac{1}{KL}\br{U_{-,\alpha}^K,U_{-,\eta}^L}=
\frac12 (-1+\epsilon(u)) \left(U_{-,\eta}^{L} u U_{-,\alpha}^{K} u^{-1} 
- U_{-,\eta}^{L} u^{-1} U_{-,\alpha}^{K}u \right)\,,
\end{equation*}
modulo commutators. This yields the desired result for $\epsilon(u)=+1$.  \qed

\section{Poisson isomorphism between MQVs}\label{Ann:iso}

Before proving Proposition \ref{IsoCMPoiss}, let's remark that Lemma \ref{LodCyxac} together with \eqref{relInv} give  
\begin{subequations}
 \begin{align}
 & \brap{\tr X^k,\tr X^l}=0\,, \quad 
\brap{\tr X^k,\tr(\As^{(m)}  E_{\alpha \beta} \Cs^{(m)}  X^l)}=k\, \tr(\As^{(m)}  E_{\alpha \beta} \Cs^{(m)}  X^{k+l})\,, \\
&\brap{\tr(\As^{(m)}  E_{\gamma \epsilon} \Cs^{(m)}  X^k),\tr(\As^{(m)}  E_{\alpha \beta} \Cs^{(m)}  X^l)} \nonumber \\
=&\frac12 \left(\sum_{v=1}^k-\sum_{v=1}^l \right) 
\left(\tr(\As^{(m)}  E_{\alpha \beta} \Cs^{(m)}  X^v \As^{(m)}  E_{\gamma \epsilon} \Cs^{(m)}  X^{k+l-v})
+\tr(\As^{(m)}  E_{\alpha \beta} \Cs^{(m)}  X^{k+l-v} \As^{(m)}  E_{\gamma \epsilon} \Cs^{(m)}  X^{v}) \right) \nonumber \\
&+\frac12  o(\alpha,\gamma) \left(\tr(\As^{(m)}  E_{\gamma \epsilon} \Cs^{(m)}  X^k \As^{(m)}  E_{\alpha \beta} \Cs^{(m)}  X^l) 
+\tr(\As^{(m)}  E_{\alpha \epsilon} \Cs^{(m)}  X^k \As^{(m)}  E_{\gamma \beta} \Cs^{(m)}  X^l)\right) \nonumber \\
&+\frac12  o(\epsilon,\beta) \left(\tr(\As^{(m)}  E_{\alpha \beta} \Cs^{(m)}  X^k \As^{(m)}  E_{\gamma \epsilon} \Cs^{(m)}  X^l) 
-\tr(\As^{(m)}  E_{\alpha \epsilon} \Cs^{(m)}  X^k \As^{(m)}  E_{\gamma \beta} \Cs^{(m)}  X^l)\right) \nonumber \\
&+\frac12 [o(\epsilon,\alpha)+\delta_{\alpha \epsilon}]\,\tr(\As^{(m)}  E_{\alpha \epsilon} \Cs^{(m)}  X^k \As^{(m)}  E_{\gamma \beta} \Cs^{(m)}  X^l) \nonumber \\
&-\frac12 [o(\beta,\gamma)+\delta_{\beta \gamma}]\,\tr(\As^{(m)}  E_{\alpha \epsilon} \Cs^{(m)}  X^k \As^{(m)}  E_{\gamma \beta} \Cs^{(m)}  X^l) \nonumber \\
&+ \delta_{\alpha \epsilon}\tr\left( \left[t^{-1}Z+\sum_{\lambda=1}^{\epsilon-1} \As^{(m)}  E_{\lambda \lambda} \Cs^{(m)}  \right]X^k \As^{(m)}  E_{\gamma \beta} \Cs^{(m)}  X^l\right) \nonumber \\
&- \delta_{\beta \gamma}\,\, \tr\left( \As^{(m)}  E_{\alpha \epsilon} \Cs^{(m)}  X^k
\left[t^{-1} Z + \sum_{\mu=1}^{\beta-1} \As^{(m)}  E_{\mu \mu} \Cs^{(m)}   \right]X^l\right)\,.
 \end{align}
\end{subequations}
Note the appearance of two constants $t^{-1}$ in the last two terms. 
Indeed, we used that $(\As^{(m)} )_{i\alpha}$ is the $i$-th component of the covector representing $a'_{\alpha}$, while $(\Cs^{(m)} )_{\beta j}$ is the $j$-th component of the vector representing $t^{-1} c'_{\alpha}$, while $X$ and $Z$ respectively represent $x,z$. 
Furthermore, for the last equality, it is an easy exercise to see that both matrices $Z$ in the right-hand side can be replaced by $Z_{m-1}$ after using that $k,l=1$ modulo $m$ in that expression. The second set of Poisson brackets that we need are given by \eqref{brtad1}--\eqref{brtad2} for the coordinates defined in \eqref{fgTadp}, and we write in our case 
for $k_0,l_0\geq 1$ 
\begin{subequations}
 \begin{align}
  &\brap{\tr(A^{k_0}),\tr(A^{l_0})}=\,0\,, \quad 
\brap{\tr(A^{k_0}),\tr(\As E_{\alpha \beta} \Cs A^{l_0})}=\,{k_0}\, \tr(\As E_{\alpha \beta} \Cs A^{{k_0}+{l_0}})\,, \\
&\brap{\tr(\As E_{\gamma\epsilon} \Cs A^{k_0}),\tr(\As E_{\alpha \beta} \Cs A^{l_0})} \nonumber  \\
=&
\frac12 \left(\sum_{r=1}^{k_0}-\sum_{r=1}^{l_0} \right) 
\left(\tr(\As E_{\alpha \beta} \Cs A^r \As E_{\gamma\epsilon} \Cs A^{{k_0}+{l_0}-r})
+\tr(\As E_{\alpha \beta} \Cs A^{{k_0}+{l_0}-r} \As E_{\gamma\epsilon} \Cs A^r)\right) \nonumber \\
&+\frac12 o(\alpha,\gamma) \left(\tr(\As E_{\gamma\epsilon} \Cs A^{k_0} \As E_{\alpha\beta} \Cs A^{l_0})
+ \tr(\As E_{\alpha\epsilon} \Cs A^{k_0} \As E_{\gamma\beta} \Cs A^{l_0}) \right) \nonumber \\
&+\frac12 o(\epsilon,\beta) 
\left( \tr( \As E_{\alpha\beta} \Cs A^{k_0} \As E_{\gamma\epsilon} \Cs A^{l_0}) 
- \tr(\As E_{\alpha\epsilon} \Cs A^{k_0} \As E_{\gamma\beta} \Cs A^{l_0}) \right) \nonumber \\
&+\frac12 [o(\epsilon,\alpha)+\delta_{\alpha \epsilon}]\,
\tr(\As E_{\alpha\epsilon} \Cs A^{k_0} \As E_{\gamma\beta} \Cs A^{l_0}) 
-\frac12 [o(\beta,\gamma)+\delta_{\beta \gamma}]\,
\tr(\As E_{\alpha\epsilon} \Cs A^{k_0} \As E_{\gamma\beta} \Cs A^{l_0}) \nonumber \\
&+\delta_{\alpha \epsilon} 
\tr\left( \left[B  + \sum_{\lambda=1}^{\epsilon-1} \As E_{\lambda \lambda} \Cs  \right] A^{k_0}
\As E_{\gamma \beta} \Cs A^{l_0}\right)  \nonumber \\
&-\delta_{\beta \gamma} \tr\left( \left[B + \sum_{\mu=1}^{\beta-1}\As E_{\mu \mu} \Cs \right] A^{l_0}
\As E_{\alpha\epsilon} \Cs A^{k_0} \right)\,.
 \end{align}
\end{subequations}

\begin{proof}  \emph{[Proposition \ref{IsoCMPoiss}].}
It suffices to show that the map $\psi:\,\Cntd^\circ \to \Cnmo$ is a Poisson map with respect to a basis of functions on  $\Cnmo$. It is not hard to see that we can pick  the functions 
$F_k:= \tr(X^k)$ and $G_{l}^{\gamma \epsilon}:=\tr(\As^{(m)}  E_{\gamma\epsilon} \Cs^{(m)}  X^l)$, which are nonzero for $k=k_0m $ and $l=l_0 m+1$ with $k_0,l_0 \geq 1$. We always assume such a choice for the indices from now on  (which thus depends on the function $F_k$ or $G_l^{\gamma \epsilon}$). Using the notations from \eqref{fgTadp}, we can see that 
\begin{equation}
 \psi^\ast F_k=m \tr(A^{k_0})=m f_{k_0}\,, \quad 
\psi^\ast G_l^{\gamma \epsilon}=\tr(A^{-1} \As  E_{\gamma\epsilon} \Cs A^{l_0+1}) = g_{l_0}^{\gamma \epsilon}\,.
\end{equation}
Indeed, we have $\Cs^{(m)}  X^l=\Cs^{(m)}  X_{m-1} (X_0\ldots X_{m-1})^{l_0}$ which is $\Cs A^{l_0+1}$ under $\psi^\ast$.
Hence, we have to show that  the following equalities hold (writing $\brap{-,-}$ for both Poisson brackets)
\begin{equation*}
 \psi^\ast\brap{F_k,F_l}=m^2 \brap{f_{k_0} , f_{l_0}}\,, \quad 
\psi^\ast\brap{F_k,G_l^{\alpha \beta}}=m\brap{f_{k_0} , g_{l_0}^{\alpha \beta}}\,, \quad 
\psi^\ast\brap{G_{k}^{\gamma \epsilon},G_l^{\alpha \beta}}=\brap{g_{k_0}^{\gamma \epsilon}, g_{l_0}^{\alpha \beta}}\,.
\end{equation*}
The first equality is obvious as both sides vanish. For the second one, 
\begin{equation}
 \psi^\ast\brap{F_k,G_l^{\alpha \beta}}= k\,\, \psi^\ast G_{k+l}^{\alpha \beta}
=(k_0m)  g_{k_0+l_0}^{\alpha \beta}=m \brap{f_{k_0} , g_{l_0}^{\alpha \beta}}\,,
\end{equation}
since $k=k_0m$ and $l=l_0m+1$, so that $k+l=(k_0+l_0)m+1$. For the last equality, we have written at the beginning of this appendix that the left-hand side is  
\begin{equation*}
 \begin{aligned}
 & \frac12 \left(\sum_{v=1}^k-\sum_{v=1}^l \right) 
\left(\psi^\ast\tr(\As^{(m)}  E_{\alpha \beta} \Cs^{(m)}  X^v \As^{(m)}  E_{\gamma \epsilon} \Cs^{(m)}  X^{k+l-v})
+\psi^\ast\tr(\As^{(m)}  E_{\alpha \beta} \Cs^{(m)}  X^{k+l-v} \As^{(m)}  E_{\gamma \epsilon} \Cs^{(m)}  X^{v}) \right) \nonumber \\
&+\frac12  o(\alpha,\gamma) \left(\psi^\ast\tr(\As^{(m)}  E_{\gamma \epsilon} \Cs^{(m)}  X^k \As^{(m)}  E_{\alpha \beta} \Cs^{(m)}  X^l) 
+\psi^\ast\tr(\As^{(m)}  E_{\alpha \epsilon} \Cs^{(m)}  X^k \As^{(m)}  E_{\gamma \beta} \Cs^{(m)}  X^l)\right) \nonumber \\
&+\frac12  o(\epsilon,\beta) \left(\psi^\ast\tr(\As^{(m)}  E_{\alpha \beta} \Cs^{(m)}  X^k \As^{(m)}  E_{\gamma \epsilon} \Cs^{(m)}  X^l) 
-\psi^\ast\tr(\As^{(m)}  E_{\alpha \epsilon} \Cs^{(m)}  X^k \As^{(m)}  E_{\gamma \beta} \Cs^{(m)}  X^l)\right) \nonumber \\
&+\frac12 [o(\epsilon,\alpha)+\delta_{\alpha \epsilon}]\,\psi^\ast\tr(\As^{(m)}  E_{\alpha \epsilon} \Cs^{(m)}  X^k \As^{(m)}  E_{\gamma \beta} \Cs^{(m)}  X^l)  \\
&-\frac12 [o(\beta,\gamma)+\delta_{\beta \gamma}]\,\psi^\ast\tr(\As^{(m)}  E_{\alpha \epsilon} \Cs^{(m)}  X^k \As^{(m)}  E_{\gamma \beta} \Cs^{(m)}  X^l) \nonumber \\
&+ \delta_{\alpha \epsilon}\psi^\ast\tr\left( \left[t^{-1}Z_{m-1}+\sum_{\lambda=1}^{\epsilon-1} \As^{(m)}  E_{\lambda \lambda} \Cs^{(m)}  \right]X^k \As^{(m)}  E_{\gamma \beta} \Cs^{(m)}  X^l\right) \nonumber \\
&- \delta_{\beta \gamma}\,\, \psi^\ast\tr\left( \As^{(m)}  E_{\alpha \epsilon} \Cs^{(m)}  X^k
\left[t^{-1} Z_{m-1} + \sum_{\mu=1}^{\beta-1} \As^{(m)}  E_{\mu \mu} \Cs^{(m)}   \right]X^l\right)
 \end{aligned}
\end{equation*}
In the first sum, we need $v$ to be congruent to $1$ modulo $m$ to have nonzero terms, which means that we can sum over $v=v_0m+1$ with $v_0=0,\ldots,k_0$ or $v_0=0,\ldots,l_0$. In that case, $\Cs^{(m)}  X^v =\Cs^{(m)}  X_{m-1}(X_0\ldots X_{m-1})^{v_0}$, and we can write the same for $v=k,l$. Composing with $\psi$, we write this last expression as 
\begin{equation}
 \begin{aligned}
 & \frac12 \left(\sum_{v_0=0}^{k_0}-\sum_{v_0=0}^{l_0} \right) 
\left(\tr(\As E_{\alpha \beta} \Cs A^{v_0} \As E_{\gamma \epsilon} \Cs A^{k_0+l_0-v_0})
+\tr(\As E_{\alpha \beta} \Cs A^{k_0+l_0-v_0} \As E_{\gamma \epsilon} \Cs A^{v_0}) \right) \nonumber \\
&+\frac12  o(\alpha,\gamma) \left(\tr(\As E_{\gamma \epsilon} \Cs A^{k_0} \As E_{\alpha \beta} \Cs A^{l_0}) 
+\tr(\As E_{\alpha \epsilon} \Cs A^{k_0} \As E_{\gamma \beta} \Cs A^{l_0})\right) \nonumber \\
&+\frac12  o(\epsilon,\beta) \left(\tr(\As E_{\alpha \beta} \Cs A^{k_0} \As E_{\gamma \epsilon} \Cs A^{l_0}) 
-\tr(\As E_{\alpha \epsilon} \Cs A^{k_0} \As E_{\gamma \beta} \Cs A^{l_0})\right) \nonumber \\
&+\frac12 [o(\epsilon,\alpha)+\delta_{\alpha \epsilon}]\,\tr(\As E_{\alpha \epsilon} \Cs A^{k_0} \As E_{\gamma \beta} \Cs A^{l_0})  
-\frac12 [o(\beta,\gamma)+\delta_{\beta \gamma}]\,\tr(\As E_{\alpha \epsilon} \Cs A^{k_0} \As E_{\gamma \beta} \Cs A^{l_0}) \nonumber \\
&+ \delta_{\alpha \epsilon}\tr\left( \left[A^{-1}B+\sum_{\lambda=1}^{\epsilon-1} A^{-1}\As E_{\lambda \lambda} \Cs \right]A^{k_0} \As E_{\gamma \beta} \Cs A^{l_0+1}\right) \nonumber \\
&- \delta_{\beta \gamma}\tr\left( \As E_{\alpha \epsilon} \Cs A^{k_0+1}
\left[A^{-1}B + \sum_{\mu=1}^{\beta-1} A^{-1}\As E_{\mu \mu} \Cs  \right]A^{l_0}\right)\,,
 \end{aligned}
\end{equation}
which is precisely $\brap{g_{k_0}^{\gamma \epsilon}, g_{l_0}^{\alpha \beta}}$. 
\end{proof}


\section{Calculations for the dynamics} \label{Ann:Dyn}

Our method goes as follow : the Hamiltonians come from functions of the form $\tr( \mathcal{X}(u_\eta)^K)$, for some $u_\eta\in A$. Then, defining the derivation $d/dt_K:=\brap{\tr(\mathcal{X}(u_\eta)^K),-}$, the evolution of a matrix $\mathcal{X}(c)$ representing an element $c\in A$  is governed by the ODE 
\begin{equation} \label{ODEdyn}
  \frac{d \mathcal{X}(c)}{d t_K}=  \mathcal{X}(\br{u_\eta^K,c})\,,
\quad \mathcal{X}(c)|_{t=0}:=C_0\,,
\end{equation}
using \ref{relBrDyn}, for some initial condition $C_0$. Thus, we are interested in computing the left Loday bracket $\br{u_\eta^K,c}=m \circ \dgal{u_\eta^K,c}$, which can be found by 
\begin{equation} \label{EqDyn}
  \br{u_\eta^K,c}=K\, \dgal{u_\eta,c}' u_\eta^{K-1} \dgal{u_\eta,c}''\,,
\end{equation}
after using the derivation property in the first variable then multiplying. Hence we need to compute $\dgal{u_\eta,c}$, then substitute the result back into \eqref{EqDyn}. Note that from the discussion at the end of \ref{ss:doubleqP}, we get for the ideal $J=(\Phi-q)$ that $\br{u_\eta^K,J}\subset J$, hence \eqref{ODEdyn} defines flows in $\Rep\left(\Lambda^{\qq},\aalpha\right)$ that we can project in $\Cnm$. The data of \eqref{EqDyn}  for a set of generators in $A$  can be seen as an analogue in the quasi-Poisson case to  an \emph{Hamiltonian ODE} on $A$ as defined in \cite[Section 2.4]{DSKV} for a double Poisson algebra.

\medskip

First, we look at the family $(G_k^m)_k$, which are the symmetric functions of the matrix representing the element  $u_\eta:=z(1+\eta \phi)$. We need the double brackets 
\begin{equation*}
 \begin{aligned}
   \dgal{z,z}=&-\frac12 \left(z^2 F_{-1}-F_{-1}z^2 \right) \,, \quad 
\dgal{z,x}=-\frac12 \left(xz F_{-1} +F_{-1} zx - z F_{-1} x + x F_{-1} z \right) \,, \\
\dgal{z, w_\beta}=& \frac12\,e_0\otimes zw_\beta
-\frac12\,e_0 z \otimes w_\beta\,,\quad 
\dgal{z, v_\beta}= \frac12\, v_\beta z\otimes e_0
-\frac12 \, v_\beta\otimes z e_0\,,
 \end{aligned}
\end{equation*}
obtained from \eqref{cyA}--\eqref{cyBz} and \eqref{EquvwCy}, together with 
 \begin{subequations}
       \begin{align}
\dgal{\phi,a}=&\frac{1}{2}\phi \ast (a F_0-F_0 a)+ \frac12 (aF_0 - F_0 a)\ast \phi\,, \quad \text{ for }a\in \{x,z\}\,, \\
\dgal{\phi,v_\beta}=&\frac{1}{2} (v_\beta \phi \otimes e_0 - v_\beta \otimes \phi e_0) \,, \quad 
\dgal{\phi,w_\beta}=\frac{1}{2} (e_0  \otimes \phi w_\beta - e_0 \phi \otimes w_\beta)\,. \label{EqPhivw}
\end{align}
  \end{subequations}
Note that the equations involving $v_\beta$ or $w_\beta$ need to be computed and do not follow from Lemma \ref{Lem1}, because such arrows appear from the framing which is not in the initial cyclic quiver for which $\phi$ is a moment map. We do it for the double bracket containing $v_\beta$ in \eqref{EqPhivw}, and the second case is left as an exercise.  Write $\phi=\phi_+ \phi_{-}^{-1}$ for $\phi_+=\sum_s e_s+xy,$ and $\phi_{-}=\sum_s e_s+yx$, and remark that  \eqref{EquvwCy} is satisfied for both $u=\phi_+,\phi_{-}$. Therefore 
\begin{equation*}
  \begin{aligned}
   \dgal{\phi, v_\beta}=& \dgal{\phi_+, v_\beta} \ast \phi_{-}^{-1} - \phi_+ \phi_{-}^{-1} \ast \dgal{\phi, v_\beta} \ast \phi_{-}^{-1} 
=\frac12  (v_\beta \phi_+ \phi_{-}^{-1} \otimes e_0-v_\beta\otimes \phi_+ \phi_{-}^{-1} e_0)\,,
  \end{aligned}
\end{equation*}
as desired.  
Now, we can begin the computations. First, 
\begin{equation*}
  \begin{aligned}
    \dgal{u_\eta,x}=&\dgal{z,x}\ast (1+\eta \phi)  +\eta  z \ast \dgal{ \phi,x} \\
=&- \frac12 (xz F_{-1} +F_{-1} zx - z F_{-1} x + x F_{-1} z ) \ast (1+\eta \phi) \\
&+\frac12 \eta (z\phi \ast (x F_0-F_0 x)+ z\ast  (xF_0 - F_0 x)\ast \phi )\\
=&- \frac12 (xz(1+\eta \phi)  F_{-1} +(1+\eta \phi) F_{-1} zx - z(1+\eta \phi) F_{-1} x + x(1+\eta \phi) F_{-1} z ) \\
&+\frac12 \eta ((x F_{-1} z \phi -F_{-1} z \phi  x)+   (x\phi  F_{-1} z -\phi   F_{-1} z x))\\
  \end{aligned}
\end{equation*}
where we used that $\phi\in \oplus_s e_s A e_s$, while $z\in \oplus_s e_{s} A e_{s-1}$. By definition, $u_\eta=z(1+\eta \phi)$, so that $u_\eta-z=\eta z\phi$. We can also use both expressions after multiplication from the left by $z^{-1}$. Thus 
\begin{equation*}
  \begin{aligned}
    \dgal{u_\eta,x}
=& - \frac12 (x u_\eta  F_{-1} +z^{-1}u_\eta F_{-1} zx - u_\eta F_{-1} x + xz^{-1}u_\eta F_{-1} z ) \\
&+\frac12  ((x F_{-1} (u_\eta-z)-F_{-1} (u_\eta-z) x)+   (xz^{-1}(u_\eta-z)  F_{-1} z -z^{-1}(u_\eta-z)   F_{-1} z x)) \\
=& -z^{-1}u_\eta F_{-1} zx - x F_{-1} z + F_{-1} zx +
\frac12 (x F_{-1} u_\eta-x u_\eta F_{-1} -F_{-1} u_\eta x+u_\eta F_{-1} x) \\
  \end{aligned}
\end{equation*}
Remarking that $u_\eta^{K-1}\in \oplus_s e_{s-1} A e_{s}$ for $K\in m \N$, we find 
\begin{equation*}
        \frac1K \br{u_\eta^K,x}= -z^{-1}u_\eta u_\eta^{K-1} zx  -x  u_\eta^{K-1} z + u_\eta^{K-1} zx 
= -\eta \phi u_\eta^{K-1} zx  -x  u_\eta^{K-1} z\,,
\end{equation*}
while this expression vanishes for $K \notin m \N$ since then $e_{s-1}u_\eta^{K-1}e_s=0$. This is similar in the other cases, hence we restrict to the case $K\in m \N$ from now on. 
Doing the same computations with $z$, we can find 
\begin{equation*}
  \begin{aligned}
    \dgal{u_\eta,z}=&\dgal{z,z}\ast (1+\eta \phi)  +\eta  z \ast \dgal{ \phi,z} \\
=&- \frac12 (z u_\eta  F_{-1} -z^{-1}u_\eta F_{-1} z^2) +\frac12  (z F_{-1} (u_\eta-z) -F_{-1} (u_\eta-z)  z) \\
&+\frac12    ( (u_\eta-z)  F_{-1} z -(z^{-1}u_\eta-1)   F_{-1} z^2)\\
=&F_{-1} z^2 - z F_{-1} z + \frac12 (zF_{-1}u_\eta - zu_\eta F_{-1} +u_\eta F_{-1}z -F_{-1}u_\eta z)\,.
  \end{aligned}
\end{equation*}
Hence $\frac1K \br{u_\eta^K,z}= -z u_\eta^{K-1} z  + u_\eta^{K-1} z^2 $. Next, we get 
\begin{equation*}
  \begin{aligned}
    \dgal{u_\eta,v_\beta}=&\dgal{z,v_\beta}\ast (1+\eta \phi)  +\eta  z \ast \dgal{ \phi,v_\beta} \\
=& \frac12 \left( 
 v_\beta u_\eta \otimes e_0- v_\beta z^{-1}u_\eta \otimes z e_0 + v_\beta (z^{-1}u_\eta - 1) \otimes z e_0 - v_\beta \otimes (u_\eta-z) e_0
 \right) \\
=& \frac12 \left(v_\beta u_\eta \otimes e_0 - v_\beta \otimes u_\eta e_0 \right)\,,
  \end{aligned}
\end{equation*}
which gives $\br{u_\eta^K,v_\beta}=0$. Similarly,  $\br{u_\eta^K,w_\beta}=0$. Gathering the expressions, we have proved 
\begin{lem}  \label{Lem:Dyz}
Write $u_\eta=z(1+\eta \phi)$ with $\phi=(\sum_s e_s+xy)(\sum_s e_s+yx)^{-1}$. 
  The left Loday bracket $\br{-,-}: A \times A \to A$ satisfies  for any $K\in m \N$ 
\begin{equation*}
\begin{aligned}
    \frac1K \br{u_\eta^K,x}=& -\eta \phi u_\eta^{K-1} zx  -x  u_\eta^{K-1} z\,, \quad 
\frac1K \br{u_\eta^K,z}= -z u_\eta^{K-1} z  + u_\eta^{K-1} z^2 \,, \\
\frac1K\br{u_\eta^K,v_\beta}=&0\,, \qquad \frac1K\br{u_\eta^K,w_\beta}=0\,.
\end{aligned}
\end{equation*}
\end{lem}

\medskip

We now adapt the discussion to the family $(H_k^m)_k$, which are the symmetric functions of the matrix representing the element $\bar{u}_\eta=y(1+\eta \phi)$. Note that, by doing so, we assume the invertibility of the element $x$ (which was used to make sense of $z^{-1}$ when deriving the double brackets $\dgal{u_\eta,-}$), but it can be proved by only inverting $y$ instead. Hence, passing from $u_\eta$ to $\bar{u}_\eta$, 
it is not hard to see that the double brackets involving $x$ only differ by an additional term  $-F_{-1}$ when replacing $\dgal{z,x}$ by $\dgal{y,x}$, so that 
\begin{equation}
\begin{aligned}
    \frac1K \br{\bar{u}_\eta^K,x}=&-(1+\eta \phi) \bar{u}_\eta^{K-1} -y^{-1}\bar{u}_\eta \bar{u}_\eta^{K-1} yx  -x \bar{u}_\eta^{K-1} y + \bar{u}_\eta^{K-1} yx \\
=&- \bar{u}_\eta^{K-1}    -x  \bar{u}_\eta^{K-1} y -\eta \phi \bar{u}_\eta^{K-1} (1+yx)\,.
\end{aligned}
\end{equation}
The double bracket with $y$ instead of $z$ does not change : $\frac1K \br{\bar{u}_\eta^K,y}= -y \bar{u}_\eta^{K-1} y  +  \bar{u}_\eta^{K-1} y^2$. 
The same holds for the couple $(y,w_\beta)$ replacing $(z,w_\beta)$, or doing it with $v_\beta$, so that 
$\br{\bar{u}_\eta^K,w_\beta}=0$ and $\br{\bar{u}_\eta^K,v_\beta}=0$. Therefore 
\begin{lem} \label{Lem:Dyy}
Write $\bar{u}_\eta=y(1+\eta \phi)$ with $\phi=(\sum_s e_s+xy)(\sum_s e_s+yx)^{-1}$. 
  The left Loday bracket $\br{-,-}: A \times A \to A$ satisfies for any $K\in m \N$ 
\begin{equation*}
\begin{aligned}
    \frac1K \br{\bar{u}_\eta^K,x}=&- \bar{u}_\eta^{K-1}   -x  u_\eta^{K-1} y -\eta \phi u_\eta^{K-1} (1+yx)\,, \quad 
\frac1K \br{\bar{u}_\eta^K,y}= -y u_\eta^{K-1} y  + u_\eta^{K-1} y^2 \,, \\
\frac1K\br{\bar{u}_\eta^K,v_\beta}=&0\,, \qquad \frac1K\br{\bar u_\eta^K,w_\beta}=0\,.
\end{aligned}
\end{equation*}
\end{lem}

\medskip

For the family $(F_k^m)_k$ of symmetric functions of the matrix representing the element $\tilde{u}_\eta=(\sum_s e_s + xy)(1+\eta \phi^{-1})$, we also do the computations assuming that $x$ is invertible so that $\sum_s e_s + xy=xz$ and $\tilde{u}_\eta=xz(1+\eta \phi)$. As a first intermediate result, note that 
\begin{equation*}
  \dgal{xz,x}=\frac12 x^2 z F_0 - \frac12 F_0 xzx - \frac12 xz F_0 x - \frac12 x F_0 xz\,.
\end{equation*}
This directly yields 
\begin{equation*}
  \begin{aligned}
    \dgal{\tilde{u}_\eta,x}=&\dgal{xz,x}\ast (1+\eta \phi^{-1}) -\eta xz \phi^{-1} \ast\dgal{\phi,x}\ast \phi^{-1} \\
=& \frac12 (x \tilde{u}_\eta F_0 - (xz)^{-1}\tilde{u}_\eta F_0 xzx - \tilde{u}_\eta F_0 x - x (xz)^{-1}\tilde{u}_\eta F_0 xz) \\
&+\frac12 (-x [(xz)^{-1}\tilde{u}_\eta-1]F_0 xz + [(xz)^{-1}\tilde{u}_\eta-1]F_0 xzx - x F_0 [\tilde{u}_\eta-xz] + F_0 [\tilde{u}_\eta-xz]x)\,.
  \end{aligned}
\end{equation*}
After simplification, we obtain for any $K \in \N$
\begin{equation*}
\begin{aligned}
 \frac1K \br{\tilde u_\eta^K,x}=& 
-x(xz)^{-1} \tilde u_\eta^K xz - \tilde u_\eta^{K-1} xzx + x \tilde u_\eta^{K-1} xz
= - \tilde u_\eta^{K-1} xzx - \eta\, x \phi^{-1}  \tilde u_\eta^{K-1} xz\,.
\end{aligned}
\end{equation*}
Second, we compute $\br{\tilde u_\eta^K,xz}$ using
\begin{equation*}
  \begin{aligned}
    \dgal{\tilde u_\eta,xz}=& \dgal{xz,xz}\ast (1+\eta \phi^{-1}) - \eta xz\phi^{-1} \ast\dgal{\phi,xz} \ast \phi^{-1} \\
=&\frac12 \left( xz \tilde u_\eta F_0 - (xz)^{-1} \tilde u_\eta F_0 (xz)^2 \right) \\
&+\frac12 \left(- [\tilde{u}_\eta-xz]F_0 xz + [(xz)^{-1}\tilde{u}_\eta-1] F_0 (xz)^2 
- xz F_0 [\tilde{u}_\eta-xz] + F_0 [\tilde{u}_\eta-xz]xz \right)\,.
  \end{aligned}
\end{equation*}
After cancellations, this is just 
\begin{equation*}
          \frac1K \br{\tilde u_\eta^K,xz}=  xz \tilde u_\eta^{K-1} xz - \tilde u_\eta^{K-1} (xz)^2 \,.
\end{equation*}
Finally,  we get $\dgal{\tilde u_\eta,v_\beta}=\frac12 v_\beta \tilde u_\eta \otimes e_0 - \frac12 v_\beta \otimes \tilde u_\eta e_0$, so that  $\br{\tilde u_\eta^K,v_\beta}=0$ as before. The same is true for $w_\beta$. 
\begin{lem} \label{Lem:Dy1xy}
Write $\tilde{u}_\eta=u(1+\eta \phi)$ with $\phi=(\sum_s e_s+xy)(\sum_s e_s+yx)^{-1}$ and $u=\sum_s e_s+xy$. 
  The left Loday bracket $\br{-,-}: A \times A \to A$ satisfies  for any $K\in \N$ 
\begin{equation*}
\begin{aligned}
    \frac1K \br{\tilde{u}_\eta^K,x}=&- \tilde u_\eta^{K-1} ux - \eta\, x \phi^{-1}  \tilde u_\eta^{K-1} u \,, \quad
\frac1K \br{\tilde{u}_\eta^K,u}= -\tilde u_\eta^{K-1} u^2  +u  \tilde u_\eta^{K-1} u  \,, \\
\frac1K\br{\tilde{u}_\eta^K,v_\beta}=&0\,, \qquad \frac1K\br{\tilde u_\eta^K,w_\beta}=0\,.
\end{aligned}
\end{equation*}
\end{lem}


\end{document}